\documentclass[letter,12pt]{amsart}

\usepackage{wrapfig}
\usepackage{amsfonts}
\usepackage{amssymb}
\usepackage{amsmath,tabu}
\usepackage{amsthm}
\usepackage{comment}
\usepackage{graphics}
\usepackage{epstopdf}
\usepackage{xcolor}
\usepackage[unicode]{hyperref}
\usepackage{mathtools}
\usepackage{subcaption}
\usepackage{multicol}
\usepackage{amssymb}
\usepackage{mathrsfs}
\usepackage{eufrak}
\usepackage{diagbox}

\usepackage{natbib}

\usepackage{verbatim}
\usepackage{eurosym}
\usepackage[utf8]{inputenc}
\usepackage[english]{babel}
\usepackage[margin=1in]{geometry}
\usepackage{apptools}
\usepackage{tabularx}
\usepackage{enumitem}
\usepackage{bbm}

\newtheorem{theorem}{Theorem}
\newtheorem{proposition}[theorem]{Proposition}

\newtheorem{procedure}[theorem]{Procedure}
\newtheorem*{prop}{Proposition}
\newtheorem{lemma}[theorem]{Lemma}

\newtheorem{assumption}[theorem]{Assumption}
\newtheorem{corollary}[theorem]{Corollary}
\newtheorem{remark}[theorem]{Remark}

\newtheorem{definition}[theorem]{Definition}

\numberwithin{theorem}{section}

\numberwithin{equation}{section}

\newcommand{\eps}{\varepsilon}

\newcommand{\U}{\mathbf U}
\newcommand{\A}{\mathbf A}
\newcommand{\B}{\mathbf B}

\newcommand{\V}{\mathbf V}

\renewcommand{\v}{\mathbf v}
\renewcommand{\u}{\mathbf u}
\newcommand{\w}{\mathbf w}
\newcommand{\x}{\mathbf x}

\newcommand{\s}{\mathfrak s}
\newcommand{\m}{\mathfrak m}

\renewcommand{\aa}{\mathfrak a}

\newcommand{\ii}{\mathbf i}
\newcommand{\dd}{\mathrm d}

\newcommand{\E}{\mathbb E}

\newcommand{\T}{\mathsf T}

\newcommand{\Var}{\mathrm{Var}}

\newcommand{\MH}{\Omega_-}

\newcommand\tomer{\overset{\text{mer}}{\rightarrow}}
\newcommand\tomerD{\overset{\text{mer, d}}{\rightarrow}}

\DeclareMathOperator{\Ai}{Ai}

\begin{document}
\title[Uniform inference for signal strength]{How weak are weak factors? Uniform inference for signal strength in signal plus noise models}

\author{Anna Bykhovskaya}
\address[Anna Bykhovskaya]{Duke University}
\email{anna.bykhovskaya@duke.edu}

\author{Vadim Gorin}
\address[Vadim Gorin]{University of California at Berkeley}
\email{vadicgor@gmail.com}

\author{Sasha Sodin}
\address[Sasha Sodin]{The Hebrew University of Jerusalem and Queen Mary University of London}
\email{alexander.sodin@mail.huji.ac.il}

\thanks{The authors thank seminar participants at Cambridge, Duke, online RMTA seminar, Princeton, and Simons Symposium for their helpful comments. Gorin’s work was partially supported by NSF grant DMS-2246449.
Sodin's work was partially supported by a Philip Leverhulme Prize (PLP-2020-064) and by the Morris Belkin Visiting Professor Program at the Weizmann Institute of Science.}




\date{\today}
\maketitle

\begin{abstract}
The paper analyzes four classical signal plus noise models: the factor model, spiked sample covariance matrices, the sum of a Wigner matrix and a low-rank perturbation, and canonical correlation analysis with low-rank dependencies. The objective is to construct confidence intervals for the signal strength that are uniformly valid across all regimes -- strong, weak, and critical signals. We demonstrate that traditional Gaussian approximations fail in the critical regime. Instead, we introduce a universal transitional distribution that enables valid inference across the entire spectrum of signal strengths. The approach is illustrated through applications in macroeconomics and finance.
\end{abstract}




\section{Introduction}

\subsection{Motivation}

In the modern era researchers increasingly have access to high-dimensional data  across a wide range of fields. These data are inevitably contaminated by various forms of error and noise, making the separation of meaningful structure from background noise a central challenge. To address this analysts commonly employ dimension-reduction techniques. The two dominant approaches are low-rank methods, which assume that the underlying signal lies in a lower-dimensional subspace, and sparsity-based methods, which assume that only a small subset of variables or parameters are truly relevant, i.e., nonzero. This paper adopts the low-rank perspective. For a discussion of settings where this assumption is appropriate, we refer to \citet{udell2019big}, \citet{giannone2021economic}, and \citet{thibeault2024low}. In particular, \citet{giannone2021economic} argue that numerous data sets in macroeconomics, microeconomics, and finance exhibit dense, rather than sparse, structures.


A prototypical example of a low rank setting is the factor model, where one observes an $N\times S$ data matrix $X$ and assumes that it can be decomposed as
\begin{equation}
\label{eq_factor_intro}
 X = L F^\T  + \mathcal E,
\end{equation}
where $F$ is an $S\times r$ matrix of factors, $L$ is an $N\times r$ matrix of factor loadings, and $LF^\T$ represents the low-rank signal of interest. The signal rank $r$ is small relative to the large dimensions $N$ and $S$. The remainder $\mathcal{E}$ is a noise matrix, often assumed to have i.i.d.~mean-zero entries in the simplest setting.

The feasibility of consistently estimating the signal component $LF^\T$ from the observed data $X$ hinges on the strength of the signal, which can be quantified by the singular values of $LF^\T$. When these singular values are large, the signal is strong and estimation is reliable, as can be directly predicted from the form of \eqref{eq_factor_intro}. As the signal weakens, the data $X$ becomes less informative, and below a certain critical threshold, accurate recovery becomes impossible. The relationships between the strength of the signal and feasibility of reconstruction of $LF^\T$ have been rigorously analyzed in a number of studies, see, e.g., \citet{stock2002forecasting,bai2002determining,bai2003inferential,paul2007asymptotics,onatski2012asymptotics,johnstone2018pca,bai2023approximate,fan2024can,barigozzi2024dynamic} and references therein.


Given this behavior, applied work using factor models should begin by assessing the strength of the factors, since the validity of any inference on $L$ or $F$ critically depends on it. However, in practice, this step is often overlooked\footnote{This pattern is evident in the vast majority of approximately 120 papers that employ factor models or PCA-related techniques, published in the five leading economics journals between 2015 and 2025.},
and most studies tacitly assume that  the factors are strong, without conducting any formal diagnostics.


This paper seeks to emphasize the importance of assessing signal strength in a broad class of ``signal plus noise'' models. To that end, we develop novel procedures for constructing confidence intervals for signal strength. Crucially, we do not assume that the signals are strong -- an assumption often unjustified in empirical applications. Instead, our analysis remains valid across the full range of regimes: strong, weak, and critical signals.


\subsection{Models and results}
We analyze four classical high-dimensional statistical models: the factor model (\ref{eq_factor_intro}), spiked sample covariance, the spiked Wigner model, and spiked canonical correlations. These correspond to three fundamental ensembles from random matrix theory: the Laguerre/Wishart ensemble for the first two, the Hermite/Gaussian/Wigner ensemble for the third, and the Jacobi ensemble for the fourth. Each model can be viewed as an instance of the signal plus noise framework, also known as spiked random matrices, a term originating with \citet{johnstone2001distribution}, in which a low-rank signal matrix is embedded in a high-dimensional noisy environment. The goal is to detect and quantify the signal.

The signal in each model can be decomposed into a sum of rank-one components. Each component is characterized by a positive scalar (its strength) and one or two unit-norm vectors (its direction), depending on the setup. In this work we focus solely on the signal strength and do not consider inference on directions.

Our analysis is based on spectral methods, whereby signal strength is inferred from the eigenvalues of certain model-specific matrices. In all four setups a well-documented phase transition phenomenon arises: the signal strength can be consistently estimated (in the high-dimensional asymptotic regime with proportional growth of data dimensions) only when it exceeds a critical threshold, see \citet{jones1978eigenvalue}, \cite{baik2006eigenvalues}, \citet{onatski2012asymptotics}, \citet{bao2019canonical} and more references in Section \ref{Section_signal_plus_noise_models}. When the signal strength falls below the threshold, only partial probabilistic information, such as asymptotics of the likelihood ratio test can be recovered, but reliable point estimation becomes impossible, see, e.g.\ \citet{onatski2013asymptotic,onatski2014signal,dobriban2017sharp,johnstone2020testing,el2020fundamental}.  The intermediate regime, where the signal strength is close to the threshold, is typically referred to as the ``critical'' regime. This regime is particularly challenging for inference.

In the super-critical case, where the strength is significantly above the threshold, the estimation procedure is quite straightforward: one takes the largest eigenvalue, applies to it a certain explicit function (see Section \ref{Section_signal_plus_noise_models} for the formulas) and gets the strength of the strongest signal. Repeating the same with the second, third, etc., eigenvalues one gets strengths of the further components of the signal and the only question is when to stop, i.e., after which step one should declare that the following signals are too weak and can not be recovered. There are many results in the literature proposing various algorithms to choose the stopping point. We further remark that for very strong signals the function one should apply to the eigenvalues is close to identity ($f(x)=x$), whereas for weaker signals the function exhibits stronger dependence on the model of interest.

Once point estimates of the signal strengths are obtained, the next natural question is how to quantify uncertainty -- specifically, how to construct confidence intervals for these estimates. The existing literature offers little guidance on this front -- particularly guidance that is consistent across models and signal strengths -- with most results focusing on strong signals.
The technical challenge is rooted in the nonstandard asymptotic behavior of the eigenvalues near the phase transition threshold. While the fluctuations of the top eigenvalues are asymptotically Gaussian for well-separated (super-critical) signals, the limiting distribution becomes highly non-Gaussian and analytically intricate as the signal strength approaches the critical boundary (see \citet{baik2005phase}, \citet{mo2012rank}, and \citet{bloemendal2013limits}  for rigorous results in the sample covariance setting).

Our paper fills this gap by proposing a general procedure for constructing confidence intervals for signal strength. Remarkably, across all four models we study, the confidence intervals are characterized by a common limiting (stochastic) object, which we call the \emph{Airy--Green function} and denote $\mathcal{G}(w)$. Our main contributions are: a rigorous construction of this function, a unified set of theorems linking it to the four canonical models, and tabulated confidence intervals based on $\mathcal{G}(w)$. The only model-specific components are a set of scaling constants, which we provide explicitly for each setting. In addition, our results imply a formula-free, bootstrap-type procedure for constructing confidence intervals.


\subsection{Econometrics and statistics contributions}

In economics and finance it has long been observed that many data sets contain factors that are either non-informative or far from strong -- see e.g.\ \citet{giglio2023prediction} and \citet{kim2024testing} for overviews and extensive references. This concern is especially apparent in the vast ``factor zoo'' of potential variables proposed to explain stock returns. This empirical reality has motivated a line of theoretical research focused on inference for weaker factors. Broadly speaking, factors can be classified by their strength into three categories: strong (as in, e.g.\ \citet{bai2002determining,stock2002forecasting}), semi-strong (as in, e.g.\ \citet{bai2023approximate,fan2024can}), and weak\footnote{What we call semi-strong factors are sometimes referred to as weak, while weak factors may be termed weakly influential or extremely weak.} (as in, e.g., \citet{onatski2012asymptotics}). The literature also includes statistical procedures for testing and distinguishing between these types of factors (see, in particular, \citet{kim2024testing}). Over the past decades a growing body of research has focused specifically on factor strength, including contributions by \citet{chudik2011weak}, \citet{bailey2016exponent}, \citet{wang2017asymptotics}, \citet{lettau2020estimating}, \citet{cai2020limiting}, \citet{bailey2021measurement}, \citet{freyaldenhoven2022factor}, \citet{uematsu2022estimation}, and \citet{pesaran2025identifying}.

In comparison to this literature, our main methodological contribution is a unified procedure for constructing confidence intervals for signal strength across all four models and all signal ranges, as presented in Section \ref{Section_confidence_intervals}. This approach does not rely on standard Gaussian quantiles, but instead uses a novel random \emph{transition process} $\mathcal{T}(\Theta)$, whose quantiles are tabulated in Table \ref{Table_T_quantiles}. Figure \ref{Fig_confidence_int} reveals that our procedure performs well for a wide range of different signals. The quality of approximations is excellent for critical and weak signals, and does not deteriorate for larger signals, hence, covering also the case of strong signals. In contrast, as shown in Figure \ref{Fig_confidence_int}, the Gaussian approximation performs poorly near the critical threshold, making $\mathcal{T}(\Theta)$ essential for accurate inference in that regime. Our approach is reminiscent of the construction of uniform confidence intervals for autoregressive models in \citet{stock1991confidence,mikusheva2007uniform}, where the non-standard asymptotics near the unit root are smoothly connected to the standard normal behavior in the stationary region.

A surprising finding is that the same transition process $\mathcal{T}(\Theta)$ governs all four models. In fact, the proofs in Section \ref{Section_proof_of asymptotic_approximations} follow different paths depending on the model, and only in the final step does a structural identity emerge, revealing that all four asymptotic distributions coincide. Random matrix theory has many universality theorems, and based on our results, we predict that the same transition process $\mathcal{T}(\Theta)$ governs a much wider class of signal plus noise models, beyond the ones analyzed here.

Beyond quantifying uncertainty in signal strength, our framework also enables signal detection and the assessment of factor informativeness. Specifically, one can check whether the uniform confidence intervals include zero and the identification threshold, respectively. \citet{paul2007asymptotics,onatski2012asymptotics,benaych2012singular,BG_CCA} show that, for weak signals, estimates of the signal direction are inconsistent. Asymptotically, the estimated direction is inclined at an angle $\phi$ relative to the true direction. There are two complementary cases: if $\phi= \frac{\pi}{2}$, then the estimated direction contains no information about the truth and can be discarded; if $\phi<\frac{\pi}{2}$, then information is present and can potentially be extracted. It turns out that $\phi$ depends on the signal strength, decreasing as the strength increases, and that the transition between these two cases occurs precisely at the identification threshold discussed above. This reinforces our results on estimating signal strength in the critical regime as a tool for distinguishing between these two cases in direction estimation.

\subsection{Mathematical contributions}

From a mathematical perspective, we develop a new approach to analyzing critical spikes, grounded in perturbation theory equations that relate the eigenvalues of spiked and unspiked random matrices. This contrasts with earlier treatments of critical spikes in real symmetric matrices, which relied on Pfaffian point processes (as in \citet{mo2012rank}) or on tridiagonal matrix models (as in \citet{bloemendal2013limits, bloemendal2016limits, lamarre2019edge}). Our central technical contribution is to show that these perturbation equations admit a well-defined edge-scaling limit, which captures the asymptotic behavior of the largest eigenvalues. While our approach is novel in all four settings, we particularly emphasize the fourth -- canonical correlations -- where no prior results on critical spikes were available.

In Section \ref{Section_asymptotics} and \ref{Section_asymptotics_proofs} we establish this edge limit result under two key assumptions on the unspiked model: (i) the asymptotics of the largest eigenvalues converge to the Airy$_1$ point process, and (ii) a form of the local law holds for the Stieltjes transform near the spectral edge. These assumptions are known to hold for a wide range of random matrix ensembles, including the four models considered in this paper. A notable strength of our approach is its minimal reliance on model-specific structure: we require only the two inputs above.

We build on some of the ideas in \citet{aizenman2015ubiquity}. In contrast, however, we focus on the limit at the spectral edge—rather than in the bulk—which requires subtracting diverging counterterms. Moreover, we establish convergence in a stronger topology, which allows us to work directly on the real axis; see Appendix \ref{Section_asymptotics_proofs} for further details.

\subsection{Outline of the paper} Section \ref{Section_signal_plus_noise_models} introduces the four main signal plus noise models. Section \ref{Section_confidence_intervals} presents a unified procedure for constructing confidence intervals for signal strengths. Section \ref{Section_asymptotics} lays out the theoretical foundations underlying this procedure. Section \ref{Section_emprical} offers three empirical illustrations. Extensions are discussed in Section \ref{Section_extensions}. Section \ref{Section_conclusion} concludes. All proofs are in Appendices \ref{Section_asymptotics_proofs} and \ref{Section_proof_of asymptotic_approximations}.

\section{Four signal plus noise models} \label{Section_signal_plus_noise_models}

In this section we present the four models, beginning with the simplest case -- the spiked Wigner model -- then proceeding to sample covariance and factor models based on PCA, and concluding with canonical correlation analysis (CCA). Although PCA-based models are the most widely used in practice, we adopt this order because the formulas are simpler in the Wigner case, making the key ideas more transparent.

\subsection{Spiked Wigner matrix} \label{Section_spiked_Wigner}

Suppose  we observe an $N\times N$ matrix $\A$ of the form
\begin{equation}
\label{eq_Spiked_Wigner}
 \A= \sum_{i=1}^r \theta_i \cdot \u^*_i (\u^*_i)^\T + \mathcal E,
\end{equation}
where $r$ is  fixed  (not growing with $N$) and $\theta_1>\dots>\theta_r>0\in\mathbb R$ are the strengths of $r$ signals, with corresponding directions $\u^*_1,\dots,\u^*_r$, which are assumed to be orthonormal $N$--dimensional vectors. The noise matrix $\mathcal E$ is a (Wigner) matrix sampled from the Gaussian Orthogonal Ensemble, meaning that $\mathcal E=\frac{1}{\sqrt{2 N}}(\mathcal Z+\mathcal Z^\T)$, where $\mathcal Z$ is an $N\times N$ matrix of i.i.d.~$\mathcal N(0,\sigma^2)$ entries (see Section \ref{Section_extension_non_Gauss} for non-Gaussian setting). We assume that $\theta_i$ and $\u^*_i$ are unknown deterministic parameters; one could alternatively allow $\u^*_i$ to be random, provided they are independent of $\mathcal E$. Our goal is to estimate the signal strengths $\theta_1,\dots,\theta_r$.

We first assume that the variance of the underlying noise $\mathcal Z$, $\sigma^2$, is known and set it to $1$ by rescaling the model.\footnote{The prefactor $\frac{1}{\sqrt{2N}}$ in the definition of $\mathcal E$ ensures that its eigenvalues fill the interval $[-2,2]$ as $N\to\infty$.} In Section \ref{Section_unknown_variance} we discuss adjustments for the case of unknown $\sigma^2$.

One common application of the spiked Wigner framework is modeling symmetric interaction networks, such as economic or social activity among $N$ agents. Each rank‑one component $\theta_i \u^*_i (\u^*_i)^\top$ captures a latent structure in agent attributes $\u^*_i$, while the observed interactions are contaminated by noise $\mathcal E$. Low-rank approximations of this form underpin seminal network models including the stochastic block model of \cite{Holland1983stochastic}, where communities are inferred from block‑structured adjacency matrices, and latent space models.


\medskip

The following result establishes the threshold for the estimation of $\theta_i$ via spectral methods.
\begin{prop}[\citet{jones1978eigenvalue,furedi1981eigenvalues,capitaine2009largest,capitaine2012central}] 
Suppose that all $\theta_i$ are distinct and ordered $\theta_1>\theta_2>\dots>\theta_r$, $\sigma^2=1$. Let $\lambda_1\ge \lambda_2\dots\ge \lambda_N$ denote the eigenvalues of $\A$ sampled from \eqref{eq_Spiked_Wigner} with $\sigma^2=1$.  Denote
 \begin{equation}
\label{eq_Spiked_Wigner_params}
 \theta^c=1,\qquad \lambda_+=2,\qquad
 \lambda(\theta)=\theta+\frac{1}{\theta},\qquad V(\theta)=2\, \frac{\theta^2-1}{\theta^2}.
\end{equation}
 For each $1\le i \le r$, if $\theta_i>\theta^c$, then as $N\to\infty$, in the sense of convergence in distribution
 \begin{equation} \label{eq_spiked_Wigner_Gaussian}
  \lambda_i = \lambda(\theta_i) + \frac{1}{\sqrt{N}} \mathcal N\bigl(0, V(\theta_i)\bigr) + o\left(\frac{1}{\sqrt{N}}\right),
 \end{equation}
 and the Gaussian limits $ \mathcal N\bigl(0, V(\theta_i)\bigr)$ are independent over $i$. If $\theta_i \le \theta^c$, then ${\lim_{N\to\infty} \lambda_i=\lambda_+}$, in probability.
\end{prop}

Informally, the proposition says that ``good'' recovery of $\theta_i$ from the largest eigenvalues is possible if and only if $\theta_i$ is larger than the critical value $\theta^c=1$. In this case, to estimate $\theta_i$, one should take $\lambda_i$ and apply the inverse of the mapping $\theta\mapsto\lambda(\theta)$, which is $\lambda\mapsto \frac{1}{2}\left(\lambda+\sqrt{\lambda^2-4}\right)$.

We assess the quality of estimating $\theta_i$ by constructing a confidence interval for it. Specifically, for each fixed $i$ and significance level $\alpha$ we aim to find endpoints $\theta^-_i(\lambda_i,N,\alpha),\, \theta^+_i(\lambda_i,N,\alpha)$ such that
\begin{equation}
\label{eq_confidence_interval_1}
 \mathrm{Prob} \left( \theta_i \in [\theta^-_i(\lambda_i,N,\alpha), \theta^+_i(\lambda_i,N,\alpha)]\right)\approx 1-\alpha,
\end{equation}
where $\approx$ denotes an $N\to\infty$ approximation, which should be uniform over the model parameters $\theta_1,\dots,\theta_r$ and $\u_1^*,\dots,\u_r^*$ in \eqref{eq_Spiked_Wigner}.

In principle, since we deal with multiple $\theta_i$ simultaneously, one could consider joint multi-dimensional confidence sets. However, due to the asymptotic independence of $\lambda_i$ in \eqref{eq_spiked_Wigner_Gaussian}, it is sufficient to construct separate intervals for each $\theta_i$, which is the approach we take.\footnote{In contrast, if $\theta_i$ coincide, then the limits in \eqref{eq_spiked_Wigner_Gaussian} are neither Gaussian nor independent, cf.\ \citet[Theorem 3.3]{capitaine2012central}.}

The asymptotics \eqref{eq_spiked_Wigner_Gaussian} provides a way to construct confidence intervals by approximating $\theta_i$ in the argument of $V(\theta_i)$ with $\theta(\lambda_i)=  \frac{1}{2}\left(\lambda_i+\sqrt{\lambda^2_i-4}\right)$ and then using Gaussian quantiles. This leads to the following formula for the confidence interval:
\begin{equation}
\label{eq_Wigner_CI}
 \theta_i \in  \left[ \frac{\lambda_i}{2} +\sqrt{\frac{\lambda_i^2}{4}-1}   - \frac{z_{\alpha/2}}{\sqrt{N}} \sqrt{1+\frac{\lambda_i }{\sqrt{\lambda_i^2-4}}},\,  \frac{\lambda_i}{2} +\sqrt{\frac{\lambda_i^2}{4}-1} + \frac{z_{\alpha/2}}{\sqrt{N}} \sqrt{1+\frac{\lambda_i}{\sqrt{\lambda_i^2-4}}}\right],
\end{equation}
where $z_{\alpha/2}$ denotes the $\alpha/2$ quantile of $\mathcal N(0,1)$. E.g., to obtain a $95\%$ confidence interval for a single fixed $i$, we set $z_{\alpha/2}=1.96$.

\begin{figure}[t]
    \centering
    \begin{subfigure}{.49\textwidth}
            \includegraphics[width=\textwidth]{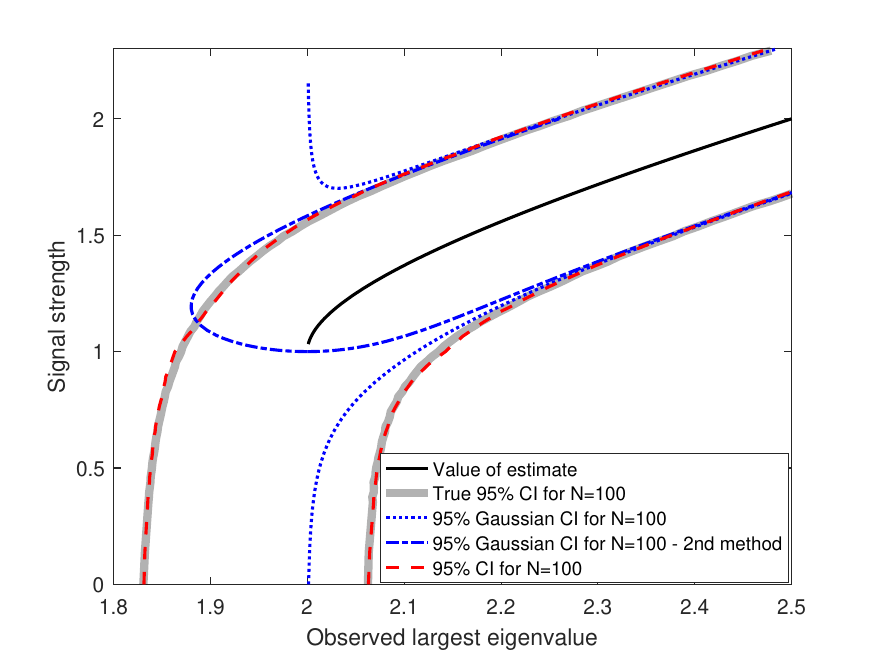}
            \caption{Spiked Wigner matrix of Section \ref{Section_spiked_Wigner}\label{Fig_confidence_Wigner}}
    \end{subfigure}
    \begin{subfigure}{.49\textwidth}
            \includegraphics[width=\textwidth]{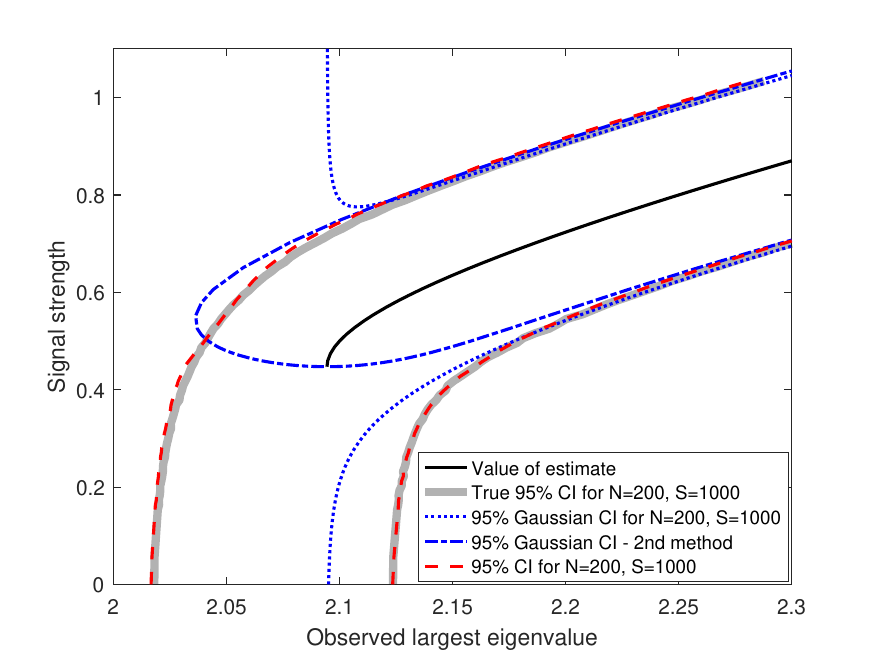}
            \caption{Factor model of Section \ref{Section_factor}\label{Fig_confidence_PCA}}
    \end{subfigure}
    \caption{Confidence intervals for $\theta$ as functions of the observed largest eigenvalue via Gaussian approximations and via our procedure of Section \ref{Section_confidence_intervals}.\label{Fig_confidence_int}}
\end{figure}

The formula \eqref{eq_Wigner_CI} reveals a problem as $\theta\to 1$ (i.e., $\lambda\to 2$): the confidence intervals diverge due to the $\sqrt{\lambda_i^2-4}$ singularity in the denominator. However, Monte Carlo simulations in Figure \ref{Fig_confidence_Wigner} indicate that no such explosion actually occurs. This suggests that the approximation error in the confidence interval \eqref{eq_Wigner_CI} becomes non-negligible when $\lambda_i$ is close to 2, making the formula unreliable in this regime. In contrast, our novel procedure, introduced in Section \ref{Section_confidence_intervals}, closely matches the simulations across all values of $\lambda_i$.

\begin{remark} \label{Remark_Second_Gausssian_CI}
An alternative way to construct confidence intervals using Gaussian asymptotics is to rewrite \eqref{eq_spiked_Wigner_Gaussian} in the equivalent form
 $$
   \lambda_i \in \left[\theta_i+\frac{1}{\theta_i} - \frac{z_{\alpha/2}}{\sqrt{N}} \sqrt{2  \frac{\theta_i^2-1}{\theta_i^2}} + o\left(\frac1{\sqrt N}\right), \quad \theta_i+\frac{1}{\theta_i} + \frac{z_{\alpha/2}}{\sqrt{N}} \sqrt{2  \frac{\theta_i^2-1}{\theta_i^2}}  + o\left(\frac1{\sqrt N}\right)\right].
 $$
We drop $o\left(\frac{1}{\sqrt{N}}\right)$ terms, plot the intervals from the preceding formula on the $(\theta,\lambda)$--plane, and then transpose the axes to obtain the desired confidence intervals on the $(\lambda,\theta)$--plane; see the 2nd method in Figure \ref{Fig_confidence_Wigner}. For $\theta$ away from $1$ (equivalently, $\lambda$ bounded away from $2$), this procedure is equivalent to the intervals \eqref{eq_Wigner_CI} as $N\to\infty$, though their finite-sample behavior differs near the cutoff. Both methods exhibit substantial bias, but in different directions.
\end{remark}

\subsection{Spiked covariance model} \label{Section_spiked_covariance} Second, we consider a deterministic $N\times N$ matrix
\begin{equation}
\label{eq_Covariance_spiked}
 \Omega=\sigma^2 I_N +  \sum_{i=1}^r (\theta_i-\sigma^2) \cdot \u^*_i (\u^*_i)^\T
\end{equation}
where $r$ is fixed, $\theta_1>\dots>\theta_r>\sigma^2$ are the signal strengths, and $\u^*_1,\dots,\u^*_r$ are orthonormal $N$-dimensional vectors representing $r$ signal directions. The eigenvalues of $\Omega$ are $\theta_1,\theta_2,\dots,\theta_r$,  and $\sigma^2$ with multiplicity $(N-r)$. As before, we assume that $\sigma^2$ is known and set it to $1$ without loss of generality; adjustments for unknown $\sigma^2$ are discussed in Section \ref{Section_unknown_variance}.

We observe an $N\times S$ data matrix $X$, whose columns are i.i.d.~$\mathcal{N}(0,\Omega)$, and aim to estimate $\theta_1,\dots,\theta_r$ from the sample covariance matrix $\frac{1}{S} X X^\T$. This model has been central in statistics and random matrix theory since \citet{johnstone2001distribution}; see \citet{johnstone2018pca} for a comprehensive overview, historical context, and many practical examples. Typically, the dimension $S$ reflects multiple independent observations: across individuals, measurement points, time periods, etc. An exact analogue of \eqref{eq_spiked_Wigner_Gaussian} holds in this setting as well.

\begin{prop}[\citet{baik2005phase,baik2006eigenvalues,paul2007asymptotics,bai2008central}]
Suppose that $\sigma^2=1$ and $\theta_1>\theta_2>\dots>\theta_r$ in \eqref{eq_Covariance_spiked}. Let $\lambda_1\ge \lambda_2\dots\ge \lambda_N$ denote the eigenvalues of  $\frac{1}{S} X X^\T$ in \eqref{eq_Covariance_spiked}. Assume\footnote{The case $\gamma>1$ can be also covered by similar methods.}:
\begin{equation}
\label{eq_limit_regime_1}
 \frac{N}{S}=\gamma^2+O\left(\frac{1}{N}\right),\quad N\to\infty, \qquad \gamma\in (0,1].
\end{equation}
 Denote
\begin{equation}
\label{eq_Spiked_covariance_params}
 \theta^c=1+\gamma,\qquad\!\! \lambda_+=(1+\gamma)^2,\qquad\!\!
 \lambda(\theta)=\theta+ \frac{\gamma^2\theta}{\theta-1},\qquad\!\! V(\theta)=2 \theta^2 \gamma^2 \left(1 - \frac{\gamma^2}{(\theta-1)^2}\right).
\end{equation}
 For each $1\le i \le r$, if $\theta_i>\theta^c$, then as $N\to\infty$, in the sense of convergence in distribution
 \begin{equation} \label{eq_spiked_covariance_as}
  \lambda_i = \lambda(\theta_i) + \frac{1}{\sqrt{N}} \mathcal N\bigl(0, V(\theta_i)\bigr) + o\left(\frac{1}{\sqrt{N}}\right),
 \end{equation}
 and the limits are independent over $i$. If $\theta_i \le \theta^c$, then $\lim_{N\to\infty} \lambda_i=\lambda_+$, in probability.
\end{prop}
As in the previous section, we can use this Gaussian approximation to construct confidence intervals for each $\theta_i$, yielding a modification of \eqref{eq_Wigner_CI}. However, this approach faces the same issue: the intervals become unreliable as $\lambda_i$ approaches $\lambda_+$ and must be corrected.

\subsection{Factor model} \label{Section_factor} For the third setup we consider a random $N\times S$ matrix $X$ defined by
\begin{equation}
\label{eq_factor_model}
 X = \sum_{i=1}^r \sqrt{\theta_i S} \cdot \u^*_i (\v^*_i)^\T + \mathcal E,
\end{equation}
where $r$ is a fixed small number, $\theta_1,\dots,\theta_r>0$ are the signal strengths, $\u^*_1,\dots,\u^*_r$ are $N$-dimensional orthonormal vectors of signal directions, called ``loadings''\footnote{Sometimes $\{\sqrt{\theta_i}\u^*_i\}$ rather than $\{\u^*_i\}$ are referred to as loadings.}, and $\v^*_1,\dots,\v^*_r$ are $S$-dimensional orthonormal vectors  called ``factors''. The noise matrix $\mathcal E$ has independent $\mathcal N(0,\sigma^2)$ entries. For now and until  Section \ref{Section_unknown_variance} we assume $\sigma^2$ to be known and set it to $1$.

Our goal is to estimate $\theta_1,\dots,\theta_r$ from the eigenvalues $\lambda_1\ge \lambda_2\ge \dots\ge \lambda_N$ of the sample covariance matrix $\frac{1}{S} X X^\T$. While the factor model has similarities to the spiked covariance model of the previous section, they are not equivalent, because we treat $\sqrt{\theta_i} \cdot \u^*_i (\v^*_i)^\T$ in \eqref{eq_factor_model} as deterministic parameters (the models would have been equivalent up to shift $\theta_i\to \theta_i+\sigma^2$, if each $\sqrt{S}\v^*_i$ were a mean $0$ Gaussian vector with i.i.d.\ components). This distinction allows the factor model to capture complex structures along the $S$-dimension, which is essential in applications across finance, macroeconomics, natural sciences, and other fields. Once again, an analogue of \eqref{eq_spiked_Wigner_Gaussian} holds.
\begin{prop}[\citet{onatski2012asymptotics,benaych2012singular}, {\citet[Theorem 5]{onatski2018asymptotics}}] 
Suppose that $\sigma^2=1$ and $\theta_1>\theta_2>\dots>\theta_r$ in \eqref{eq_factor_model}. Let $\lambda_1\ge \lambda_2\dots\ge \lambda_N$ denote the eigenvalues of  $\frac{1}{S} X X^\T$. Assume\footnote{Swapping the roles of $N$ and $S$ we also cover the case $\gamma>1$.}:
\begin{equation}
\label{eq_limit_regime_2}
 \frac{N}{S}=\gamma^2+O\left(\frac{1}{N}\right),\quad N\to\infty, \qquad \gamma\in (0,1].
\end{equation}
 Denote
\begin{equation}
\label{eq_Factor_params}
 \theta^c=\gamma,\quad\! \lambda_+=(1+\gamma)^2,\quad\!
 \lambda(\theta)=(\theta+1)(1+\frac{\gamma^2}{\theta}),\quad\! V(\theta)=2 \gamma^2 \frac{(2 \theta+1+\gamma^2)(\theta^2-\gamma^2)}{\theta^2}.
\end{equation}
 For each $1\le i \le r$, if $\theta_i>\theta^c$, then as $N\to\infty$, in the sense of convergence in distribution
 \begin{equation} \label{eq_factors_as}
  \lambda_i = \lambda(\theta_i) + \frac{1}{\sqrt{N}} \mathcal N\bigl(0, V(\theta_i)\bigr) + o\left(\frac{1}{\sqrt{N}}\right),
 \end{equation}
 and the limits are independent over $i$. If $\theta_i \le \theta^c$, then $\lim_{N\to\infty} \lambda_i=\lambda_+$, in probability.
\end{prop}
As in Section \ref{Section_spiked_Wigner}, the Gaussian approximation of $\lambda_i$ leads to two methods for constructing confidence intervals. An analogue of \eqref{eq_Wigner_CI} is
$$
\theta_i\in \left[\theta(\lambda_i) - \frac{\sigma(\lambda_i)}{\sqrt{N}} z_{\alpha/2},
 \theta(\lambda_i) + \frac{\sigma(\lambda_i)}{\sqrt{N}} z_{\alpha/2}\right], \qquad \text{ where }
$$
$$
 \theta(\lambda)=\frac{\lambda-1-\gamma^2+ \sqrt{(1+\gamma^2-\lambda)^2-4\gamma^2}}{2}, \quad
 \sigma(\lambda)= \frac{\sqrt{2 \gamma^2 (2 \theta(\lambda)+1+\gamma^2)(\theta(\lambda)^2-\gamma^2)}}{\sqrt{(1+\gamma^2-\lambda)^2-4\gamma^2}}.
$$
There is also a direct analogue of the second Gaussian method described in Remark \ref{Remark_Second_Gausssian_CI}. Figure \ref{Fig_confidence_PCA} compares these two Gaussian-based intervals with our new approach, presented in Section \ref{Section_confidence_intervals}. The comparison reveals the same key features as in the spiked Wigner model.

\subsection{Canonical correlation analysis} \label{Section_spiked_CCA}
For the final setup we fix a small integer $r$ and parameters $1\ge \theta_1,\dots,\theta_r\ge 0$. We consider a deterministic symmetric positive-definite $(N+M)\times (N+M)$ matrix $\Omega$ that satisfies
\begin{equation} \label{eq_CCA_model}
 \begin{pmatrix} A & 0_{N\times M}\\ 0_{M\times N} & B \end{pmatrix} \Omega \begin{pmatrix} A^\T & 0_{N\times M}\\ 0_{M\times N} & B^\T \end{pmatrix}\!=\! \begin{pmatrix} I_{N} & \mathrm{diag}(\sqrt{\theta_1},\dots,\sqrt{\theta_r})\\ \mathrm{diag}(\sqrt{\theta_1},\dots,\sqrt{\theta_r}) & I_{M} \end{pmatrix}\!,
\end{equation}
where $A$ and $B$ are $N\times N$ and $M\times M$ matrices, respectively, $I_N$ and $I_M$ are identity matrices of $N\times N$ and $M\times M$ dimensions, respectively, and $\mathrm{diag}(\sqrt{\theta_1},\dots,\sqrt{\theta_r})$ is a rectangular matrix with $\sqrt{\theta_1},\dots,\sqrt{\theta_r}$ on the first $r$ elements of the main diagonal and $0$ everywhere else.

Let $\x$ be an $(N+M)$--dimensional Gaussian mean $0$ random vector with covariance $\Omega$, and let $\u$ and $\v$  denote its first $N$ and last $M$ coordinates, respectively. The parameters $\theta_1,\dots,\theta_r$ are the squared canonical correlations between $\u$ and $\v$; see \citet{BG_review}, as well as classical statistics references such as \citet{thompson1984canonical,gittins1985canonical,anderson1958introduction,muirhead2009aspects} for detailed introductions to canonical correlation analysis (CCA). Algorithmically, $\theta_i$ are the largest eigenvalues of the matrix $(\E \u \u^\T)^{-1} \E \u \v^\T (\E \v\v^\T)^{-1} \E \v \u^{\T}$.

Given $S$ independent samples of $\x$, we construct two matrices: the $N\times S$ matrix $\U$ has $S$ samples of $\u$ as columns and the $M\times S$ matrix $\V$ has $S$ samples of $\v$ as columns. The sample squared canonical correlations $\lambda_1\ge \lambda_2\ge \dots$ are the eigenvalues of the $N\times N$ matrix $(\U \U^\T)^{-1} \U \V^\T (\V \V^\T)^{-1} \V \U^\T$. Our goal is to estimate $\theta_1,\dots,\theta_r$ from these eigenvalues.

In typical applications CCA is used to explore dependencies between two data sets, for example, two sets of individual characteristics, brain measurements versus behavioral scores, or two groups of stocks. The parameter $\theta_i$ quantify the strength of these dependencies. Once again, an analogue of \eqref{eq_spiked_Wigner_Gaussian} holds.

\begin{prop}[\citet{bao2019canonical,yang2022limiting,bai2022limiting,hou2023spiked,BG_CCA}]
Suppose $\theta_1>\theta_2>\dots>\theta_r$ in \eqref{eq_CCA_model}. Let $\lambda_1\ge \lambda_2\dots\ge \lambda_N$ denote the sample squared canonical correlations. Assume
\begin{equation}
\label{eq_limit_regime_3}
\frac{S}{N}=\tau_N+O\left(\frac{1}{N}\right), \quad \frac{S}{M}=\tau_M+O\left(\frac{1}{N}\right),\quad N\to\infty, \qquad \tau_N,\tau_M>1, \quad \tau_N^{-1}+\tau_M^{-1}<1.
\end{equation}
 Denote
\begin{equation}
\label{eq_CCA_params}\begin{split}
 \theta^c&=\frac{1}{\sqrt{(\tau_M-1)(\tau_N-1)}},\qquad \lambda_+=\left(\sqrt{\tau_M^{-1}(1-\tau_N^{-1})}+ \sqrt{\tau_N^{-1}(1-\tau_M^{-1})}  \right)^2,\\
 \lambda(\theta)&=\frac{\bigl(  (\tau_N-1)\theta  + 1 \bigr) \bigl(  (\tau_M-1) \theta + 1\bigr)}{\theta \tau_N \tau_M },
 \\  V(\theta)&=2 \frac{(1-\theta)^2}{\theta^2 \tau_M^2\tau_N^3} \bigl(2(\tau_M-1)(\tau_N-1)\theta+\tau_M+\tau_N-2\bigr)\bigl(
 (\tau_M-1)(\tau_N-1)\theta^2 -1\bigr).
\end{split}\end{equation}
 For each $1\le i \le r$, if $\theta_i>\theta^c$, then as $N\to\infty$, in the sense of convergence in distribution
 \begin{equation} \label{eq_CCA_as}
  \lambda_i = \lambda(\theta_i) + \frac{1}{\sqrt{N}} \mathcal N\bigl(0, V(\theta_i)\bigr) + o\left(\frac{1}{\sqrt{N}}\right),
 \end{equation}
 and the limits are independent over $i$. If $\theta_i \le \theta^c$, then $\lim_{N\to\infty} \lambda_i=\lambda_+$, in probability.
\end{prop}
\begin{remark}
The choice of $\frac{1}{\sqrt{N}}$ normalization introduces an asymmetry between $M$ and $N$ in the expression for the variance $V(\theta)$ in \eqref{eq_CCA_params}.
\end{remark}
The same conclusion applies here: using \eqref{eq_CCA_params} and Gaussian quantiles we can construct confidence intervals for $\theta_i$ that perform well when $\lambda_i$ is bounded away from $\lambda_+$, but become inaccurate as $\lambda_i$ approaches $\lambda_+$ and, therefore, require correction.

\section{Construction of confidence intervals}
\label{Section_confidence_intervals}

In this section we present our algorithm for constructing confidence intervals and explain how they can be interpreted and used to distinguish between noise, non-informative signals, and meaningful signals. We begin by introducing the transition process $\mathcal T(\Theta)$ and its properties, and then show how to use it to construct confidence intervals. We present two approaches: The first one is based on pre-tabulated quantiles of $\mathcal T(\Theta)$. The second one relies on bootstrap-type methodology. The underlying theorems will be presented in Section \ref{Section_asymptotics}.

\subsection{Transition process}
\label{Section_transition_function}

As highlighted in Figure \ref{Fig_confidence_int}, the Gaussian limits in \eqref{eq_spiked_Wigner_Gaussian}, \eqref{eq_spiked_covariance_as}, \eqref{eq_factors_as}, and \eqref{eq_CCA_as} ought to be replaced by a different limiting object, which we call the \emph{transition process} $\mathcal T(\Theta)$. This is a random function of $\Theta \in \mathbb{R}$. Its formal definition is provided in Section \ref{Section_Def_Gw}, while for the purposes of constructing confidence intervals, the key quantities of interest are the quantiles of its distribution, which may be computed as follows:

\begin{table}
\linespread{1}
\small
\makebox[\textwidth][c]{
\begin{tabular}{|r|r|r|r|r|r|r|r|}
\hline
\diagbox[width=0.9cm, height=0.75cm]{$\Theta$}{$\alpha$}  & .005 & .025 & .05 & .5 & .95 & .975 & .995 \\
\hline
\hline
-3.0 & -3.85 & -3.22 & -2.89 & -0.96 & 1.32 & 1.80 & 2.78 \\
-2.9 & -3.85 & -3.22 & -2.88 & -0.95 & 1.32 & 1.80 & 2.79 \\
-2.8 & -3.84 & -3.20 & -2.86 & -0.94 & 1.33 & 1.81 & 2.80 \\
-2.7 & -3.83 & -3.20 & -2.86 & -0.93 & 1.35 & 1.83 & 2.82 \\
-2.6 & -3.82 & -3.19 & -2.85 & -0.92 & 1.36 & 1.85 & 2.83 \\
-2.5 & -3.82 & -3.18 & -2.84 & -0.91 & 1.38 & 1.86 & 2.83 \\
-2.4 & -3.80 & -3.17 & -2.83 & -0.89 & 1.38 & 1.87 & 2.85 \\
-2.3 & -3.80 & -3.16 & -2.82 & -0.89 & 1.42 & 1.91 & 2.90 \\
-2.2 & -3.79 & -3.15 & -2.82 & -0.87 & 1.42 & 1.90 & 2.86 \\
-2.1 & -3.77 & -3.13 & -2.79 & -0.85 & 1.44 & 1.93 & 2.94 \\
-2.0 & -3.75 & -3.12 & -2.78 & -0.84 & 1.46 & 1.95 & 2.97 \\
-1.9 & -3.75 & -3.11 & -2.77 & -0.82 & 1.49 & 1.98 & 2.98 \\
-1.8 & -3.74 & -3.10 & -2.75 & -0.80 & 1.52 & 2.01 & 3.03 \\
-1.7 & -3.73 & -3.09 & -2.74 & -0.78 & 1.54 & 2.03 & 3.05 \\
-1.6 & -3.71 & -3.06 & -2.72 & -0.76 & 1.57 & 2.07 & 3.11 \\
-1.5 & -3.69 & -3.05 & -2.70 & -0.74 & 1.61 & 2.11 & 3.14 \\
-1.4 & -3.67 & -3.03 & -2.69 & -0.71 & 1.64 & 2.14 & 3.19 \\
-1.3 & -3.65 & -3.01 & -2.66 & -0.69 & 1.68 & 2.19 & 3.26 \\
-1.2 & -3.64 & -2.99 & -2.65 & -0.66 & 1.72 & 2.23 & 3.27 \\
-1.1 & -3.61 & -2.97 & -2.63 & -0.63 & 1.77 & 2.29 & 3.37 \\
-1.0 & -3.58 & -2.94 & -2.59 & -0.59 & 1.83 & 2.35 & 3.44 \\
-0.9 & -3.57 & -2.91 & -2.56 & -0.56 & 1.88 & 2.41 & 3.52 \\
-0.8 & -3.55 & -2.89 & -2.54 & -0.53 & 1.94 & 2.48 & 3.62 \\
-0.7 & -3.53 & -2.87 & -2.51 & -0.48 & 2.01 & 2.57 & 3.70 \\
-0.6 & -3.51 & -2.84 & -2.49 & -0.44 & 2.08 & 2.64 & 3.79 \\
-0.5 & -3.49 & -2.82 & -2.46 & -0.39 & 2.17 & 2.74 & 3.93 \\
-0.4 & -3.45 & -2.77 & -2.42 & -0.35 & 2.27 & 2.85 & 4.10 \\
-0.3 & -3.41 & -2.74 & -2.38 & -0.28 & 2.38 & 2.97 & 4.25 \\
-0.2 & -3.37 & -2.70 & -2.33 & -0.22 & 2.48 & 3.09 & 4.39 \\
-0.1 & -3.35 & -2.66 & -2.30 & -0.15 & 2.61 & 3.24 & 4.59 \\
0.0 & -3.31 & -2.62 & -2.25 & -0.08 & 2.75 & 3.40 & 4.76 \\
0.1 & -3.26 & -2.57 & -2.20 & 0.00 & 2.91 & 3.57 & 4.96 \\
0.2 & -3.22 & -2.52 & -2.14 & 0.09 & 3.08 & 3.77 & 5.24 \\
0.3 & -3.17 & -2.47 & -2.09 & 0.19 & 3.25 & 3.95 & 5.45 \\
0.4 & -3.13 & -2.40 & -2.02 & 0.29 & 3.46 & 4.19 & 5.72 \\
0.5 & -3.08 & -2.34 & -1.96 & 0.41 & 3.67 & 4.42 & 5.99 \\
0.6 & -3.01 & -2.28 & -1.89 & 0.54 & 3.90 & 4.68 & 6.29 \\
0.7 & -2.96 & -2.21 & -1.80 & 0.69 & 4.17 & 4.96 & 6.60 \\
0.8 & -2.90 & -2.14 & -1.73 & 0.83 & 4.41 & 5.23 & 6.91 \\
0.9 & -2.81 & -2.04 & -1.62 & 1.01 & 4.73 & 5.57 & 7.30 \\
1.0 & -2.71 & -1.95 & -1.53 & 1.18 & 5.02 & 5.88 & 7.62 \\
1.1 & -2.63 & -1.86 & -1.42 & 1.39 & 5.37 & 6.23 & 8.01 \\
1.2 & -2.56 & -1.75 & -1.31 & 1.60 & 5.69 & 6.59 & 8.38 \\
1.3 & -2.47 & -1.64 & -1.18 & 1.84 & 6.06 & 6.96 & 8.80 \\
1.4 & -2.37 & -1.52 & -1.05 & 2.10 & 6.44 & 7.37 & 9.22 \\
1.5 & -2.25 & -1.38 & -0.89 & 2.37 & 6.82 & 7.79 & 9.72 \\
\hline
\end{tabular}
\quad
\begin{tabular}{|r|r|r|r|r|r|r|r|}
\hline
\diagbox[width=0.9cm, height=0.75cm]{$\Theta$}{$\alpha$}  & .005 & .025 & .05 & .5 & .95 & .975 & .995 \\
\hline
\hline
1.5 & -2.25 & -1.38 & -0.89 & 2.37 & 6.82 & 7.79 & 9.72 \\
1.6 & -2.14 & -1.23 & -0.73 & 2.67 & 7.27 & 8.22 & 10.19 \\
1.7 & -2.03 & -1.08 & -0.55 & 3.00 & 7.67 & 8.65 & 10.66 \\
1.8 & -1.88 & -0.91 & -0.37 & 3.34 & 8.12 & 9.13 & 11.19 \\
1.9 & -1.73 & -0.72 & -0.14 & 3.70 & 8.59 & 9.62 & 11.71 \\
2.0 & -1.56 & -0.53 & 0.07 & 4.09 & 9.09 & 10.14 & 12.20 \\
2.1 & -1.39 & -0.30 & 0.34 & 4.50 & 9.61 & 10.66 & 12.84 \\
2.2 & -1.21 & -0.05 & 0.62 & 4.91 & 10.12 & 11.20 & 13.31 \\
2.3 & -0.99 & 0.22 & 0.92 & 5.37 & 10.67 & 11.75 & 13.96 \\
2.4 & -0.77 & 0.51 & 1.24 & 5.84 & 11.24 & 12.36 & 14.52 \\
2.5 & -0.53 & 0.82 & 1.59 & 6.32 & 11.82 & 12.94 & 15.14 \\
2.6 & -0.24 & 1.16 & 1.96 & 6.82 & 12.45 & 13.57 & 15.88 \\
2.7 & 0.04 & 1.52 & 2.37 & 7.35 & 13.05 & 14.22 & 16.52 \\
2.8 & 0.36 & 1.91 & 2.78 & 7.90 & 13.68 & 14.85 & 17.18 \\
2.9 & 0.73 & 2.34 & 3.23 & 8.47 & 14.35 & 15.55 & 17.90 \\
3.0 & 1.08 & 2.77 & 3.70 & 9.04 & 15.02 & 16.21 & 18.60 \\
3.1 & 1.52 & 3.27 & 4.22 & 9.67 & 15.71 & 16.92 & 19.38 \\
3.2 & 1.93 & 3.72 & 4.72 & 10.28 & 16.42 & 17.63 & 20.07 \\
3.3 & 2.41 & 4.25 & 5.26 & 10.93 & 17.15 & 18.40 & 20.85 \\
3.4 & 2.93 & 4.79 & 5.82 & 11.61 & 17.91 & 19.18 & 21.74 \\
3.5 & 3.40 & 5.36 & 6.41 & 12.29 & 18.65 & 19.93 & 22.48 \\
3.6 & 3.96 & 5.95 & 7.03 & 13.01 & 19.48 & 20.78 & 23.37 \\
3.7 & 4.52 & 6.59 & 7.68 & 13.72 & 20.26 & 21.56 & 24.14 \\
3.8 & 5.12 & 7.18 & 8.30 & 14.48 & 21.13 & 22.47 & 25.05 \\
3.9 & 5.76 & 7.89 & 9.01 & 15.25 & 21.95 & 23.31 & 25.98 \\
4.0 & 6.38 & 8.57 & 9.72 & 16.05 & 22.83 & 24.17 & 26.87 \\
4.1 & 7.04 & 9.25 & 10.41 & 16.85 & 23.70 & 25.06 & 27.75 \\
4.2 & 7.70 & 9.93 & 11.15 & 17.68 & 24.61 & 25.96 & 28.68 \\
4.3 & 8.41 & 10.69 & 11.91 & 18.52 & 25.55 & 26.94 & 29.71 \\
4.4 & 9.14 & 11.49 & 12.73 & 19.41 & 26.49 & 27.88 & 30.62 \\
4.5 & 9.93 & 12.27 & 13.52 & 20.28 & 27.43 & 28.84 & 31.60 \\
4.6 & 10.68 & 13.09 & 14.35 & 21.20 & 28.43 & 29.81 & 32.63 \\
4.7 & 11.46 & 13.91 & 15.19 & 22.12 & 29.41 & 30.84 & 33.67 \\
4.8 & 12.23 & 14.76 & 16.07 & 23.06 & 30.44 & 31.88 & 34.68 \\
4.9 & 13.12 & 15.63 & 16.94 & 24.03 & 31.52 & 32.99 & 35.90 \\
5.0 & 13.97 & 16.54 & 17.87 & 25.04 & 32.55 & 34.02 & 36.98 \\
5.1 & 14.88 & 17.44 & 18.80 & 26.04 & 33.64 & 35.13 & 38.05 \\
5.2 & 15.76 & 18.39 & 19.73 & 27.06 & 34.74 & 36.22 & 39.18 \\
5.3 & 16.70 & 19.36 & 20.71 & 28.11 & 35.83 & 37.34 & 40.37 \\
5.4 & 17.67 & 20.35 & 21.75 & 29.19 & 37.01 & 38.52 & 41.57 \\
5.5 & 18.57 & 21.35 & 22.75 & 30.28 & 38.14 & 39.70 & 42.71 \\
5.6 & 19.63 & 22.36 & 23.78 & 31.40 & 39.34 & 40.91 & 43.98 \\
5.7 & 20.67 & 23.43 & 24.85 & 32.54 & 40.53 & 42.08 & 45.15 \\
5.8 & 21.65 & 24.46 & 25.91 & 33.67 & 41.72 & 43.31 & 46.43 \\
5.9 & 22.78 & 25.52 & 27.00 & 34.83 & 42.97 & 44.54 & 47.70 \\
6.0 & 23.82 & 26.63 & 28.11 & 36.03 & 44.25 & 45.86 & 48.95 \\
\hline
\end{tabular}
}
\caption{\label{Table_T_quantiles} Quantiles of $\mathcal T(\Theta)$ for $-3\le \Theta\le 6$ based on $MC=10^6$ Monte Carlo simulations.}
\end{table}

\begin{table}
\begin{tabular}{|r|r|r|r|r|r|r|r|}
\hline
$\alpha$  & .005 & .025 & .05 & .5 & .95 & .975 & .995 \\
\hline
quantile $F_1^{-1}(\alpha)$ & -4.15& -3.52& -3.18& -1.27& 0.98& 1.45& 2.42\\
\hline
\end{tabular}
 \caption{\label{Table_TW_quantiles} Quantiles of the Tracy--Widom$_1$ distribution from \citet{Bejan}.}
\end{table}

\begin{itemize}
\item For $-3\le \Theta\le 6$, quantiles are tabulated in Table \ref{Table_T_quantiles} using the algorithm described in Section \ref{Section_Def_Gw}.
\item For large positive values of $\Theta$, the Gaussian approximation $\mathcal T(\Theta)\approx \mathcal N(\Theta^2 ,4\Theta)$ should be used, i.e.,
\[ \mathbb P \left\{ \mathcal T(\Theta) \leq t \right\} \approx \Phi((t-\Theta^2)/(2\sqrt{\Theta}))~.\]
\item For large negative values of $\Theta$, the Tracy--Widom$_1$ approximation should be used:
\[  \mathbb P \left\{ \mathcal T(\Theta) \leq t \right\}  \approx F_1 (t + 1/\Theta)~, \]
where the relevant Tracy--Widom quantiles are provided in Table \ref{Table_TW_quantiles}.
\end{itemize}

The transition process $\mathcal T(\Theta)$, with appropriate centering and scaling, can be used to approximate the fluctuations of the largest eigenvalues, leading to the following algorithm.
\begin{procedure} \label{Proposal_uniform_asymptotics}
For each of the four models in Section \ref{Section_signal_plus_noise_models} with $\sigma^2=1$, the asymptotic distribution of the largest eigenvalues $\lambda_i$ can be approximated as:
\begin{equation}
 \label{eq_Unified_asymptotics_Proposal}
\begin{cases}
 \lambda(\theta_i)-\frac{ \kappa_2^{3/2}}{2} \sqrt{V(\theta_i) (\theta_i-\theta^c)^3}+ \frac{\kappa_2^{-1/2}}{2 N^{2/3}} \sqrt{\frac{V(\theta_i)}{\theta_i-\theta^c}} \mathcal T \Bigl(\kappa_2 N^{1/3} (\theta_i-\theta^c) \Bigr)+\frac{\kappa_3}{N}, & \text{ if } \theta_i>\theta^c,\\
 \lambda_+ + N^{-2/3} \kappa_1 \mathcal T \Bigl(\kappa_2 N^{1/3} (\theta_i-\theta^c) \Bigr)+\frac{\kappa_3}{N}, & \text{ if } \theta_i \le \theta^c,
\end{cases}
\end{equation}
where the constants are taken from \eqref{eq_Spiked_Wigner_params}, \eqref{eq_Spiked_covariance_params}, \eqref{eq_Factor_params}, \eqref{eq_CCA_params}; $\kappa_1=\frac{1}{2}  \frac{[V'(\theta^c)]^{2/3}}{[\lambda''(\theta^c)]^{1/3}}$, ${\kappa_2= \frac{[\lambda''(\theta^c)]^{2/3}}{[V'(\theta^c)]^{1/3}}}$,  $\kappa_3=-\frac{3}{2}\frac{ \kappa_1}{\kappa_2 \theta^c}$, and we assume $\theta_{i-1}>\theta^c$.
\end{procedure}

Theorem \ref{Theorem_main_convergence_statement} and Corollary \ref{Corollary_uniform_asymptotics} establish that the approximation \eqref{eq_Unified_asymptotics_Proposal} is valid both when $\theta$ is bounded away from the critical value $\theta^c$ and when $\theta$ is close to $\theta^c$.  One can also show that the approximation remains valid as $\theta\to\infty$ (corresponding to strong signals)\footnote{In the case of CCA, $\theta$ represents a correlation and is therefore bounded by $1$.}; we omit the proof since the present work focuses on the weak-signal regime. As is evident from the simulations in Figure \ref{Fig_confidence_int}, the quality of the approximation improves as $\theta \to \infty$; we provide one formal result in this direction at the end of Section \ref{Section_spiked_covariance_proof}.

Our results further show that, in many cases, the approximations for different $\lambda_i$ are asymptotically independent. Consequently, we can use \eqref{eq_Unified_asymptotics_Proposal} as a foundation for constructing confidence intervals.

\subsection{Confidence intervals with known $\sigma^2$: the first algorithm} \label{Section_confidence_algorithm_1} We begin with the case where the noise variance $\sigma^2$ is known, as specified in Section \ref{Section_signal_plus_noise_models}. A simple rescaling allows us to assume $\sigma^2 = 1$ without loss of generality. The algorithm then proceeds as follows:

\smallskip

\textbf{The first step} is to draw a histogram of all eigenvalues $\lambda_1,\lambda_2,\dots$. In the settings of \eqref{eq_Spiked_Wigner}, \eqref{eq_Covariance_spiked}, \eqref{eq_factor_model}, or \eqref{eq_CCA_model}, the histogram should resemble a known limiting shape; namely, the semicircle law, Marchenko-Pastur law, or Wachter law, depending on the model, as detailed in Table \ref{Table_LLN_shapes}, with parameters specified in Table \ref{Table_parameters}, see Appendix \ref{Section_beta_ensembles_a} for more details. If the histogram is reminiscent of one of these shapes, we regard the modelling assumptions as valid and apply Procedure \ref{Proposal_uniform_asymptotics} to construct confidence intervals. Section \ref{Section_extension_limit_shapes} discusses possible extensions when the empirical histogram deviates from the expected limit shape.

\begin{table}[t]
\centering
\begin{tabular}{|c|c|}
\hline
\textbf{Spiked Wigner matrix} & \textbf{Spiked covariance model} \\
\hline
\hline
$\displaystyle  \frac{1}{2\pi}\sqrt{4-x^2} \, \mathbf{1}_{[-2,2]} \, \dd x$ &
$\displaystyle
 \frac{1}{2\pi} \frac{\sqrt{(\lambda_+-x)(x-\lambda_-)}}{\gamma^2 x} \, \mathbf{1}_{[\lambda_-,\lambda_+]} \, \dd x
$
\\
\textit{Semicircle law} & \textit{Marchenko-Pastur law}\\
\hline
\hline
\textbf{Factor model} & \textbf{CCA} \\
\hline
\hline
$\displaystyle
 \frac{1}{2\pi} \frac{\sqrt{(\lambda_+-x)(x-\lambda_-)}}{\gamma^2 x} \, \mathbf{1}_{[\lambda_-,\lambda_+]} \, \dd x
$
 &
$\displaystyle
 \frac{\tau_N}{2\pi} \frac{\sqrt{(\lambda_+ - x)(x - \lambda_-)}}{x(1-x)} \mathbf{1}_{[\lambda_-,\lambda_+]} \, \dd x
$
\\
\textit{Marchenko-Pastur law} & \textit{Wachter law}
\\
\hline
\end{tabular}
 \caption{Limiting behavior of the empirical measures  of eigenvalues, $\lim_{N\to\infty} \frac{1}{N} \sum_{i=1}^N \delta_{\lambda_i}$, in signal plus noise models.\label{Table_LLN_shapes}}
\end{table}

\begin{table}[t]
\centering
\renewcommand{\arraystretch}{1.2}
\begin{tabular}{|c|c|c|}
\hline
& \textbf{Spiked Wigner matrix} & \textbf{Spiked covariance model} \\
\hline
\hline
Parameters& -- & $\gamma^2=\tfrac{N}{S}\in (0,1]$\\
\hline
$\theta^c$& $1$&  $1+\gamma$
 \\
\hline
$\lambda_\pm$& $\pm2$ & $(1\pm \gamma)^2$\\
\hline
$\lambda(\theta)$& $\theta+\frac{1}{\theta}$&
 $\theta+\gamma^2 \frac{\theta}{\theta-1}$\\
\hline
$V(\theta)$& $2\frac{\theta^2-1}{\theta^2}$& $2 \theta^2 \gamma^2 \left(1 - \frac{\gamma^2}{(\theta-1)^2}\right)$
\\
\hline
$\kappa_1$, $\kappa_2$, $\kappa_3$ & $1$,\, $1$, \, $-\frac{3}{2}$  &  $\gamma(1+\gamma)^{4/3}$, \, $\frac{1}{\gamma (1+\gamma)^{2/3}}$,\, $-\frac{3}{2} \gamma(1+\gamma)^2$
\\
\hline
\hline
&\textbf{Factor model} & \textbf{CCA} \\
\hline\hline
Parameters& $\gamma^2=\tfrac{N}{S}\in (0,1]$ &
\begin{minipage}{0.4\textwidth}
$\tau_N=\tfrac{S}{N}>1$, $\tau_M=\tfrac{S}{M}>1$,\\ with $\tau_N^{-1} + \tau_M^{-1} < 1$, \, $\tau_N \ge \tau_M$
\end{minipage}
\\
\hline
$\theta^c$ & $\gamma$ &  $\frac{1}{\sqrt{(\tau_M-1)(\tau_N-1)}}$ 
\\
\hline
$\lambda_\pm$&  $(1\pm \gamma)^2$ & $\left(\sqrt{\tau_M^{-1}(1 - \tau_N^{-1})} \pm \sqrt{\tau_N^{-1}(1 - \tau_M^{-1})}\right)^2$
\\
\hline
$\lambda(\theta)$&  $\theta+1+\gamma^2 \frac{\theta+1}{\theta}$
  & $\frac{\bigl(  (\tau_N-1)\theta  + 1 \bigr) \bigl(  (\tau_M-1) \theta + 1\bigr)}{\theta \tau_N \tau_M }$
\\
\hline
$V(\theta)$& $2 \gamma^2 \frac{(2 \theta+1+\gamma^2)(\theta^2-\gamma^2)}{\theta^2}$ & \begin{minipage}{0.45\textwidth} \small \begin{multline*}
2 \frac{(1-\theta)^2}{\theta^2 \tau_M^2\tau_N^3} \bigl(2(\tau_M-1)(\tau_N-1)\theta+\tau_M+\tau_N-2\bigr)\\ \times \bigl(
 (\tau_M-1)(\tau_N-1)\theta^2 -1\bigr)\end{multline*}
 \end{minipage}
\\
\hline
$\kappa_1$, $\kappa_2$, $\kappa_3$ & $\gamma(1+\gamma)^{4/3}$, \, $\frac{1}{\gamma (1+\gamma)^{2/3}}$, \, $-\frac{3}{2} \gamma(1+\gamma)^2$ &
\begin{minipage}{0.4\textwidth}
$\frac{(\sqrt{\tau_N-1}\sqrt{\tau_M-1}-1)^{4/3}(\sqrt{\tau_N-1}+\sqrt{\tau_M-1})^{4/3}}{ \tau_N^{5/3}\tau_M(\tau_N-1)^{1/6}(\tau_M-1)^{1/6}}$
,
\\

$\frac{ \tau_N^{1/3} (\tau_N-1)^{5/6}(\tau_M-1)^{5/6} }{(\sqrt{\tau_N-1}\sqrt{\tau_M-1}-1)^{2/3}(\sqrt{\tau_N-1}+\sqrt{\tau_M-1})^{2/3}}$
,
\\

$
-\frac{3}{2} \frac{(\sqrt{\tau_N-1}\sqrt{\tau_M-1}-1)^{2}(\sqrt{\tau_N-1}+\sqrt{\tau_M-1})^{2}}{ \tau_N^{2}\tau_M \sqrt{\tau_N-1}\sqrt{\tau_M-1}}
$
\end{minipage}
\\
\hline
\end{tabular}
 \caption{Parameters. In the factor model the roles of $S$ and $N$ can be swapped when $\gamma^2 > 1$. In CCA $N$ and $M$ can be swapped when $\tau_N < \tau_M$.
 \label{Table_parameters}}
\end{table}

\smallskip

\textbf{For the second step}, we choose a significance level $\alpha$ (or confidence level $1-\alpha$) and, using Section \ref{Section_transition_function}, construct two deterministic functions $t_{\alpha/2,+}(\Theta)$ and $t_{\alpha/2,-}(\Theta)$ such that
\begin{equation}
 \mathrm{Prob}\bigl(\mathcal T(\Theta)>t_{\alpha/2,+}(\Theta)\bigr)=\mathrm{Prob}\bigl(\mathcal T(\Theta)<t_{\alpha/2,-}(\Theta)\bigr)=\frac{\alpha}{2}.
\end{equation}

Following \eqref{eq_Unified_asymptotics_Proposal} and using the parameter choices from Table \ref{Table_parameters}, we rescale the functions $t_{\alpha/2,\pm}(\Theta)$ to obtain $\widehat{t}_{\pm}(\theta)$, defined as
\begin{equation}
\label{eq_confidence_intervals_final}
 \widehat{t}_{\pm}(\theta)=\begin{cases} \lambda(\theta)-\frac{ \kappa_2^{3/2}}{2} \sqrt{V(\theta) (\theta-\theta^c)^3}+ \frac{\kappa_2^{-1/2}}{2 N^{2/3}} \sqrt{\frac{V(\theta)}{\theta-\theta^c}}  t_{\alpha/2,\pm} \Bigl(\kappa_2 N^{1/3} (\theta-\theta^c) \Bigr)+\frac{\kappa_3}{N}, & \theta>\theta^c,\\
 \lambda_++\kappa_1 N^{-2/3} t_{\alpha/2,\pm} \Bigl(\kappa_2 N^{1/3} (\theta-\theta^c) \Bigr)+\frac{\kappa_3}{N},&\theta\le \theta^c.\end{cases}
\end{equation}

\smallskip

\textbf{For the third step}, we fix an index $i$ and consider the $i$th largest eigenvalue $\lambda_i$, such that $\lambda_i > \lambda_+$. We then determine two numbers $\theta_-<\theta_+$ such that
\begin{equation}
  \widehat{t}_{+}(\theta_-)=
  \widehat{t}_{-}(\theta_+)= \lambda_i.
\end{equation}
The procedure amounts to plotting the functions $\theta\mapsto \widehat{t}_{\pm}(\theta)$ and finding their intersection with the horizontal line $y=\lambda_i$. The resulting $[\theta_-,\theta_+]$ serves as the confidence interval for the $i$th signal strength $\theta_i$.  We have $\mathrm{Prob}(\theta_i\in [\theta_-,\theta_+])\to 1-\alpha$ as $N\to\infty$ by Corollary \ref{Corollary_uniform_asymptotics}.

There are two special cases to consider at this step. First, it may happen that no value $\theta_-$ satisfies $\widehat{t}_{+}(\theta-) = \lambda_i$. This occurs when the shifted and rescaled $\lambda_i$ falls below the $(1-\alpha/2)$ quantile of the Tracy-Widom distribution $F_1$. In this case the confidence interval becomes one-sided, and one should set $\theta_-=-\infty$ or, equivalently, to the lower bound of admissible values of $\theta_i$, that is $\theta_i\ge \sigma^2$ for the spiked covariance and $\theta_i\ge 0$ for the others. Second, it may happen that $\theta_-$ exists, but lies below the lower bound for admissible values of $\theta_i$. In this case $\theta_-$ should again be replaced by the appropriate lower bound. In terms of statistical consequences the two cases are equivalent.


\subsection{Bootstrap algorithm}
\label{Section_confidence_algorithm_2}

An important and unexpected feature of our asymptotic results (see Theorem \ref{Theorem_main_convergence_statement} for details) is that for a fixed index $i$ the asymptotic approximation of $\lambda_i$ depends only on $\theta_i$, but not on other parameters of the models \eqref{eq_Spiked_Wigner}, \eqref{eq_Covariance_spiked}, \eqref{eq_factor_model}, \eqref{eq_CCA_model}, such as $\{\theta_j\}_{j\ne i}$ or vectors $\u^*_i$.  Additionally, the rank $r$ does not enter into the formulas. Hence, instead of relying on formulas \eqref{eq_confidence_intervals_final} for the functions $ \widehat{t}_{\pm}(\theta)$, we can obtain them through a bootstrap procedure using only the $r=1$ case.


{\bf Alternative formula-free second step.} Consider one of the models \eqref{eq_Spiked_Wigner}, \eqref{eq_Covariance_spiked}, \eqref{eq_factor_model}, or \eqref{eq_CCA_model} with the desired matrix sizes $N$, $S$, and $M$, and instead of the true rank set $r=1$. Fix the direction of the unique signal arbitrarily; for example, set $\u^*_1$ (and $\v^*_1$ for the factor model, or the canonical variables for CCA) to the first coordinate vector. The model then depends on a single remaining parameter, $\theta_1$. Discretize $\theta_1$ on a grid, and for each value compute the largest eigenvalue $\lambda_1$ of the corresponding model matrix ($\A$ for the spiked Wigner matrix, $\frac{1}{S} X X^\T$ for the spiked covariance and factor models, and $(\U \U^\T)^{-1} \U \V^\T (\V \V^\T)^{-1} \V \U^\T$ for CCA). By repeating sufficiently many Monte-Carlo simulations, compute the $\alpha/2$ and $(1-\alpha/2)$ quantiles of $\lambda_1$ for each $\theta_1$, which yield the desired $\widehat{t}_{-}(\theta)$ and $\widehat{t}_{+}(\theta)$. This procedure produces the thick gray curves in Figure~\ref{Fig_confidence_int}. Then proceed to the third step as in Section~\ref{Section_confidence_algorithm_1}.

An advantage of the algorithm in this section is that \eqref{eq_confidence_intervals_final} and Tables \ref{Table_T_quantiles} and \ref{Table_parameters} are not required. However, this comes at the cost of running many Monte-Carlo simulations, making the implementation slower than that of the algorithm in Section \ref{Section_confidence_algorithm_1}.

\subsection{Unknown $\sigma^2$} \label{Section_unknown_variance}
For the CCA setting in Section \ref{Section_spiked_CCA} the asymptotics in \eqref{eq_CCA_as} does not depend on the noise covariance, i.e., the matrices $A$ and $B$ in \eqref{eq_CCA_model}. In contrast, for the other three settings, Sections \ref{Section_spiked_Wigner}, \ref{Section_spiked_covariance}, and \ref{Section_factor}, the scaling depends on the noise variance, denoted by $\sigma^2$. The same holds for Theorem~\ref{Theorem_main_convergence_statement}, which underlies the algorithms for constructing confidence intervals in Sections \ref{Section_confidence_algorithm_1} and \ref{Section_confidence_algorithm_2}. In particular, Procedure \ref{Proposal_uniform_asymptotics} assumes $\sigma^2 = 1$. If $\sigma^2 \neq 1$ but is known, then the entries of the data matrix $\A$ or $X$ should be divided by $\sigma$ to reduce to the baseline case $\sigma^2 = 1$. If $\sigma^2$ is unknown, it must first be estimated.

We propose estimating the variance by discarding $25\%$ of the eigenvalues at both ends and matching sample moments to their theoretical values to solve for $\sigma^2$.

For the spiked Wigner model, let $\ell\approx 0.81$ 
denote the positive number such that
\begin{equation}
\label{eq_quantile_1}
  \int_{-2}^{-\ell} \frac{1}{2\pi} \sqrt{4-x^2} \dd x = \int_{\ell}^{2} \frac{1}{2\pi} \sqrt{4-x^2} \dd x=\frac{1}{4},
\end{equation}
and set
\begin{equation}
\label{eq_x18}
  \sigma_0^2=\int_{-\ell}^{\ell} \frac{x^2}{2\pi} \sqrt{4-x^2} \dd x.
  \end{equation}
Given eigenvalues $\lambda_1\ge \dots\ge \lambda_N$ of $\A$, we can form an estimate
\begin{equation}
\label{eq_x19}
 \widehat \sigma^2 =\frac{1}{\sigma_0^2} \frac{1}{N} \sum_{i=\lfloor N/4\rfloor+1}^{\lfloor 3 N/4\rfloor} \lambda_i^2.
\end{equation}
The Wigner semicircle law for the GOE with explicit estimates for the remainders (see e.g., \citet{o2010gaussian}), combined with the interlacing inequalities between the eigenvalues of $\A$ and $\B$ in \eqref{eq_Spiked_Wigner}, as in Corollary \ref{Corollary_Wigner_interlacement}, can be used to show that
\begin{equation}
 \label{eq_variance_standard}
  \widehat \sigma^2=\sigma^2 + O\left(\frac{\log(N)}{N}\right), \qquad N\to\infty.
\end{equation}
The scale of the random component in Theorem \ref{Theorem_main_convergence_statement} is much larger than the error term in \eqref{eq_variance_standard}. Henc3, our confidence intervals are much wider than this error term and normalizing the data by $\widehat \sigma$ does not change the validity of the confidence intervals of Sections \ref{Section_confidence_algorithm_1}, \ref{Section_confidence_algorithm_2}.

\smallskip

For the spiked covariance and factor models, the procedure is analogous, but relies on the Marchenko-Pastur law (see Table \ref{Table_LLN_shapes})
rather than the semicircle law. Fixing the parameter $\gamma^2=\frac{N}{S}\in [0,1)$, we define $\ell_-$ and $\ell_+$ as two positive numbers such that
\begin{equation}
\label{eq_quantile_2}
  \int_{\lambda_-}^{\ell-}  \frac{1}{2\pi} \frac{\sqrt{(\lambda_+-x)(x-\lambda_-)}}{\gamma^2 x} \dd x = \frac{1}{4}, \qquad
  \int_{\ell_+}^{\lambda_+}  \frac{1}{2\pi} \frac{\sqrt{(\lambda_+-x)(x-\lambda_-)}}{\gamma^2 x} \dd x = \frac{1}{4}.
\end{equation}
and set
\begin{equation}
\label{eq_x20}
  \sigma_0^2=\int_{\ell_1}^{\ell_2} \frac{x}{2\pi} \frac{\sqrt{(\lambda_+-x)(x-\lambda_-)}}{\gamma^2 x}  \dd x.
  \end{equation}
Given eigenvalues $\lambda_1\ge \dots\ge \lambda_N$ of $\frac{1}{S} X X^\T$, we can form an estimate
\begin{equation}
\label{eq_x20_2}
  \widehat \sigma^2 =\frac{1}{\sigma_0^2} \frac{1}{N} \sum_{i=\lfloor N/4\rfloor+1}^{\lfloor 3 N/4\rfloor} \lambda_i.
\end{equation}
The Marchenko--Pastur law with explicit estimates for the reminders (see e.g., \citet{bourgade2022optimal}), combined with the interlacing inequalities between the eigenvalues of spiked and unspiked models, as in Corollaries \ref{Corollary_spiked_covariance_interlacement} and \ref{Corollary_factor_interlacement}, can be used to show that
\begin{equation}
 \label{eq_variance_standard_2}
  \widehat \sigma^2=\sigma^2 + O\left(\frac{\log(N)}{N}\right), \qquad N\to\infty.
\end{equation}
Once again, normalizing the data by dividing by $\widehat \sigma$ does not affect the validity of the confidence intervals constructed in the previous section.

For a discussion of alternative procedures for estimating $\sigma^2$ see, for example, \citet[Section III.C]{kritchman2009non}, \citet[Section 4.1]{shabalin2013reconstruction},  \citet[Section 3.E]{gavish2014optimal}, or \citet[End of Section 2]{ke2023estimation}.

\subsection{Implications and interpretations}

Confidence intervals play two key roles in the analyzing signal strength. First, they measure uncertainty: the narrower the interval, the more precisely the signal is estimated.

Second, they assess the informativeness of estimated signals. If the lower bound starts at $-\infty$ or at the minimal admissible value of $\theta$, (which is $\sigma^2$ for the spiked covariance and zero for other models), the signal may be spurious and reflect pure noise.
Alternatively, if the  interval is bounded away from the minimal admissible value, but contains the identification threshold $\theta^c$, then a signal exists, but we cannot reject that its strength falls below the cutoff. When $\theta \leq \theta^c$ the sample estimates of the directions $\u$ and $\v$ are asymptotically orthogonal to their true counterparts (see, e.g., \citet{paul2007asymptotics,onatski2012asymptotics,benaych2012singular,johnstone2018pca,BG_CCA}), rendering the signal effectively non-informative.



\section{Asymptotics through the Airy--Green function}
\label{Section_asymptotics}

The new asymptotics, which improves upon the Gaussian approximations \eqref{eq_spiked_Wigner_Gaussian}, \eqref{eq_spiked_covariance_as}, \eqref{eq_factors_as}, and \eqref{eq_CCA_as}, is based on a novel stochastic object we call the \emph{Airy--Green function}. Its definition, along with the transition process $\mathcal T(\Theta)$ constructed from it, is presented in Section \ref{Section_Def_Gw}. A discussion of its nature is provided in Section \ref{subsection_GT_discussion}. Theorem on the convergence of the eigenvalue distributions in four models towards this object is stated in Section \ref{Section_convergence_to_transitional}.

\subsection{Definition of $\mathcal G(w)$ and $\mathcal T(\Theta)$}
 \label{Section_Def_Gw}

We begin by recalling the Airy$_1$ point process, a random sequence of points $\aa_1 \ge \aa_2 \ge \aa_3 \ge \dots$, which can be defined as the scaling limit of the largest eigenvalues of Wigner matrices.
\begin{prop}[\citet{forrester1993spectrum,tracy1996orthogonal}] \label{Proposition_Airy_Gauss} Let $Y_N$ be an $N\times N$ matrix of i.i.d.~$\mathcal{N}(0,\tfrac{2}{N})$ Gaussian random variables and let $\lambda_{1;N}\ge \lambda_{2;N}\ge \dots \ge\lambda_{N;N}$ be the eigenvalues of $\B=\frac{1}{2}\left(Y_N+Y_N^\T\right)$. Then in  finite-dimensional distributions
	\begin{equation}
	\label{eq_GOE_to_Airy}
	\lim_{N\to\infty} \left\{N^{2/3}\left(\lambda_{i;N}-2\right) \right\}_{i=1}^N = \{ \aa_i\}_{i=1}^\infty.
	\end{equation}
\end{prop}
Similar asymptotic results hold for the other models we consider. All existing formulae for the finite-dimensional distributions of  $\{\aa_i\}_{i=1}^{\infty}$ are quite complicated and do not provide explicit distribution function, see, e.g., \cite{forrest}. Nevertheless, the distribution can be sampled and tabulated, see \cite{bornemann2009numerical,vignette_largevars}. In particular, Table \ref{Table_TW_quantiles} lists quantiles of the Tracy–Widom distribution, which describes the law of $\aa_1$.

\smallskip

The following theorem defines the \emph{Airy--Green function} $\mathcal G(w)$; see Section \ref{Section_Gw} for the proof.

\begin{theorem} \label{Theorem_G_def_main_text} Let $\aa_1\ge \aa_2\ge\aa_3\ge\dots$ be a realization of the Airy$_1$ point process and let $\{\xi_j\}_{j=1}^{\infty}$ be i.i.d.~$\mathcal N(0,1)$ independent of $\{\aa_j\}_{j=1}^{\infty}$. Almost surely, for each $w \in \mathbb C \setminus \{\aa_j\}$ there exists a (random) limit
\begin{equation}
\label{eq_G_definition}
 \mathcal G(w)=\lim_{x\to-\infty} \left[\left(\sum_{j:\, \aa_j>x} \frac{\xi_j^2}{w-\aa_j}\right) - \frac{2}{\pi}\sqrt{-x} \right]~,
\end{equation}
and, moreover, the convergence is uniform on any compact set $W\subset \mathbb{C}$ disjoint from $\{\aa_j\}$.
\end{theorem}

Note that any fixed $w \in \mathbb C$ is almost surely not in $\{\aa_j\}$, hence, for such $w$ the convergence holds almost surely. Eq.\ \eqref{eq_G_definition} and Proposition \ref{Proposition_G_at_infinity} imply that $\mathcal G(w)$ changes monotonically from $+\infty$ to $-\infty$ over the interval $[\aa_1,+\infty)$, allowing us to state the following key definition.

\begin{definition} \label{Definition_Transition_function}
 The \emph{transition process} $\mathcal T(\Theta)$, $\Theta\in\mathbb R$, is a random function, defined as the unique solution to the equation $\mathcal G(w)=-\Theta$ satisfying $w\in [\aa_1,+\infty)$.
\end{definition}

\begin{proposition} \label{Proposition_Transition_to_Gauss}
Almost surely, $\Theta\mapsto \mathcal T(\Theta)$ is an increasing bijection of $\mathbb R$ onto $(\aa_1, \infty)$.  As $\Theta\to +\infty$, $\mathcal T(\Theta)$ is asymptotically Gaussian: 
 \begin{equation}
 \label{eq_Transition_to_Gauss}
  \lim_{\Theta\to +\infty} \frac{\mathcal T(\Theta)-\Theta^2}{2 \sqrt{\Theta}}\stackrel{d}{=} \mathcal N(0,1).
 \end{equation}
\end{proposition}
\begin{remark} \label{Remark_large_negative_Theta}
For large negative $\Theta$ we have distributional approximations:
 \begin{equation}
 \label{eq_large_negative_Theta}
  \mathcal T(\Theta)\stackrel{d}{=} \aa_1-\frac{\xi_1^2}{\Theta} + O\left(\frac{1}{\Theta^2}\right)\stackrel{d}{=} \aa_1-\frac{1}{\Theta}+ O\left(\frac{1}{\Theta^2}\right), \qquad \Theta\to-\infty.
 \end{equation}
 The first approximation follows directly from \eqref{eq_G_definition}; the second from writing the distribution function as the expectation of the distribution function of $\aa_1$ shifted by random $\frac{\xi_1^2}{\Theta}$.
\end{remark}

Figure \ref{Fig_Transition distributions} and Table \ref{Table_T_quantiles} show the simulated quantiles for the random variables $\mathcal T(\Theta)$ as functions of $\Theta$, or equivalently, confidence intervals for $\Theta$ as a function of $\mathcal T$. These results are based on $MC=10^6$ Monte Carlo simulations of the $\sqrt{N}\times\sqrt{N}$ top-left corners of $N\times N$ tridiagonal matrices of \citet{dumitriu_edelman}, with a perturbed $(1,1)$ matrix element and $N=10^8$; see \citet[Section 1.1]{edelman2005numerical} and \citet[Lemma 5.2]{johnstone2021spin} for justifications of this approach. The figure shows that the Gaussian approximation from Proposition \ref{Proposition_Transition_to_Gauss} performs well for large $\Theta$, but deteriorates near $\Theta = 0$.

\begin{figure}[t]
\begin{subfigure}{.49\textwidth}
  \centering
  \includegraphics[width=1.0\linewidth]{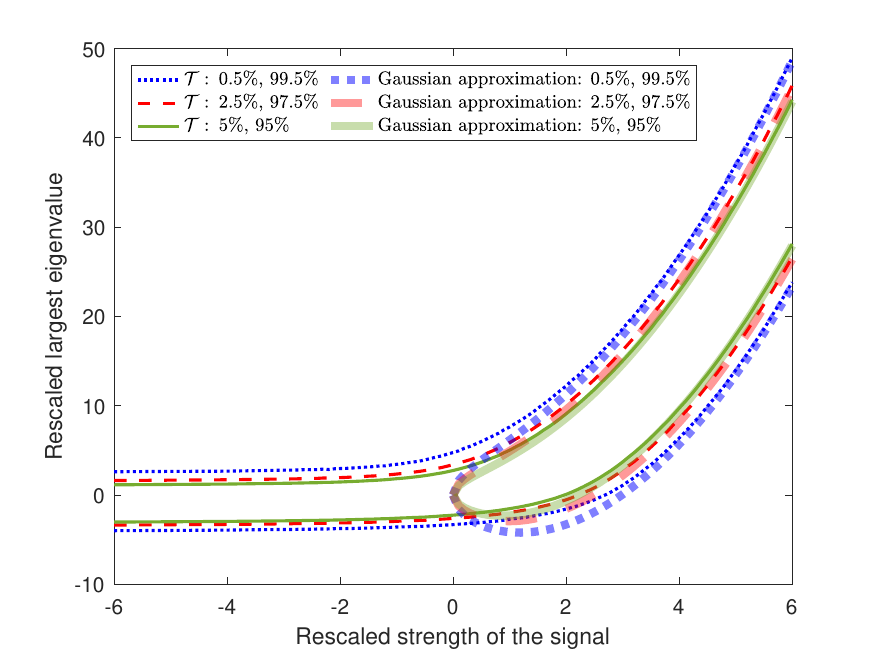}
  \caption{Quantiles of $\mathcal T(\Theta)$ as functions of $\Theta$.}
  \label{Fig_Transition distributions1}
\end{subfigure}%
\begin{subfigure}{.49\textwidth}
  \centering
  \includegraphics[width=1.0\linewidth]{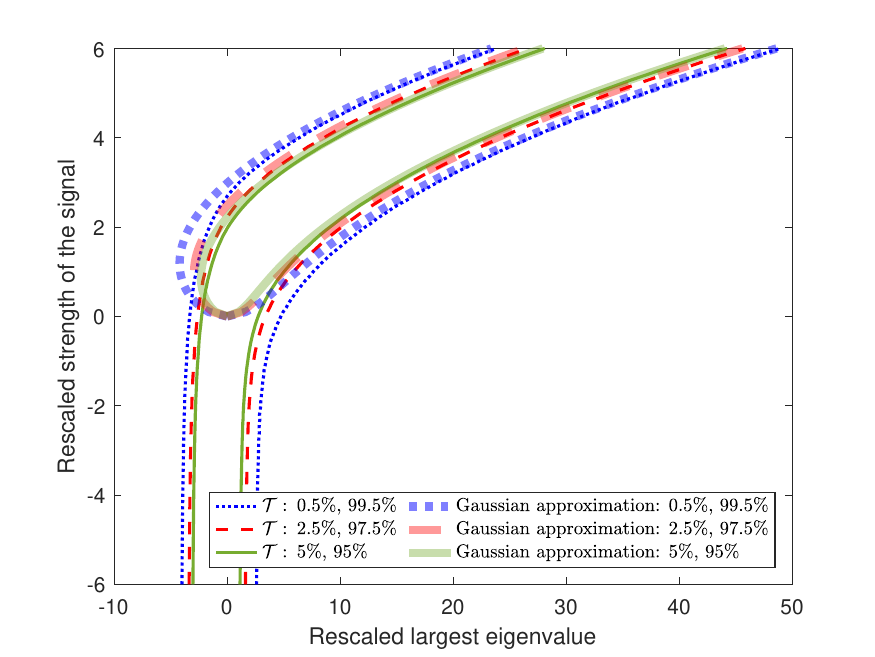}
  \caption{Confidence intervals for $\Theta$ (transposed axes).}
  \label{Fig_Transition distributions2}
\end{subfigure}
\caption{Quantiles of $\mathcal T(\Theta)$ from Corollary \ref{Corollary_uniform_asymptotics} and from the Gaussian approximations based on Proposition \ref{Proposition_Transition_to_Gauss}.}
\label{Fig_Transition distributions}
\end{figure}

\subsection{Discussion of the definition} \label{subsection_GT_discussion}
The term ``Airy'' in the name $\mathcal G(w)$ refers to the Airy point process, whose points $\aa_i$ appear in its definition. The term ``Green'' stands from the tradition in random matrix theory to refer to matrix elements of the resolvent $(z I-D)^{-1}$ of a symmetric matrix $D$ as the Green's function. Via eigenvalue decomposition, the $(1,1)$ matrix element of $(z I -D)^{-1}$ is $\sum_i \frac{u_{1i}^2}{z-d_i}$, where $u_{1i}$ is the first coordinate of the $i$th normalized eigenvector of $D$ corresponding to the eigenvalue $d_i$, making it reminiscent of the sum in \eqref{eq_G_definition}.

The term ``transition'' in $\mathcal T(\Theta)$ refers to its role in capturing the transition between subcritical $\theta<\theta^c$ and supercritical $\theta>\theta^c$ behavior in \eqref{eq_spiked_Wigner_Gaussian}, \eqref{eq_spiked_covariance_as}, \eqref{eq_factors_as}, \eqref{eq_CCA_as}. This phenomenon is commonly known as the BBP phase transition, following \cite{baik2005phase}.

There are two other approaches to the transition process $\mathcal T(\Theta)$ in the literature. One is based on the limit of tridiagonal matrix model: \citet{bloemendal2013limits,lamarre2019edge} construct $\mathcal T(\Theta)$ as the largest eigenvalue of the Stochastic Airy Operator with $\Theta$--dependent boundary condition. Another approach, developed by \citet{mo2012rank} using the framework of Pfaffian point processes, provides an integral representation for the one-dimensional marginal distribution of $\mathcal T(\Theta)$.



An advantage of our definition via the Airy--Green function is its robustness. Proving convergence to either of the two alternative definitions, requires finding delicate algebraic structures (tridiagonalization or Pfaffians) in the prelimit objects, which are not known in some cases (e.g., CCA). In contrast, our approach relies only on identifying the eigenvalues of a spiked model as solutions to an equation, that can be obtained in all spiked models via finite-rank perturbation theory.

\begin{remark}
One can go beyond real matrices, and deal with complex, quaternionic, or even general $\beta$ random matrix ensembles. In the latter setting the definition of the Airy--Green function should be extended to
\begin{equation}
\label{eq_G_def_beta}
 \mathcal G_\beta(w)=\lim_{x\to-\infty} \left[\left(\sum_{j:\, \aa_{j,\beta}>x} \frac{\beta^{-1} \xi_{j,\beta}^2}{w- \aa_{j,\beta}}\right) - \frac{2}{\pi}\sqrt{-x} \right],
\end{equation}
where for $\beta>0$, $(\aa_{j,\beta})_{j=1}^{\infty}$ are the points of the Airy$_\beta$ point process (see e.g., \citet{ramirez2011beta}) and $\xi_{j,\beta}^2$ are i.i.d.~chi--squared random variables with $\beta$ degrees of freedom, defined as Gamma-distributions for general $\beta$. For $\beta=1$ we are back to \eqref{eq_G_definition}. For $\beta=2,4$, $\mathcal G_\beta(w)$ and the corresponding transition function, defined as in Definition \ref{Definition_Transition_function}, play the same role as $\mathcal G_1(w)$ in the signal plus noise models for complex and quaternionic matrices respectively.
\end{remark}

\subsection{Universal asymptotics for spiked models}
\label{Section_convergence_to_transitional}

The next theorem presents the asymptotics of the largest eigenvalues for all signal plus noise models of Section \ref{Section_signal_plus_noise_models}.

\begin{theorem} \label{Theorem_main_convergence_statement}
 Consider any of the four models of Section \ref{Section_signal_plus_noise_models} with signal strengths $\theta_1>\dots>\theta_r$, with $\sigma^2=1$, and in the regime \eqref{eq_limit_regime_1}, \eqref{eq_limit_regime_2}, or \eqref{eq_limit_regime_3}. Fix an index $1\le q \le r$ and suppose that as $N\to\infty$:
 \begin{enumerate}
  \item $\theta_1,\dots,\theta_{q-1}$ are fixed, distinct, and all larger than $\theta^c$.
  \item $\theta_q=\theta^c + N^{-1/3} \tilde \theta$ for a fixed $\tilde \theta\in\mathbb R$.
  \item $\theta_{q+1},\dots,\theta_r$ are fixed and all smaller than $\theta^c$.
 \end{enumerate}
 Then, in the sense of joint convergence in distribution,
 \begin{align}
 \label{eq_Gaussian_limit}  &\sqrt{N}(\lambda_i-\lambda(\theta_i)) \xrightarrow[]{d} \mathcal N(0, V(\theta_i)),    \qquad 1\le i \le q-1,\\
  \label{eq_Transition_limit}   &N^{2/3}(\lambda_q-\lambda_+) \xrightarrow[]{d} \kappa_1 \mathcal T(\kappa_2 \tilde \theta),
 \end{align}
 where the $q$ limiting random variables in \eqref{eq_Gaussian_limit}, \eqref{eq_Transition_limit} are jointly independent and the constants are as in \eqref{eq_Spiked_Wigner_params},\eqref{eq_Spiked_covariance_params},\eqref{eq_Factor_params},\eqref{eq_CCA_params} with
 \begin{equation}
 \label{eq_transition_parameters}
  \kappa_1=\frac{1}{2} \left[{\frac{[V'(\theta^c)]^2}{\lambda''(\theta^c)}}\right]^{\frac13}, \qquad \kappa_2= \left[{\frac{[\lambda''(\theta^c)]^2}{V'(\theta^c)}}\right]^{\frac13}~.
 \end{equation}
If no signal strengths are close to $\theta^c$, then the same limits hold  without the \eqref{eq_Transition_limit} part.
\end{theorem}
\begin{remark}
 While the distributional limit of $\lambda_{q+1}, \dots, \lambda_{r}$ can be computed,
 it is of no use for the confidence intervals: the limits depend on $\tilde \theta$, but not on $\theta_{q+1},\dots,\theta_r$.
\end{remark}

Note that the two limit regimes \eqref{eq_Gaussian_limit} and \eqref{eq_Transition_limit} heuristically agree with each other:
if one sets $\tilde \theta = \eps N^{1/3}$ with a small $\eps>0$, then using \eqref{eq_Transition_limit} and Proposition \ref{Proposition_Transition_to_Gauss}, we expect
$$
\lambda_q \approx \lambda_++ N^{-2/3} \kappa_1 \mathcal T(\kappa_2 \tilde \theta)\approx \lambda_+ + \kappa_1 \kappa_2^2  \eps^2 + 2N^{-1/2} \kappa_1 \sqrt{\eps \kappa_2} \mathcal N(0,1).
$$
If one sets $\theta_i=\theta^c+\eps$, then Taylor expanding \eqref{eq_Gaussian_limit} (noting $V(\theta^c)=\lambda'(\theta^c)=0$), we expect
$$
 \lambda_i\approx \lambda_+ + \frac{\eps^2}{2} \lambda''(\theta^c) + N^{-1/2} \sqrt{\eps V'(\theta^c)} \mathcal N(0,1).
$$
Using \eqref{eq_transition_parameters}, we see that the last two asymptotic expansions are the same. In parallel, using Proposition \ref{Proposition_Transition_to_Gauss} we can combine two asymptotic regimes of Theorem \ref{Theorem_main_convergence_statement} into one (among several asymptotically equivalent formulas, we chose the one with the best finite sample performance):

\begin{corollary} \label{Corollary_uniform_asymptotics}
 The asymptotics \eqref{eq_Gaussian_limit} and \eqref{eq_Transition_limit} can be written in unified form as:
\begin{equation}
 \label{eq_Unified_asymptotics}
 \lambda_i  \approx
 \lambda(\theta_i)-\frac{ \kappa_2^{3/2}}{2} \sqrt{V(\theta_i) (\theta_i-\theta^c)^3}+ \frac{\kappa_2^{-1/2}}{2 N^{2/3}} \sqrt{\frac{V(\theta_i)}{\theta_i-\theta^c}} \mathcal T \Bigl(\kappa_2 N^{1/3} (\theta_i-\theta^c) \Bigr)+\frac{\kappa_3}{N},
\end{equation}
where the error is $o\bigl(N^{-2/3} + N^{-1/2}(V(\theta_i))^{1/2}\bigr)$ for $\theta_i>\theta^c$. For $\theta_i\le \theta^c$, one instead uses
\begin{equation}
\label{eq_x83}
  \lambda_i\approx \lambda_++\frac{\kappa_1}{N^{2/3}} \mathcal T\left(\kappa_2 \tilde \theta\right)+\frac{\kappa_3}{N}.
\end{equation}
In \eqref{eq_Unified_asymptotics} and \eqref{eq_x83} we use $\kappa_1=\frac{1}{2} {\frac{[V'(\theta^c)]^{2/3}}{[\lambda''(\theta^c)]^{1/3}}}$, $\kappa_2= \frac{[\lambda''(\theta^c)]^{2/3}}{[V'(\theta^c)]^{1/3}}$, and $\kappa_3=-\frac{3}{2}\frac{ \kappa_1}{\kappa_2 \theta^c}$.
\end{corollary}
The formula \eqref{eq_x83} is a direct corollary of \eqref{eq_Transition_limit}, while \eqref{eq_Unified_asymptotics} combines \eqref{eq_Gaussian_limit}  and \eqref{eq_Transition_limit} together. Indeed, when $\theta_i$ is bounded away from $\theta^c$, \eqref{eq_Transition_to_Gauss} converts \eqref{eq_Unified_asymptotics} into
$$
 \lambda(\theta_i)-\frac{ \kappa_2^{3/2}}{2} \sqrt{V(\theta_i) (\theta_i-\theta^c)^3}+ \frac{\kappa_2^{-1/2}}{2} \sqrt{\frac{V(\theta_i)}{\theta_i-\theta^c}}\kappa_2^2  (\theta_i-\theta^c)^2 + \frac{\kappa_2^{-1/2}}{N^{1/2}} \sqrt{\frac{V(\theta_i)}{\theta_i-\theta^c}} \sqrt{\kappa_2(\theta_i-\theta^c)} \mathcal{N}(0,1),
$$
which is readily seen to be equivalent to
\eqref{eq_Gaussian_limit}. When $\theta_i$ is close to $\theta^c$, $\theta_i=\theta^c + N^{-1/3} \tilde \theta$, Taylor expanding $\lambda(\cdot)$ and $V(\cdot)$ near $\theta^c$, \eqref{eq_Unified_asymptotics} turns into an equivalent form of \eqref{eq_Transition_limit}:
$$
\lambda_++\frac{\lambda''(\theta^c)}{2}(\theta_i-\theta^c)^2-\sqrt{V'(\theta_c)}\frac{ \kappa_2^{3/2}}{2} (\theta_i-\theta^c)^2+ \frac{\kappa_2^{-1/2}}{2 N^{2/3}} \sqrt{V'(\theta^c)} \mathcal T \Bigl(\kappa_2 N^{1/3} (\theta_i-\theta^c) \Bigr),
$$

Since $\kappa_3/N=o(N^{-2/3})$, and therefore the choice of $\kappa_3$ does not affect the validity of the asymptotic formulas \eqref{eq_Unified_asymptotics} and \eqref{eq_x83}. These terms are introduced to improve the performance of the formulas for intermediate values of $N$, cf.\ \citet{Johnstone_Jacobi,ma2012accuracy,johnstone2012fast}, which emphasize the importance of $1/N$ corrections for the practical applicability. The reasoning behind our choice of $\kappa_3$ is as follows. First, we require continuity at $\theta^c$; hence, \eqref{eq_Unified_asymptotics} and \eqref{eq_x83} use the same $\kappa_3/N$. Second, we leverage additional information available at $q = r = 1$ and $\tilde \theta = -N^{1/3}\theta^c$ for the spiked Wigner, factor, and CCA models.  (For the spiked covariance model, one instead takes $\tilde \theta = -N^{1/3}\gamma$ and adjusts the formula accordingly.) On one hand, combining \eqref{eq_x83} with the asymptotic approximation \eqref{eq_large_negative_Theta}, we obtain
\begin{equation}
\label{eq_x84}
  \lambda_1\approx \lambda_++\frac{\kappa_1}{N^{2/3}} \left(\aa_1+\frac{1}{\kappa_2 N^{1/3}\theta^c}\right) +\frac{\kappa_3}{N}.
\end{equation}
On the other hand, $\tilde \theta= - N^{1/3}\theta^c$ corresponds to $\theta=0$ (or $\theta = 1$ for the spiked covariance model with $\tilde \theta = -N^{1/3}\gamma$), meaning that in all four models, we are in the unspiked regime without a signal component. In this setting, the convergence of $\lambda_1$ to the Tracy–Widom distribution $\aa_1$ is well established, and $1/N$-order asymptotic corrections have been studied in \citet{Johnstone_Jacobi,ma2012accuracy,johnstone2012fast}. From these works, one can extract
\begin{equation}
\label{eq_x85}
 \lambda_1\approx  \lambda_+ +\frac{\kappa_1}{N^{2/3}} \aa_1 -\frac{1}{2 N} \times \begin{cases} 1 & \text{ for spiked Wigner},\\ \gamma(1+\gamma)^2&\text{ for spiked covariance and factors},\\ \frac{(\sqrt{\tau_N-1}\sqrt{\tau_M-1}-1)^{2}(\sqrt{\tau_N-1}+\sqrt{\tau_M-1})^{2}}{ \tau_N^{2}\tau_M \sqrt{\tau_N-1}\sqrt{\tau_M-1}} &\text{ for CCA.}\end{cases}
\end{equation}

Equating \eqref{eq_x84} with \eqref{eq_x85} yields the formula for $\kappa_3$, as recorded in Table \ref{Table_parameters}. An interesting observation is that in each case the term $\frac{\kappa_1}{\kappa_2 \theta^c}$ in \eqref{eq_x84} is twice the $\frac{1}{N}$ correction term in \eqref{eq_x85}, which leads to the $\frac{3}{2}$ coefficient appearing in $\kappa_3$ across all four models.

\begin{remark}
 We expect that \eqref{eq_Unified_asymptotics} also remains asymptotically valid on all mesoscropic scales, i.e., when  $\theta_i=\theta^c + N^{-\alpha} \tilde \theta$, $0<\alpha<1/3$. We omit a detailed proof.
\end{remark}

\bigskip

The proof of Theorem \ref{Theorem_main_convergence_statement} in Section \ref{Section_proof_of asymptotic_approximations} begins with a rank-one perturbation equation, which expresses the eigenvalues in a signal plus noise model with $r$ spikes (the ``target model'') as solutions to an algebraic equation involving a simpler model with $r - 1$ spikes (the ``base model''). 
Analyzing the asymptotic behavior of this equation leads to the following conclusion, stated informally below:
\begin{itemize}
\item If the strength of the added spike is subcritical, $\theta<\theta^c$, then the largest eigenvalues in the target model are very close
to the largest eigenvalues in the base model.
\item If the strength of the added spike is supercritical, $\theta>\theta^c$, then the largest eigenvalues in the target model are very close to the
largest eigenvalues for the base model, except for one additional eigenvalue for the target model, which is close to $\lambda(\theta)$.
\item If the strength of the added spike is critical, $\theta = \theta^c+ N^{-1/3} \tilde \theta$, then in the target model eigenvalues
which are (macrosopically) larger than $\lambda_+$ are very close to the eigenvalues in the base model. Near $\lambda_+$ the equations rescale to $\mathcal G(w)=-\kappa_2\tilde \theta$, where $\{\aa_j\}$ in the definition of $\mathcal G(w)$ arise as limits of the eigenvalues in the base model, and the eigenvalues in the target model converge to the roots of this equation.
\end{itemize}

On the technical level, the key novelty is in our ability to handle the most delicate case, when the spike is critical. If all spikes are subcritical or supercritical, the arguments are much simpler and follow ideas similar to those found in the references cited in Section \ref{Section_signal_plus_noise_models}. Some special cases of Theorem \ref{Theorem_main_convergence_statement} can be handled by other methods, for example, the $r=1$ case for the spiked Wigner and spiked covariance models is addressed in \citet{mo2012rank,bloemendal2013limits}; see also \citet{bloemendal2016limits,lamarre2019edge}. However, we believe that the level of generality achieved here -- particularly our treatment of the factor model and CCA -- was not previously available in the literature and is beyond the reach of those alternative methods.





\begin{remark}
We expect that our methods can be extended to handle the case of multiple ($k > 1$) critical spikes, as well as the remaining largest eigenvalues $\lambda_{q+1}, \lambda_{q+2}, \dots$. The limiting behavior should be described by a higher-rank Airy point process, which we define recursively. The rank $0$ process is the classical Airy point process $\{\aa_j\}$. The rank $1$ process $\{\aa_j^{(\Theta)}\}$ consists of all real solutions to the equation $\mathcal G(w) = -\Theta$; in particular, the largest point $\aa_1^{(\Theta)}$ coincides with $\mathcal T(\Theta)$ from Definition \ref{Definition_Transition_function}. We then iterate this construction: given the rank $k$ point process $\{\aa_j^{(\Theta_1,\dots,\Theta_k)}\}$ depending on $k$ real parameters $\Theta_1,\dots,\Theta_k$, we define the rank $(k+1)$ process $\{\aa_j^{(\Theta_1,\dots,\Theta_k,\Theta_{k+1})}\}$ as the set of all real solutions to the equation
\begin{equation}
  \lim_{x\to-\infty} \left[\left(\sum_{j:\, \aa_j^{(\Theta_1,\Theta_2,\dots,\Theta_k)}>x} \frac{[\xi^{(k)}_j]^2}{w-\aa_j^{(\Theta_1,\Theta_2,\dots,\Theta_k)}}\right) - \frac{2}{\pi}\sqrt{-x} \right]=-\Theta_{k+1},
\end{equation}
where $\xi^{(k)}_j$, $j=1,2,\dots$ are $\mathcal N(0,1)$, independent over $j$ and $k$.

We anticipate that in a signal plus noise model with $k$ critical spikes the eigenvalues near $\lambda_+$ converge, after the same recentering and rescaling as in Theorem \ref{Theorem_main_convergence_statement}, to the points $\{\aa_j^{(\tilde \Theta_1,\tilde \Theta_2,\dots,\tilde \Theta_k)}\}$. A different construction is provided in \citet{bloemendal2016limits}, but it is ultimately expected to yield the same point process $\bigl\{\aa_j^{(\tilde \Theta_1,\tilde \Theta_2,\dots,\tilde \Theta_k)}\bigr\}_{j=1}^{\infty}$.
\end{remark}

\section{Empirical illustrations}
\label{Section_emprical}

We present three examples that illustrate the application of the procedure described in Section \ref{Section_confidence_intervals} to empirical data sets. In each case, the first step reveals a strong agreement between the histogram of eigenvalues and the corresponding theoretical curve—specifically, the Marchenko-Pastur law in the first two examples (factor models) and the Wachter law in the third (CCA). This suggests that the data aligns well with our modeling assumptions.

\subsection{Industrial Production}

Industrial production (IP) accounts for more than $10\%$ of the United States' Gross Domestic Product (GDP), making it a significant component of total U.S.~output. In this subsection we use data from \citet{andreou2019inference}, which investigates whether IP constitutes a dominant factor in U.S.~economic activity. The data set contains quarterly IP growth rates across 117 sectors, spanning the period from $1977:Q1$ to $2011:Q4$. We de-mean the data and standardize each sector to have unit sample variance, thereby working with the sample correlation matrix of IP.

Figure \ref{Fig_IP} shows all eigenvalues of the standardized IP and highlights four of them ($3.75,\, 4.26,\, 6.12,\, 30.05$) that lie to the right of the theoretical Marchenko–Pastur upper edge, $\lambda_+ = 3.68$. The largest eigenvalue, $30.05$, stands out markedly and represents a strong ``market'' factor. The two smallest among the four, being close to the cutoff, may reflect spurious signals arising from noise. To assess their significance, we construct $95\%$ confidence intervals. These intervals are represented by vertical segments in Figure \ref{Fig_IP} and are also summarized in Table \ref{Table_IP_CI}.

Notably, the interval for the fourth largest eigenvalue at $3.86$ differs substantially from what one would obtain using a Gaussian approximation (discussed in Section \ref{Section_spiked_Wigner} and at the end of Section \ref{Section_factor}), highlighting the importance of our new procedure. This interval intersects the critical identification threshold $\theta^c = 0.92$, indicating that we cannot reject the null hypothesis that it represents noise (or a non-informative signal). For the remaining eigenvalues, the null is rejected. The confidence interval for the largest eigenvalue is nearly identical under our method and the Gaussian approximation, whereas the differences grow as the eigenvalues and corresponding signal strengths decrease.

The discussion in the previous paragraph leads to the conclusion that the IP growth rate is driven by three factors: one strong factor and two weak factors. It is instructive to compare this result with the classical information criterion of \citet{bai2002determining}, one of the most commonly used methods in applied research and one that is based on strong-signal asymptotics. Using $IC_{p_2}$ from
\citet[Section 5]{bai2002determining}, we identify only the largest market factor. If we instead use $IC_{p_1}$ from the same article, which imposes a smaller penalty on the number of factors, we select the two largest factors. As a result, informative weak factors are missed by these procedures, underscoring the importance of accounting for weak signals and motivating our proposed methodology.

\begin{figure}[t]
\begin{subfigure}{.49\textwidth}
  \centering
  \includegraphics[width=1.0\linewidth]{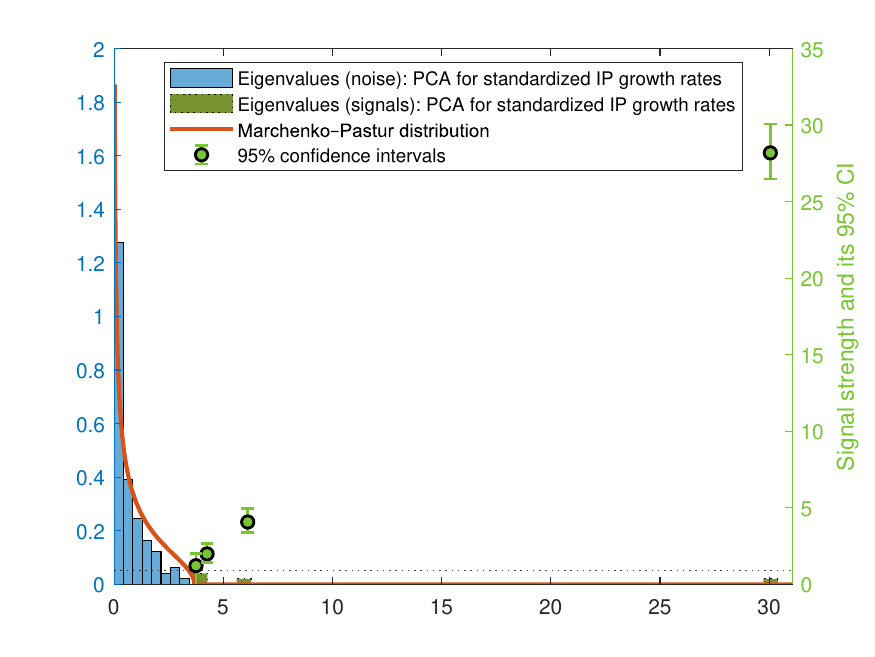}
  \caption{All eigenvalues.}
  \label{Fig_IP1}
\end{subfigure}%
\begin{subfigure}{.49\textwidth}
  \centering
  \includegraphics[width=1.0\linewidth]{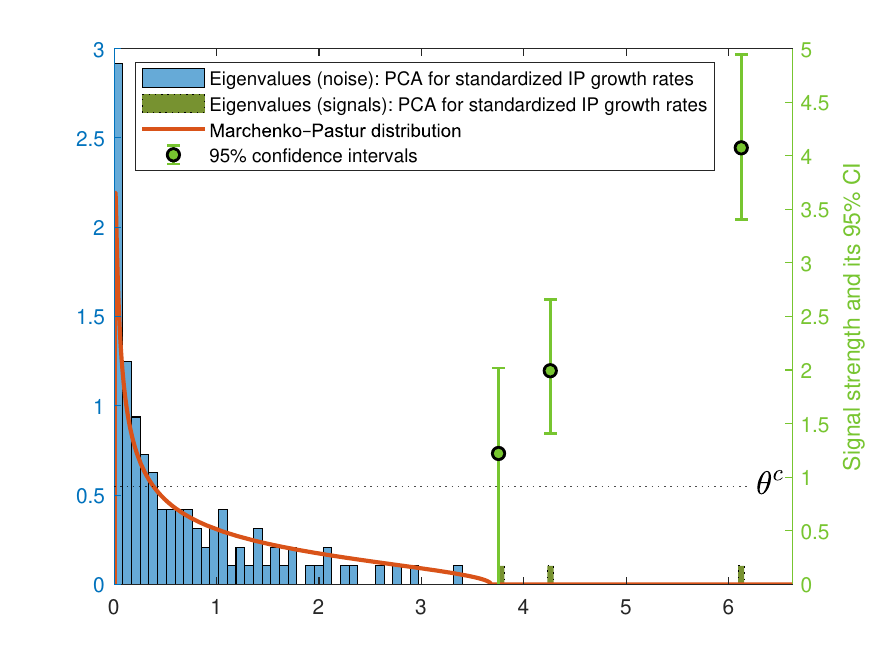}
  \caption{All except the largest (``market'') signal.}
  \label{Fig_IP2}
\end{subfigure}
\caption{IP sample correlation eigenvalues: signals, their $95\%$ confidence intervals, noise, and Marchenko-Pastur distribution with $N=117$, $S=139$.}
\label{Fig_IP}
\end{figure}

\begin{table}[t]
\centering
\begin{tabular}{c|c|c|c|c}
& & \multicolumn{3}{c}{Confidence Intervals}\\
\hline
\hline
$\lambda_i$ & $\theta_i$ & Our method & Gaussian method $\sharp1$ & Gaussian method $\sharp2$ \\
\hline
\hline
30.05 & 28.18 &[26.48\, 30.07] & [26.39\, 29.98] & [26.44\, 30.03]\\
6.12 & 4.07 &[3.41\, 4.94] & [3.31\, 4.84] & [3.37\, 4.90]\\
4.26 & 1.99 &[1.41\, 2.65] & [1.35\, 2.63] & [1.44\, 2.70] \\
3.75 & 1.22 &[0\, 2.02] & [0.48\, 1.96] &  [0.93\, 1.99]\\
\hline
\hline
\end{tabular}
 \caption{IP: Confidence intervals based on different methods.\label{Table_IP_CI}}
\end{table}

\subsection{S\&P100}

Analyzing stock returns is essential for understanding market dynamics, evaluating investment performance, and guiding both individual and institutional investment strategies. A key statistical object in this context is the covariance matrix of stock returns, which plays a central role in portfolio optimization, such as in the Markowitz mean–variance framework. The vast ``factor zoo'' -- the large number of potential variables proposed to explain stock returns -- highlights the practical challenge of distinguishing meaningful factors from noise in high-dimensional settings (see \citet{cochrane2011presidential} for an influential discussion). Here we demonstrate how our methodology can be applied to the sample covariance matrix of weekly S$\&$P$100$ stock returns. We use data from \citet{BG1}, which covers 92 stocks over the period from January 1, 2010, to January 1, 2020.

Before turning to the empirical spectrum, it is important to address a potential concern regarding the applicability of PCA-based methods in financial data. Financial returns may exhibit heavy tails, complicating factor identification via PCA. For pure-noise covariance matrices, a sharp transition occurs at tail index $\alpha = 4$ \citep{yin1988limit,bai1988note}: lighter tails keep the largest eigenvalue at the Marchenko–Pastur edge, while heavier tails can produce spurious spikes. Although short-horizon returns often have $\alpha \approx 3$, weekly and longer-horizon returns typically have $\alpha\geq  5$ \citep{gopikrishnan1999scaling,gabaix2009power,fan2017elements}, as aggregation and the central limit theorem dampen extremes. This places weekly returns in a regime where spurious tail-driven spikes are less likely to arise.

Figure \ref{Fig_SP} presents the full spectrum of the sample covariance matrix and identifies ten eigenvalues, $10^{-4}(365,\, 52.4,\, 41.7,\, 30.7,\, 26.1,\, 20.7,\, 17.6,\, 15.1,\, 13.9,\, 13.6)$, that lie to the right of the theoretical Marchenko–Pastur upper edge, $\lambda_+ = 13\times10^{-4}$. To better fit the empirical data, we adopt an effective parameter value $\gamma^2 = 0.4$, in contrast to the true value $N/S = 0.18$. This adjustment may reflect temporal dependence in the data, which effectively reduces the sample size and increases $\gamma^2 = 0.4$. We also set $\sigma = 0.02$ to align the overall variance of the eigenvalues in the data.

\begin{figure}[t]
\begin{subfigure}{.49\textwidth}
  \centering
  \includegraphics[width=1.0\linewidth]{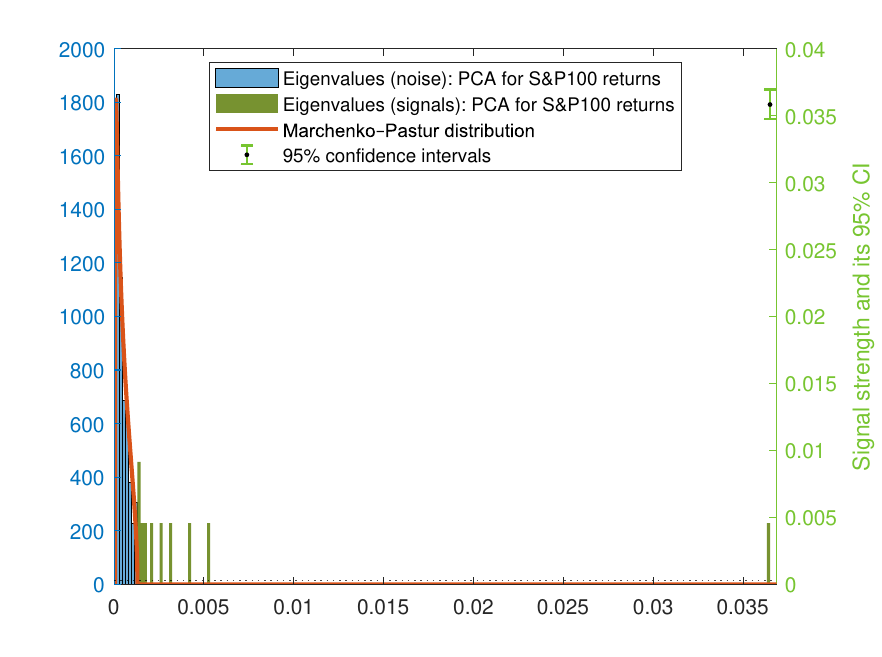}
  \caption{All eigenvalues.}
  \label{Fig_SP1}
\end{subfigure}%
\begin{subfigure}{.49\textwidth}
  \centering
  \includegraphics[width=1.0\linewidth]{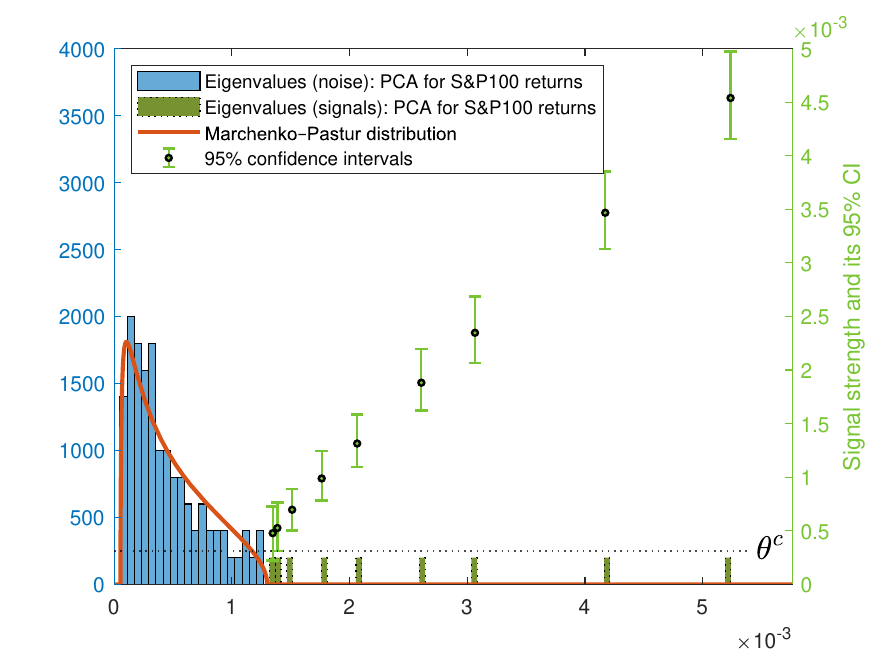}
  \caption{All except the largest (``market'') signal.}
  \label{Fig_SP2}
\end{subfigure}
\caption{S$\&$P$100$ sample covariance: signals, their $95\%$ confidence intervals, noise, and Marchenko-Pastur distribution with $N=92$, $\gamma^2=0.4$, $\sigma=0.02$.}
\label{Fig_SP}
\end{figure}

\begin{table}[t]
\centering
\begin{tabular}{c|c|c|c|c}
& & \multicolumn{3}{c}{Confidence Intervals $\times10^4$}\\
\hline
\hline
$10^4\lambda_i$ & $10^4\theta_i$ & Our method & Gaussian method $\sharp1$ & Gaussian method $\sharp2$ \\
\hline
\hline
365 & 358 &[347\, 369] & [347\, 369] & [347\, 369]\\
52.4 & 45 &[41.5\, 49.7] & [41.3\, 49.5] & [41.4\, 49.7]\\
41.7 & 35 &[31.3\, 38.5] & [31.0\, 38.3] & [31.2\, 38.5] \\
30.7 & 23 &[20.7\, 26.9] & [20.4\, 26.5] &  [20.6\, 26.7]\\
26.1 & 19 & [16.2\, 21.9] & [16.0\, 21.6] & [16.2\, 21.8] \\
20.7 & 13 & [10.9\, 15.8] & [10.7\, 15.6] & [10.9\, 15.8] \\
17.6 & 10 & [7.84\, 12.4] & [7.63\, 12.1] & [7.83\, 12.3] \\
15.1 & 7.0 & [5.04\, 8.90] & [4.86\, 9.07] & [5.10\, 9.27] \\
13.9 & 5.3 & [3.11\, 7.65] & [3.12\, 7.39] & [3.63\, 7.59] \\
13.6 & 4.8 & [2.24\, 7.26] & [2.57\, 6.98] & [3.35\, 7.17] \\
\hline
\hline
\end{tabular}
 \caption{S\&P100: Confidence intervals $\times10^4$ based on different methods.\label{Table_SP100_CI}}
\end{table}

Figure \ref{Fig_SP} and Table \ref{Table_SP100_CI} report the 95$\%$ confidence intervals for ten candidate signals. As in the previous example, the largest eigenvalue is much larger than the others and corresponds to the ``market'' factor. The two smallest eigenvalues among them yield intervals that intersect the identification threshold $\theta^c = 3.11\times 10^{-4}$, indicating that they cannot be statistically distinguished from being non-informative. This is also the region in which the two Gaussian approximations produce markedly different intervals, reflecting the fact that the variance term $V(\theta_i)$ is very close to zero. We therefore conclude that only eight of the ten observed spikes represent informative signals.  In contrast, applying the information criterion $IC_{p_2}$ of \citet{bai2002determining} selects only four of these eight factors, while $IC_{p_1}$ identifies five; both criteria miss many of the weaker factors.

\medskip

\subsection{Cyclical vs.~non-cyclical stocks}

\begin{figure}[t]
  \centering
  \includegraphics[width=0.6\linewidth]{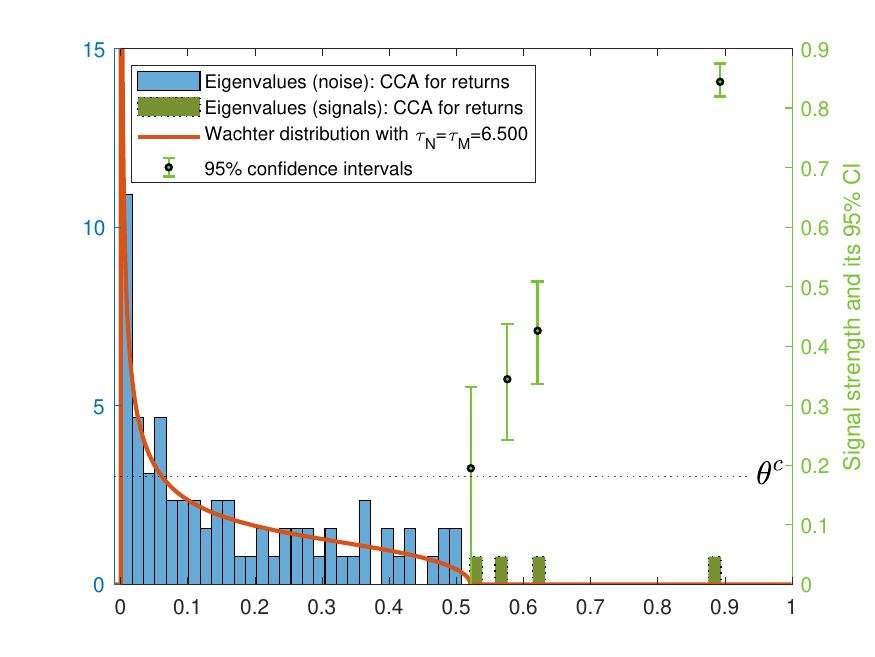}
\caption{Squared sample canonical correlations between cyclical and noncyclical stocks: signals, their $95\%$ confidence intervals, noise, and Wachter distribution with $N=M=80$, $S=520$.}
\label{Fig_cycl_nocycl}
\end{figure}

Financial stocks are typically classified into cyclical and non-cyclical (defensive) categories, depending on whether their performance tracks economic business cycles. These groups are generally assumed to be uncorrelated, aside from exposure to a common ``market'' factor. \citet{BG_CCA} identify three non-zero canonical correlations between these two groups, suggesting the presence of three common factors.
Here we revisit their analysis using the same data set to assess whether these observed correlations reflect genuine signals or could instead be attributed to noise.

The data set comprises weekly returns for $80$ cyclical and $80$ defensive stocks, spanning the period from January 1, 2010, to January 1, 2020. Figure \ref{Fig_cycl_nocycl} reproduces the canonical correlations reported by \citet{BG_CCA} and augments it with $95\%$ confidence intervals. \citet{BG_CCA} used the signal strengths to estimate the angles between true and estimated canonical variables. By incorporating our results, one can generate confidence intervals for these angles.

As shown in Figure \ref{Fig_cycl_nocycl}, the $95\%$ confidence intervals for three largest canonical correlations, $0.58$, $0.62$, and $0.89$, lie above the cutoff $\theta^c=0.18$, confirming them as true signals. In contrast, the confidence interval for the fourth largest value intersects the cutoff, indicating that it cannot be reliably distinguished from noise and is therefore classified as a non-informative component.

A comparison of Figures \ref{Fig_IP}, \ref{Fig_SP}, and \ref{Fig_cycl_nocycl} reveals an interesting pattern: in the factor models of the first two figures, the confidence intervals widen as the signal strengths increase, whereas in the CCA setting, the intervals become narrower as the signals approach 1. Theoretically, this behavior in CCA can be attributed to the factor $(1-\theta)^2$ in $V(\theta)$, as shown in Table \ref{Table_parameters}.

\section{Extensions} \label{Section_extensions}

In this section we discuss possible extensions of our results, focusing on non-Gaussian data and broader classes of models than those considered in Section \ref{Section_signal_plus_noise_models}.

\subsection{Non-Gaussian noise} \label{Section_extension_non_Gauss}

The four models in Section \ref{Section_signal_plus_noise_models} are based on Gaussian noise matrices. A natural generalization is to replace the Gaussian vectors with more general random vectors having the same mean and covariance. It is well known (see \citet{lee2014necessary,ding2018necessary,FanYang} and the broader reviews \citet{deift2009random,tao2012random,erdos2017dynamical}) that for pure noise models without signals, the distribution of the largest eigenvalues remains unchanged in many non-Gaussian settings. This raises the  question of whether the robustness extends to the confidence intervals of Section \ref{Section_confidence_intervals}.

\begin{figure}[t]
   \begin{subfigure}{.49\textwidth}
     \includegraphics[width=\linewidth]{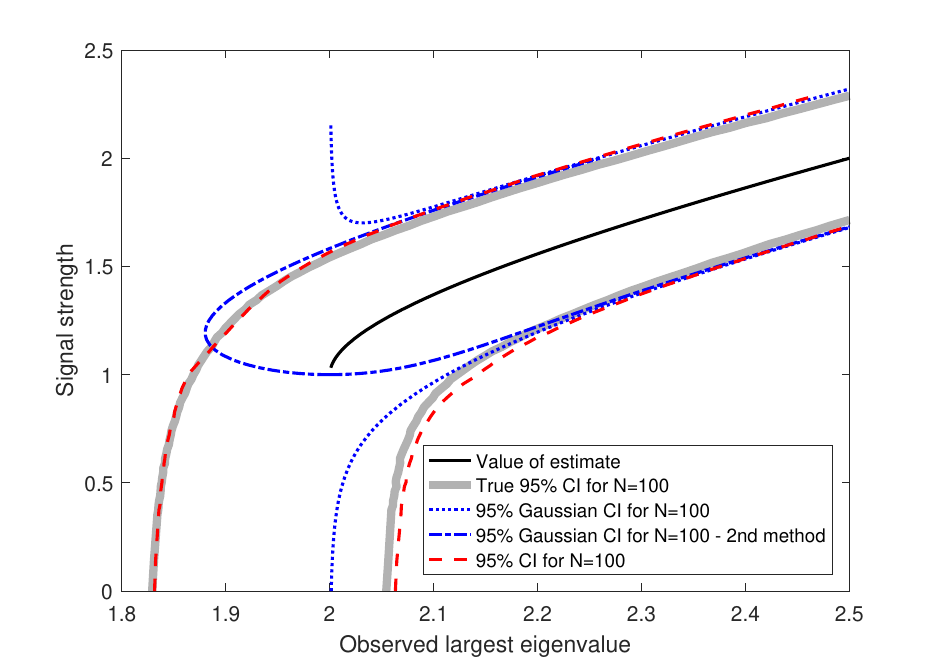}
     \caption{Uniform noise, localized signal}
   \end{subfigure}
   \begin{subfigure}{.49\textwidth}
     \includegraphics[width=\linewidth]{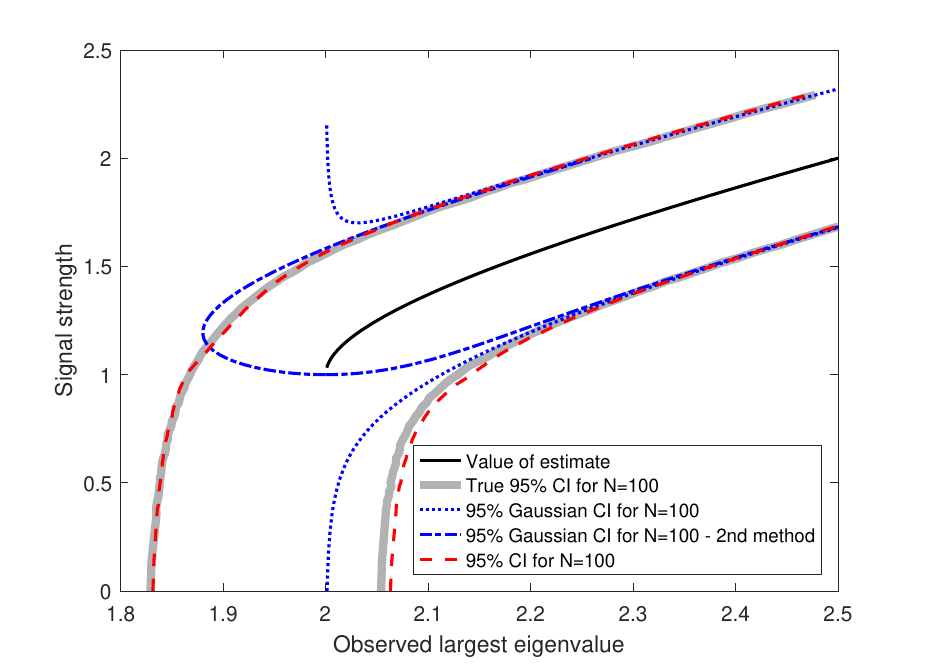}
     \caption{Uniform noise, delocalized signal}
   \end{subfigure}
   \begin{subfigure}{.49\textwidth}
     \includegraphics[width=\linewidth]{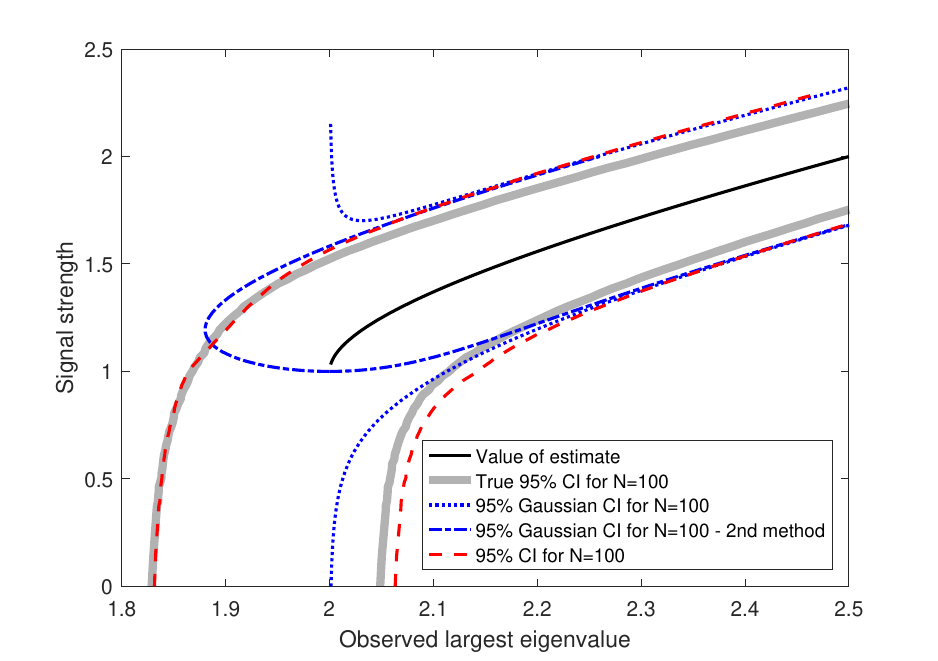}
     \caption{Bernoulli $p=\tfrac{1}{2}$ noise, localized signal}
   \end{subfigure}
   \begin{subfigure}{.49\textwidth}
     \includegraphics[width=\linewidth]{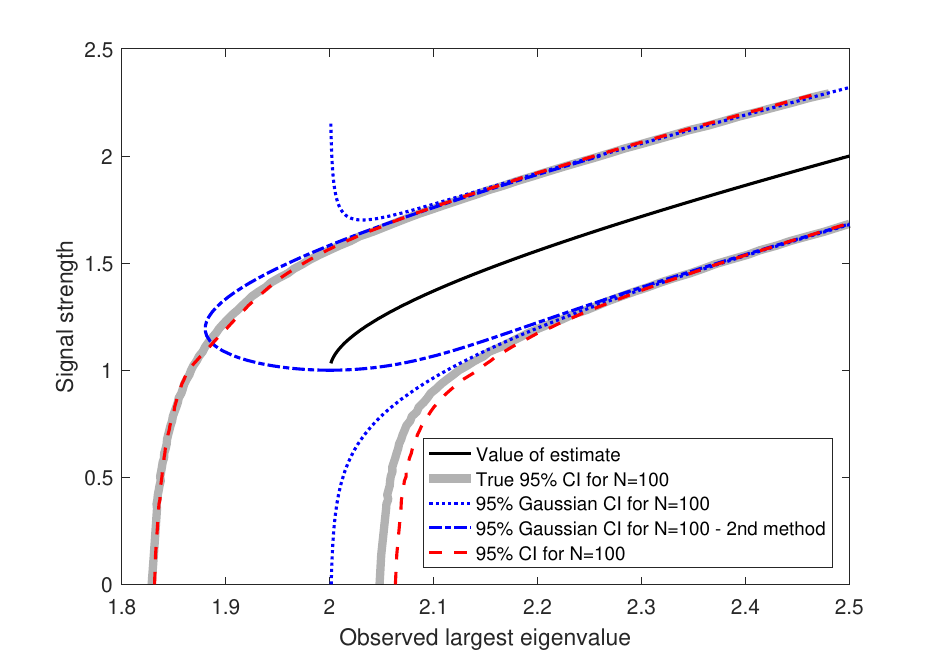}
     \caption{Bernoulli $p=\tfrac{1}{2}$ noise, delocalized signal}
   \end{subfigure}
   \begin{subfigure}{.49\textwidth}
     \includegraphics[width=\linewidth]{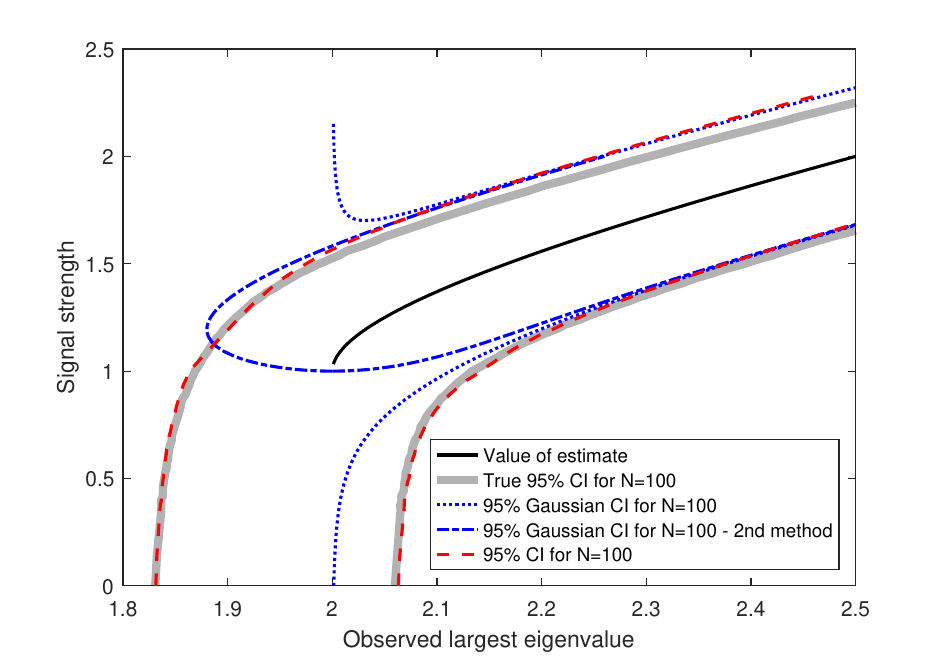}
     \caption{$\mathrm{Binomial}(2,0.3)$ noise, localized signal}
   \end{subfigure}
   \begin{subfigure}{.49\textwidth}
     \includegraphics[width=\linewidth]{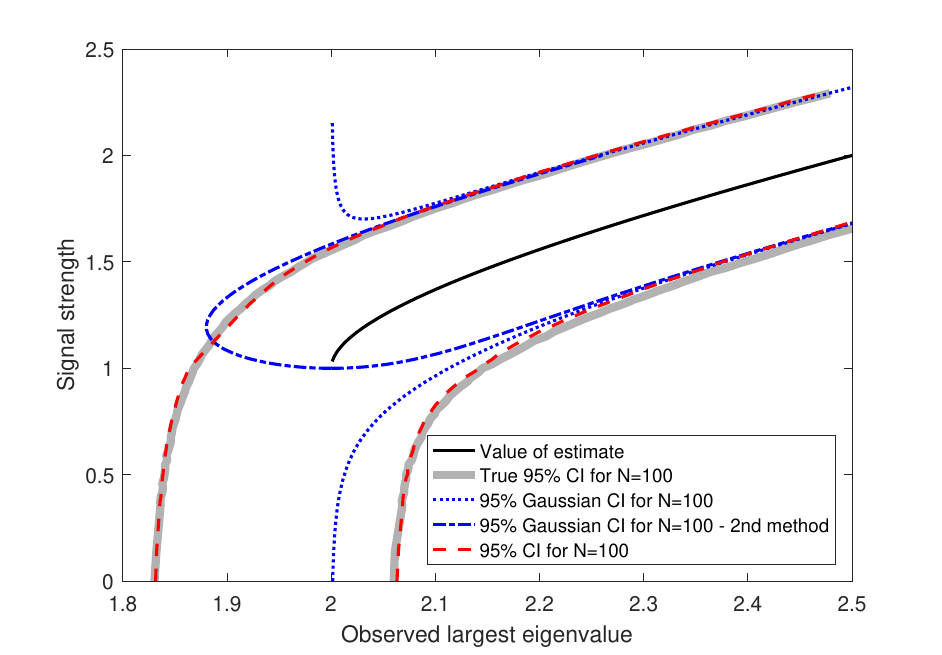}
     \caption{$\mathrm{Binomial}(2,0.3)$ noise, delocalized signal}
   \end{subfigure}
\caption{Confidence intervals for Wigner matrices with non-Gaussian noise.}
\label{Fig_non_Gauss}
\end{figure}

Figure \ref{Fig_non_Gauss} shows the results of Monte Carlo simulations of confidence intervals for Wigner matrices with a single spike and non-Gaussian noise. These can be compared to the Gaussian case in Figure \ref{Fig_confidence_Wigner}. The thick gray line depicts sample confidence intervals based on $10^5$ simulations of $100 \times 100$ matrices. The noise matrix $\mathcal E$ in \eqref{eq_Spiked_Wigner} is still defined as $\mathcal E = \frac{1}{\sqrt{2N}}(\mathcal Z + \mathcal Z^\T)$, with $\mathcal Z$ having i.i.d.~entries. We vary the distribution of these entries, each rescaled to have mean $0$ and variance $1$, to be: (i) uniform on $[0,1]$, (ii) Bernoulli with success probability $p=1/2$, or (iii) Binomial with parameters $n=2$, $p=0.3$. We also consider two signal vectors: a localized signal $\u^* = (1,0,\dots,0)^\T$ and a delocalized signal $\u^* = \tfrac{1}{\sqrt{N}}(1,1,\dots,1)^\T$. While in the Gaussian setting the choice of $\u^*$ is irrelevant due to rotational invariance, this is no longer true for general noise distributions.

Figure \ref{Fig_non_Gauss} shows that, while the sample confidence intervals remain reasonably close to those constructed via the procedure in Section \ref{Section_confidence_intervals}, the agreement is notably better for delocalized signals. For the localized signal significant deviations appear at larger values of the observed eigenvalue $\lambda_1$. A heuristic explanation follows directly from the model $\A=\theta \cdot \u^* (\u^*)^\T + \mathcal E$ in \eqref{eq_Spiked_Wigner}. When $\u^* = (1,0,\dots,0)^\T$, the parameter $\theta$ enters $\A$ only through its sum with the $(1,1)$ entry of $\mathcal E$, so the distribution of that single element directly affects any estimate of $\theta$. By contrast, when $\u^* = \tfrac{1}{\sqrt{N}}(1,1,\dots,1)^\T$, the projection of the noise onto the signal direction aggregates many independent entries. Thus, by the central limit theorem, the influence of the individual noise distribution diminishes as $N$ grows, leading to behavior indistinguishable from the Gaussian case. Similar robustness is expected for other delocalized signals.

For spiked Wigner matrices with supercritical signals ($\theta > \theta^c$), non-Gaussian cases have been rigorously analyzed in several papers. For localized signals, non-trivial dependence of the asymptotics of $\lambda_1$ on the distribution of the noise has been established in \citet{capitaine2009largest,capitaine2012central,pizzo2013finite,knowles2013isotropic,knowles2014}. In contrast, universality of the limit for delocalized signals has been shown under various conditions (including different definitions of “delocalized”) in \citet{feral2007largest,benaych2011fluctuations,capitaine2012central,pizzo2013finite,renfrew2013finite,knowles2013isotropic,knowles2014}. Also \citet[Remark 2.18]{knowles2013isotropic} explain that even for localized signals the dependence of the limit on the noise distribution is washed out as $\theta$ approaches $\theta^c$. Similar phenomenology is expected (and in some cases proven) for other signal plus noise models.

This discussion, together with the simulation results, leads us to conjecture that for all signal plus noise models of interest, the confidence intervals constructed via the procedure in Section \ref{Section_confidence_intervals} remain good approximations even under non-Gaussian noise, particularly when the signal is delocalized. From a practical standpoint, this reinforces the robustness of our method, especially in light of \citet{giannone2021economic}, who argue that many economic data sets are non-sparse and should therefore be modeled using delocalized signals.

\subsection{General models and empirical distributions of eigenvalues} \label{Section_extension_limit_shapes}
In all four models from Section \ref{Section_signal_plus_noise_models} the empirical eigenvalue  distributions take specific parametric forms, as summarized in Table \ref{Table_LLN_shapes}. Although our estimation procedure relies only on the largest eigenvalues, the formulas in Table \ref{Table_parameters} were derived using empirical distributions in Table \ref{Table_LLN_shapes}.

In other signal plus noise models and some empirical data sets the eigenvalue histograms may differ substantially from the four cases we consider. For $\theta > \theta^c$, the fluctuations of the largest eigenvalues have been analyzed in alternative settings; see, e.g., \citet{benaych2011fluctuations,benaych2012singular,onatski2012asymptotics}. These works develop analogues of \eqref{eq_spiked_Wigner_Gaussian}, \eqref{eq_spiked_covariance_as}, and \eqref{eq_factors_as}, but the formulas for $\theta^c$, $\lambda_+$, $\lambda(\theta)$, and $V(\theta)$ become more complex. We conjecture that the confidence interval procedure from Section~\ref{Section_confidence_algorithm_1} remains valid in such broader contexts, with updated model-dependent parameters from Table \ref{Table_parameters}.

One applied setting of particular interest is the approximate factor model, which resembles the setup in Section \ref{Section_factor} but allows the noise matrix $\mathcal E$ to have a more complex correlation structure rather than being i.i.d. To apply our confidence interval procedure here, proceed as follows: first, select a value for $\lambda_+$; then use all eigenvalues below $\lambda_+$ to estimate the empirical distribution of noise eigenvalues (replacing the parametric forms in Table \ref{Table_LLN_shapes}); next, substitute this estimate into the formulas from \citet{onatski2012asymptotics} for $\lambda(\theta)$ and $V(\theta)$; and finally, apply our method to construct confidence intervals for the eigenvalues exceeding $\lambda_+$. Choosing $\lambda_+$ optimally is delicate; one approach, suggested in \citet[Section 4]{onatski2010determining}, is to select it based on the characteristic $\sqrt{x}$ behavior of the eigenvalue density near the edge -- a feature clearly visible in the four models of Table \ref{Table_LLN_shapes} and present in many other cases.

\section{Conclusion} \label{Section_conclusion}

The paper presents a unified framework for conducting inference on signal strength in high-dimensional signal plus noise models, with a particular focus on the critical regime where standard Gaussian approximations fail. We demonstrate that the limiting distribution of top eigenvalues is governed by a universal stochastic process, the transition process $\mathcal T(\Theta)$, whose quantiles can be tabulated and used to construct valid confidence intervals. This approach applies uniformly across four canonical models: spiked Wigner matrices, spiked sample covariance matrices, factor models, and canonical correlation analysis.

Our procedure is robust to both weak and critical signals, enabling practitioners to distinguish between informative and non-informative components without imposing assumptions on signal strength. Our methodology reveals a surprising universality: despite differences in the statistical structure of the models, the same transition process governs the fluctuations of their top eigenvalues. This suggests deeper underlying principles in high-dimensional inference and opens an avenue for future research in more general signal plus noise settings.

\section{Appendix A: Random Stieltjes transform at the edge and asymptotics}

\label{Section_asymptotics_proofs}

Our goal in this section is to rigorously introduce the Airy--Green function $\mathcal G(w)$ and related limit theorems. In all settings we study, the equations linking the spiked model parameters to the observed largest eigenvalues can be asymptotically expressed via $\mathcal G(w)$.

We first recall the definitions of meromorphic functions and their convergence for use in the proofs. A function $f(z)$ is meromorphic if it is analytic on $\mathbb C \setminus \{x_j\}$, where $\{x_j\}$ is at most a countable set of isolated poles of finite order. Equivalently, $f$ is  analytic from $\mathbb C$ to the Riemann sphere $\overline {\mathbb C}$.
For meromorphic $f_n$ and $f$, we write $f_n \tomer f$ if for any compact $W \subset \mathbb C$ we have
\[ \sup_{z \in W} d(f_n(z), f(z)) \to 0~, \qquad \text{ with } \qquad
d(\zeta, \zeta') = \frac{|\zeta - \zeta'|}{\sqrt{1 + |\zeta|^2}\sqrt{1+|\zeta'|^2}},\]
the spherical distance. Equivalently, $f_n \tomer f$ if and only if $f_n \to f$ uniformly on any compact $W \subset \mathbb C$ not containing the poles of $f$.


\subsection{The function $\mathcal G(w)$} \label{Section_Gw} We restate and prove the statements from Section \ref{Section_Def_Gw}.

\begin{theorem} \label{Theorem_G_def} Let $\aa_1\ge \aa_2\ge\aa_3\ge\dots$ be a realization of the Airy$_1$ point process and let $\{\xi_j\}_{j=1}^{\infty}$ be i.i.d.\ Gaussian $\mathcal N(0,1)$ independent of $\{\aa_j\}_{j=1}^{\infty}$.   Almost surely, for every $w \in \mathbb C \setminus \{\aa_j\}$ there exists a limit
\begin{equation}
\label{eq_G_def_2}
 \mathcal G(w)=\lim_{x\to-\infty} \left[\left(\sum_{j:\, \aa_j>x} \frac{\xi_j^2}{w-\aa_j}\right) - \frac{2}{\pi}\sqrt{-x} \right],
\end{equation}
and, moreover, the convergence is uniform on any compact   $W\subset \mathbb{C} \setminus \{\aa_j\}$.
\end{theorem}
\begin{remark}
The conditional expectation $\mathbb E (\mathcal G (w) | \{\aa_j\})$ (equivalent to replacing $\xi^2_j$ by $1$) was used recently in another context by \citet{huang2024convergence}. The restriction of $\mathcal G(w)$ to real $w>\aa_1$ resembles the objects studied in \cite{baik2021spherical}.
\end{remark}
The theorem asserts that almost surely the functions converge in the topology $\tomer$.

\begin{proposition} \label{Proposition_G_at_infinity}
 $\mathcal G(w)$ is a meromorphic function with poles at $\{\aa_j\}_{j=1}^{\infty}$ and satisfying:
\begin{equation} \label{eq_G_limits}
 \lim_{\begin{smallmatrix} w\to\infty\\ \mathrm{Re}(w)\ge 0\end{smallmatrix}} \left|\mathcal G(w)+\sqrt{w}\right|=0, \qquad \text{almost surely.}
\end{equation}
For real $w$, in the sense of convergence in distribution, we also have
\begin{equation}
\label{eq_G_Gaussian}
 \lim_{w\to+\infty} w^{1/4}\bigl(\mathcal G(w)+\sqrt{w}\bigr) \stackrel{d}{=} \mathcal N(0,1).
\end{equation}
\end{proposition}
\begin{remark}
 For $\sqrt{w}$ in \eqref{eq_G_limits}, one should use the branch of the square root which is positive on positive reals. In particular, $\sqrt{\ii R}=\frac{1+\ii}{\sqrt{2}} \sqrt{R}$, $R>0$, where $\ii=\sqrt{-1}$. While we do not need this, the asymptotics \eqref{eq_G_limits} can be extended from $\mathrm{Re}(w)\ge 0$ to all $w$ such that $|\arg(w)|<\pi-\eps$ for a fixed $\eps>0$. Similarly, \eqref{eq_G_Gaussian} can be extended to complex $w$.
\end{remark}

\begin{proposition} \label{Proposition_Transition_to_Gauss_appendix}
 Recall $\mathcal T(\Theta)$ of Definition \ref{Definition_Transition_function} that solves $\mathcal G(w)=-\Theta$. Almost surely, $\Theta\mapsto \mathcal T(\Theta)$ is an increasing bijection of $\mathbb R$ onto $(\aa_1, \infty)$.  As $\Theta\to +\infty$, $\mathcal T(\Theta)$ is asymptotically Gaussian: in distribution
 \begin{equation}
 \label{eq_Transition_to_Gauss_ap}
  \lim_{\Theta\to +\infty} \frac{\mathcal T(\Theta)-\Theta^2}{2 \sqrt{\Theta}}\stackrel{d}{=} \mathcal N(0,1).
 \end{equation}
\end{proposition}

The proofs are based on four lemmas describing the asymptotics of $\aa_j$.

\begin{lemma} \label{Lemma_1st_cor_function}
Let $\rho(x)\dd x$ be the first correlation measure of $\{\aa_j\}$, which means that for any compactly supported bounded $f(x)$ with finitely many discontinuity points, we have
$\E \sum_{j=1}^\infty f(\aa_j)=\int_{-\infty}^{\infty} f(x)\rho(x) \dd x.$
Then:
\begin{enumerate}
 \item $\rho(x)$ is a bounded continuous function of $x$;
 \item $\rho(x)$ decays faster than $\exp(-Cx)$ for any $C>0$ as $x\to+\infty$;
 \item  As $x\to -\infty$, $\rho(x)$ has the asymptotics $
 \rho(x)=  \frac{ (-x)^{1/2}}{\pi} +  O(1/|x|).$
\end{enumerate}
\end{lemma}
\begin{proof} Recall that the Airy function $\Ai(x)$ is a solution of the differential equation $\Ai''(x)-x\Ai(x)=0$, and is given by the improper integral $\Ai(x)=\frac{1}{\pi} \int_0^{\infty} \cos(t^3+xt)\dd t$. We need two asymptotic expansions for $\Ai(x)$, which can be found in \citet[Section 10.4]{abramowitz1968handbook}:
\begin{align}
 \Ai(x)&= \frac{1}{2 \sqrt{\pi}\cdot x^{1/4}} \exp\left(-\frac{2}{3} x^{3/2}\right) \cdot \bigl[1+O\left(x^{-3/2}\right)\bigr],\qquad x\to +\infty,
 \\ \Ai(x)&= \frac{1}{\sqrt{\pi} (-x)^{1/4}}\sin \left(\frac{2}{3}(-x)^{3/2}+\frac{\pi}{4}\right) \cdot \bigl[1+O\left((-x)^{-3/2}\right)\bigr], \qquad x\to -\infty.
\end{align}
The expansions are valid in a complex neighborhood of the real axis and, hence, can be differentiated to get
\begin{align}
 \Ai'(x)&= -\frac{x^{1/4}}{2 \sqrt{\pi}} \exp\left(-\frac{2}{3} x^{3/2}\right) \cdot \bigl[1+O\left(x^{-3/2}\right)\bigr],\qquad x\to +\infty,
 \\ \Ai'(x)&= - \frac{ (-x)^{1/4}}{\sqrt{\pi}}\cos \left(\frac{2}{3}(-x)^{3/2}+\frac{\pi}{4}\right) \cdot \bigl[1+O\left((-x)^{-3/2}\right)\bigr], \qquad x\to -\infty.
\end{align}
\citet[(6.3.2), (6.1.18), (5.3.6)]{pastur2011eigenvalue} give an explicit formula for the first correlation measure of the Airy$_1$ point process in terms of the Airy function:
\begin{align}
 \label{eq_Airy_1st_corr} \rho(x)&=  [\Ai'(x)]^2 - x [\Ai(x)]^2  +\frac{1}{2}\Ai(x)\left(1-\int_x^{+\infty} \Ai(z) \dd z\right)
 \\&=\int_0^{\infty} \Ai(z+x)^2\dd z +\frac{1}{2}\Ai(x)\left(1-\int_x^{+\infty} \Ai(z) \dd z\right).
\end{align}

Plugging the asymptotics of $\Ai(x)$ into the definition of $\rho(x)$, we see that it decays super-exponentially as $x\to+\infty$, while for $x\to-\infty$
\begin{equation}
 \rho(x)=  \frac{ (-x)^{1/2}}{\pi}   \cdot \bigl[1+O\left((-x)^{-3/2}\right)\bigr] + O\bigl( (-x)^{-1}\bigr). \qedhere
\end{equation}
\end{proof}

\begin{lemma}
 For $\aa_1>\aa_2>\aa_3>\dots$ being (a realisation of) the Airy$_1$ point process,
 \begin{equation}
 \label{eq_Airy_expectation}
   \E \bigl( \#\{j\ge 1\mid \aa_j>-T\} \bigr)= \frac{2}{3\pi} T^{3/2} + O(\ln(T)), \quad T \to +\infty~.
 \end{equation}
\end{lemma}
\begin{proof}
 We can write the expectation in terms of the first correlation measure of $\{\aa_j\}$:
\begin{equation}
\label{eq_x8}
  \E \bigl( \#\{j\ge 1\mid \aa_j>-T\} \bigr)=\int_{-T}^{+\infty} \rho(x) \dd x.
\end{equation}
Plugging the result of Lemma \ref{Lemma_1st_cor_function} into \eqref{eq_x8} and integrating, we get \eqref{eq_Airy_expectation}.
\end{proof}

\begin{lemma}
 For $\aa_1>\aa_2>\aa_3>\dots$ being the Airy$_1$ point process, there exists a constant $C_1>0$ such that
 \begin{equation}
 \label{eq_Airy_variance}
  \lim_{T\to+\infty} \frac{\Var \bigl( \#\{j\ge 1\mid \aa_j>-T\} \bigr)}{\ln(T)}= \frac{3}{2\pi^2}. 
 \end{equation}
\end{lemma}
\begin{proof} We use a trick from \citet{o2010gaussian}. The edge scaling limit of \citet[Theorem 4.3]{forrester2001interrelationships} is the following identity in law:
 \begin{equation}
 \label{eq_x7}
  \mathrm{even}\bigl(\{\aa_i\}_{i=1}^{\infty}\cup \{\aa'_i\}_{i=1}^{\infty}\bigr)\stackrel{d}{=}\{\aa^{\beta=2}_i\}_{i=1}^{\infty},
 \end{equation}
 where $\{\aa_i\}$ and $\{\aa'_i\}$ are two independent copies of the Airy$_1$ point process, $\aa^{\beta=2}_i$ is the Airy$_2$ point process, and ``even'' is the operation of removing all particles with odd indices, i.e., keeping the 2nd, 4th, 6th, etc, largest ones. From \citet[Theorem 1]{soshnikov2000gaussian}, it is known that as $T\to+\infty$, $\Var \bigl( \#\{j\ge 1\mid \aa_j^{\beta=2}>-T\} \bigr)\sim  C \ln(T)$, where the value\footnote{\cite{soshnikov2000gaussian} stated a different value of $C$. The arithmetic error was noticed and corrected in  \cite[Theorem 6.2]{landon2022fluctuations}.} of the constant is $C=\frac{3}{4\pi^2}$. Through \eqref{eq_x7} this implies
 $$
  \lim_{T\to+\infty} \frac{\Var \bigl( \#\{j\ge 1\mid \aa_j>-T\} + \#\{j\ge 1\mid \aa'_j>-T\}  \bigr)}{\ln(T)}= 4\cdot C, 
 $$
 because the ``even'' operator divides the number of particles by two (up to error of at most $1$), and therefore divides the variance by four (up to an error negligible as $T\to\infty$).
 Since variances for independent random variables are added, we get \eqref{eq_Airy_variance} with right-hand side $2C$.
\end{proof}
\begin{lemma} \label{Lemma_Airy_asymptotics} For each $\eps>0$ there exists a random variable $\mathfrak J=\mathfrak J(\eps)$, such that
\begin{equation}\label{eq_Airy_asymptotics}
 \left|\aa_j +\left(\frac{3 \pi j}{2}\right)^{2/3}\right|\le j^{\eps}, \qquad \text{ almost surely for all }j>\mathfrak J.
\end{equation}
\end{lemma}
\begin{proof}
 Choose $0<\eps<\tfrac{1}{3}$. For $n=1,2,\dots$, let $A_n$ be the event
 $$
 \bigr|\#\{j\ge 1\mid \aa_j>-n\}-\E\#\{j\ge 1\mid \aa_j>-n\}\bigl|>n^{1/2+\eps}.
 $$
 Using Chebyshev's inequality, we have
 $$
  \mathrm{Prob}(A_n)\le \frac{\mathrm{Var}\bigl( \#\{j\ge 1\mid \aa_j>-n\}\bigr)}{n^{1+2\eps}}.
 $$
 Combining with \eqref{eq_Airy_variance}, we conclude that $\sum_{n=1}^{\infty} \mathrm{Prob}(A_n)<\infty$. Therefore, by the Borel--Cantelli lemma, there exists a random variable $\mathfrak n$, such that
 $$
  \bigr|\#\{j\ge 1\mid \aa_j>-n\}-\E\#\{j\ge 1\mid \aa_j>-n\}\bigl|\le n^{1/2+\eps}, \qquad \text{for all }n>\mathfrak n.
 $$
 Combining with \eqref{eq_Airy_expectation} and increasing $\mathfrak n$, if necessary, we conclude that almost surely
 $$
  \left|\#\{j\ge 1\mid \aa_j>-n\}-\frac{2}{3\pi} n^{3/2} \right|\le n^{1/2+\eps}, \qquad \text{for all }n>\mathfrak n.
 $$
 Therefore,
 $$
  \aa_{\lfloor \frac{2}{3\pi} n^{3/2}+n^{1/2+\eps}\rfloor}\le-n\qquad \text{and}\qquad  \aa_{\lfloor \frac{2}{3\pi} n^{3/2}-n^{1/2+\eps}\rfloor}>-n-1, \qquad \text{for all }n>\mathfrak n.
 $$
 Denoting $k=\lfloor \frac{2}{3\pi} n^{3/2}+n^{1/2+\eps}\rfloor$, $\ell=\lfloor \frac{2}{3\pi} n^{3/2}-n^{1/2+\eps}\rfloor$ we conclude that for large $k$ and $\ell$,
 \begin{equation}
 \label{eq_x9}
  \aa_{k}<-\left(\frac{3\pi k}{2}\right)^{2/3} +\frac{1}{2} k^{\eps}\qquad \text{and}\qquad  \aa_{\ell}>-\left(\frac{3\pi \ell}{2}\right)^{2/3} -\frac{1}{2}\ell^{\eps}.
 \end{equation}
 In order to extend the inequalities from $k$ and $\ell$ of special form we used to all large $k$ and $\ell$, note that the distance between adjacent allowed values of $k$ is (assuming $k$ is large):
 $$
   \frac{2}{3\pi} (n+1)^{3/2}+(n+1)^{1/2+\eps} - \left( \frac{2}{3\pi} n^{3/2}+n^{1/2+\eps} \right)< \frac{1}{3} n^{1/2}<k^{1/3},
 $$
 and similarly for $\ell$. Hence, the monotonicity of $\aa_k$ in $k$ and the first inequality in \eqref{eq_x9} imply that for a (random) $\mathfrak K$, we have almost surely
 $$
   \aa_{k}<-\left(\frac{3(k-k^{1/3})}{2\pi}\right)^{2/3} +\frac{1}{2}k^{\eps}< -\left(\frac{3k}{2\pi}\right)^{2/3}+k^{\eps}, \qquad \text{for all }k>\mathfrak K.
 $$
Similarly producing a corollary of the second inequality in \eqref{eq_x9}, we get \eqref{eq_Airy_asymptotics}.
\end{proof}

\begin{proof}[Proof of Theorem \ref{Theorem_G_def}]
 We split $\mathcal G(w)$ into two parts:
 \begin{equation}
\label{ed_G_def_3}
 \mathcal G(w)=\lim_{x\to-\infty} \left[\sum_{j:\, \aa_j>x} \frac{1}{w-\aa_j} - \frac{2}{\pi}\sqrt{-x} \right]+\lim_{x\to-\infty} \left[\sum_{j:\, \aa_j>x} \frac{\xi_j^2-1}{w-\aa_j} \right]=\mathcal G_1(w)+\mathcal G_2(w).
\end{equation}
For $\mathcal G_1(w)$, we further write it as:
\begin{equation}\label{eq_x10}
 \mathcal G_1(w)=\lim_{x\to-\infty} \sum_{j:\, \aa_j>x} \left[ \frac{1}{w-\aa_j} -  \frac{1}{\ii-\aa_j}\right] +\lim_{x\to-\infty} \left[\left(\sum_{j:\, \aa_j>x} \frac{1}{\ii-\aa_j}\right) - \frac{2}{\pi}\sqrt{-x} \right].
\end{equation}
Note that
$
  \frac{1}{w-\aa_j} -  \frac{1}{\ii-\aa_j}=\frac{\ii - w}{(w-\aa_j)(\ii-\aa_j)}.
$
Hence, using \eqref{eq_Airy_asymptotics}, the $j$--th term in the first sum of \eqref{eq_x10} decays as $j^{-4/3}$ and the sum is absolutely convergent, uniformly in $w$ bounded away from $\aa_j$. For the second sum, its imaginary part is
$$
 \lim_{x\to-\infty} \left[\sum_{j:\, \aa_j>x} \frac{-\ii}{1+\aa_j^2}\right],
$$
which is again absolutely convergent. The real part is
\begin{equation}
\label{eq_x11}
 \lim_{x\to-\infty} \left[\left(\sum_{j:\, \aa_j>x} \frac{-\aa_j}{1+\aa_j^2}\right) - \frac{2}{\pi}\sqrt{-x} \right].
\end{equation}
Using \eqref{eq_Airy_asymptotics}, for large $j$ we have $-\aa_j=\left(\frac{3\pi j}{2}\right)^{2/3}+O(j^{\eps})$, and, therefore, for large $m$ and $n$:
\begin{multline} \label{eq_x14}
 \sum_{j=m}^n \frac{-\aa_j}{1+\aa_j^2}= \sum_{j=m}^n \frac{1}{\left(\frac{3\pi j}{2}\right)^{2/3}+O(j^{\eps})}=  \sum_{j=m}^n\left[ \frac{1}{\left(\frac{3\pi j}{2}\right)^{2/3}}+O( j^{-4/3+\eps})\right]\\
 =\left(\frac{2}{3\pi}\right)^{2/3}\int_m^n x^{-2/3} \dd x + o(1)=3\left(\frac{2}{3\pi}\right)^{2/3}\left(n^{1/3}-m^{1/3}\right) + o(1).
\end{multline}
We check the Cauchy criterion for \eqref{eq_x11} and compute the difference of its values at $x=-y$ and $x=-z$ for large $y>z>0$. Using \eqref{eq_x14}, we get
\begin{equation}
\label{eq_x12}
 3\left(\frac{2}{3\pi}\right)^{2/3}\left(n(y)^{1/3}-m(z)^{1/3}\right) - \frac{2}{\pi}\sqrt{y}+\frac{2}{\pi}\sqrt{z} + o(1),
\end{equation}
where $n(y)$ is the index $j$ for the closest to $-y$ point $\aa_j$ and $m(z)$ is the index for the closest to $-z$ point $\aa_j$. Assuming $y$, $z$ large, so that $n(y)>\mathfrak J$ and $m(z)>\mathfrak J$, we use \eqref{eq_Airy_asymptotics} and get
$$
  n(y)=\frac{2}{3\pi} \bigl(y+O(y^{\eps})\bigr)^{3/2}= \frac{2}{3\pi} y^{3/2} +O\bigl(y^{1/2+\eps}\bigr), \qquad m(z)= \frac{2}{3\pi} z^{3/2} +O\bigl(z^{1/2+\eps}\bigr).
$$
Plugging into \eqref{eq_x12} and choosing $\eps$ to be small enough, we get
\begin{multline*}
 3\left(\frac{2}{3\pi}\right)^{2/3}\left(\left(\frac{2}{3\pi} y^{3/2} +O\bigl(y^{1/2+\eps}\bigr)\right)^{1/3}-\left(\frac{2}{3\pi} z^{3/2} +O\bigl(z^{1/2+\eps}\bigr)\right)^{1/3}\right) - \frac{2}{\pi}\sqrt{y}+\frac{2}{\pi}\sqrt{z} + o(1)
 \\=O(y^{-1/2+\eps})+O(z^{-1/2+\eps})+o(1) \to 0,\qquad \text{as}\quad y>z\to\infty.
\end{multline*}
Therefore, \eqref{eq_x11} has an almost sure limit and $\mathcal G_1(w)$ is well-defined. We proceed to $\mathcal G_2(w)$ and again split it into two parts:
\begin{equation}
\label{eq_x13}
 \mathcal G_2(w)= \lim_{x\to-\infty} \left[\sum_{j:\, \aa_j>x} \frac{\xi_j^2-1}{\ii-\aa_j} \right]+\lim_{x\to-\infty} \left[\sum_{j:\, \aa_j>x} \frac{(\xi_j^2-1)(\ii -w)}{(w-\aa_j)(\ii-\aa_j)} \right].
\end{equation}
We would like to condition on the (typical) values of $\{\aa_j\}$, and then prove that both limits exist almost surely with respect to the randomness coming from $\xi_j$. For the imaginary part of the sum in the first limit, we notice that
$$
 \left|\mathrm{Im}\left( \frac{\xi_j^2-1}{\ii-\aa_j} \right) \right|=\frac{|\xi_j^2-1|}{1+(\aa_j)^2}.
$$
Using \eqref{eq_Airy_asymptotics} and the monotone convergence theorem (conditionally on $\{\aa_j\}$ the sum of expectations with respect to $\xi_j$ is finite), we see that almost surely
$$
 \sum_{j=1}^{\infty}\frac{|\xi_j^2-1|}{1+(\aa_j)^2}<\infty.
$$
Hence, the imaginary part of the first sum in \eqref{eq_x13} is absolutely convergent and $x\to-\infty$ limit is well-defined. The same monotone convergence argument shows that the second sum in \eqref{eq_x13} is dominated by a convergent series, and, therefore, it is absolutely convergent uniformly over $w$ in compact sets bounded away from $\aa_j$. It remains to deal with the real part of the first sum in \eqref{eq_x13}:
$$
 \lim_{x\to-\infty} \mathrm{Re} \left[\sum_{j:\, \aa_j>x} \frac{\xi_j^2-1}{\ii-\aa_j} \right]= \lim_{x\to-\infty}  \left[\sum_{j:\, \aa_j>x} (\xi_j^2-1)\frac{-\aa_j}{1+(\aa_j)^2} \right].
$$
We condition on $\{\aa_j\}$ and note that we deal with a sum of independent mean $0$ random variables. Hence, by the Kolmogorov two-series theorem (see, e.g., \citet[Theorem 2.5.6]{durrett2019probability}), the almost sure convergence would follow from the convergence of the sum of (conditional) variances, i.e., convergence of the series
$$
 \sum_{j=1}^{\infty} \E\bigl[ (\xi_j^2-1)^2\bigr] \frac{(\aa_j)^2}{(1+(\aa_j)^2)^2},
$$
which readily follows from \eqref{eq_Airy_asymptotics}. We conclude that $\mathcal G_2(w)$ is also well-defined.
\end{proof}

\begin{proof}[Proof of Proposition \ref{Proposition_G_at_infinity}] We rewrite the definition \eqref{eq_G_def_2} of $\mathcal G(w)$ as
$$
\lim_{x\to-\infty} \left[\sum_{j:\, \aa_j>x} \left( \frac{1}{w+\left(\frac{3\pi j}{2}\right)^{2/3}} +  \frac{\xi_j^2-1}{w+\left(\frac{3\pi j}{2}\right)^{2/3}}+  \xi_j^2\left[\frac{\left(\frac{3\pi j}{2}\right)^{2/3}+\aa_j}{(w-\aa_j)(w+\left(\frac{3\pi j}{2}\right)^{2/3})}\right] \right) - \frac{2}{\pi}\sqrt{-x} \right].
$$
Splitting the sum into three and using Lemma \ref{Lemma_Airy_asymptotics}, the last expression is transformed into
\begin{multline}
\label{eq_G_def_3}
 \mathcal G(w)=\lim_{x\to-\infty} \left[\left(\sum_{j:\, \left(\frac{3\pi j}{2}\right)^{2/3}<-x} \frac{1}{w+\left(\frac{3\pi j}{2}\right)^{2/3}}\right) -  \frac{2}{\pi}\sqrt{-x} \right]
   +\sum_{j=1}^{\infty} \frac{\xi_j^2-1}{w+\left(\frac{3\pi j}{2}\right)^{2/3}}\\
  + \sum_{j=1}^{\infty} \xi_j^2\left[\frac{\left(\frac{3\pi j}{2}\right)^{2/3}+\aa_j}{(w-\aa_j)(w+\left(\frac{3\pi j}{2}\right)^{2/3})}\right].
\end{multline}
We show that as $w\to\infty$ with $\mathrm{Re}(w)\ge 0$, the second and third sums in \eqref{eq_G_def_3} vanish.

For the third sum, let us show that it is $o(|w|^{-1/4})$, in the sense that there exists a random variable $\mathfrak c$, such that for all $w$ with $\mathrm{Re}(w)\ge 0$ and $|w|\ge 1$, we have:
\begin{equation}
\label{eq_x59}
 \left| \sum_{j=1}^{\infty} \xi_j^2\left[\frac{\left(\frac{3\pi j}{2}\right)^{2/3}+\aa_j}{(w-\aa_j)(w+\left(\frac{3\pi j}{2}\right)^{2/3})}\right]\right|\le \mathfrak c |w|^{-1/4}.
\end{equation}
We use Lemma \ref{Lemma_Airy_asymptotics} and note that for $j>\mathfrak J$, the numerator satisfies $|\left(\frac{3\pi j}{2}\right)^{2/3}+\aa_j|\le j^{\eps}$. In addition, since $\xi_j^2$ has exponential tails, the Borel--Cantelli lemma implies that there exists a random $\mathfrak C=\mathfrak C(\eps)>0$, such that almost surely $\xi_j^2< \mathfrak C j^{\eps}$ for all $j=1,2,\dots$. We choose $\eps$ to be small enough and upper-bound the series \eqref{eq_x59} by three sums:
\begin{equation}
\label{eq_x57}
  \sum_{j=1}^{\mathfrak J}  \frac{\xi_j^2\left|\left(\frac{3\pi j}{2}\right)^{2/3}+\aa_j\right|}{|w-\aa_j||w+\left(\frac{3\pi j}{2}\right)^{2/3}|} + \sum_{j=\mathfrak J+1}^{\lfloor |w|^{3/2}\rfloor } \frac{ 2 \mathfrak C  j^{2 \eps} }{|w+\left(\frac{3\pi j}{2}\right)^{2/3}|^2}+\sum_{j=\lfloor |w|^{3/2}\rfloor+1}^{\infty} \frac{2 \mathfrak C  j^{2\eps} }{|w+\left(\frac{3\pi j}{2}\right)^{2/3}|^2}.
\end{equation}
The first sum is finite, and, therefore, almost surely converges to $0$ at speed $|w|^{-2}$ as $|w|\to\infty$. The second sum has at most $|w|^{3/2}$ terms and each term is upper-bounded as $O(|w|^{-2})$; hence, the sum is $O(w^{-1/2})$. The terms of the last sum can be upper-bounded by a constant times $j^{2\eps-4/3}$, and therefore the sum is upper bounded by a constant times $|w|^{\frac{3}{2}(2\eps -1/3)}$. Combining all three bounds, we arrive at \eqref{eq_x59}.

 For the second sum in \eqref{eq_G_def_3}, summation by parts converts it into:
$$
\frac{1}{w+\left(\frac{3\pi (M+1)}{2}\right)^{2/3}}\sum_{k=1}^M (\xi_k^2-1)\\-\sum_{j=1}^{M} \left[\sum_{k=1}^j [\xi_k^2-1]\right]\left(\frac{1}{w+\left(\frac{3\pi (j+1)}{2}\right)^{2/3}}-\frac{1}{w+\left(\frac{3\pi j}{2}\right)^{2/3}}\right).
$$
Taking absolute values, we get an upper-bound for a deterministic constant $C>0$
$$
 \left|\sum_{j=1}^{M} \frac{\xi_j^2-1}{w+\left(\frac{3\pi j}{2}\right)^{2/3}}\right|\le\frac{C}{|w|+M^{2/3}}\left|\sum_{k=1}^M (\xi_k^2-1)\right|+\sum_{j=1}^{M} \left|\sum_{k=1}^j [\xi_k^2-1]\right| \frac{ C j^{-1/3}}{(|w|+j^{2/3})^2}.
$$
Applying the Law of Iterated Logarithm to the sums $\sum_{k=1}^j [\xi_k^2-1]$, we find another random variable $\mathfrak C' > 0$,  not dependent on $w$, such that
$$
 \left|\sum_{j=1}^{M} \frac{\xi_j^2-1}{w+\left(\frac{3\pi j}{2}\right)^{2/3}}\right|\le\frac{\mathfrak C'}{|w|+M^{2/3}} \sqrt{M\ln\ln M}+\sum_{j=1}^{M} \sqrt{j\ln\ln (j+2)} \frac{\mathfrak C' j^{-1/3}}{(|w|+j^{2/3})^2}.
$$
The last expression tends to $0$ as $|w|\to\infty$, uniformly in $M$. Hence, for large $w$, up to $o(1)$ error only the first term in \eqref{eq_G_def_3} contributes to the $w\to\infty$ asymptotics. This term is deterministic and approximates an integral. For large $x$ and $w$ with $\mathrm{Re}(w)\ge 0$, we have
$$
 \sum_{j:\, \left(\frac{3\pi j}{2}\right)^{2/3}<-x} \frac{1}{w+\left(\frac{3\pi j}{2}\right)^{2/3}}=
 \sum_{j:\, \left(\frac{3\pi j}{2}\right)^{2/3}<-x}\left( \int_{j-1}^j\frac{\dd y}{w+\left(\frac{3\pi y}{2}\right)^{2/3}} +O\left(\frac{j^{-1/3}}{\left|w+\left(\frac{3\pi j}{2}\right)^{2/3}\right|^2}\right)\right).
$$
Let us upper bound the sum of the $O(\cdot)$ terms:
\begin{multline}
 \sum_{j=1}^{\infty} \frac{j^{-1/3}}{\left|w+\left(\frac{3\pi j}{2}\right)^{2/3}\right|^2}=
 \sum_{j=1}^{\lfloor |w|^{3/2}\rfloor } \frac{j^{-1/3}}{\left|w+\left(\frac{3\pi j}{2}\right)^{2/3}\right|^2} +  \sum_{j=\lfloor |w|^{3/2}\rfloor+1}^{\infty} \frac{j^{-1/3}}{\left|w+\left(\frac{3\pi j}{2}\right)^{2/3}\right|^2}\\
 \le |w|^{3/2} |w|^{-2}+ \mathrm{const}\cdot  \sum_{j=\lfloor |w|^{3/2}\rfloor+1}^{\infty} j^{-5/3}\le  |w|^{-1/2} + \mathrm{const} \cdot |w|^{-\frac{3}{2}\cdot \frac{2}{3}}= O(|w|^{-1/2}).
\end{multline}
 It remains to analyze the integral, for which we change the variables $v=\left(\frac{3\pi y}{2}\right)^{2/3}$:
\[\begin{split}
 \int_0^{\frac{2}{3\pi}(-x)^{3/2}} \frac{\dd y}{w+\left(\frac{3\pi y}{2}\right)^{2/3}}&=\frac{1}{\pi} \int_0^{-x} \frac{\sqrt{v}}{w+v} \dd v=\frac{1}{\pi}\left[2\sqrt{v} - 2\sqrt{w}\arctan\left(\frac{\sqrt{v}}{\sqrt{w}}\right)\right]_{v=0}^{v=-x}
 \\&=\frac{2}{\pi}\sqrt{-x}-\frac{2}{\pi}\sqrt{w}\arctan\left(\frac{\sqrt{-x}}{\sqrt{w}}\right).
\end{split}\]
We conclude that the first term in \eqref{eq_G_def_3} is asymptotically
\begin{multline}
\label{eq_x58}
 -\frac{2}{\pi}\sqrt{w} \lim_{x\to-\infty} \arctan\left(\frac{\sqrt{-x}}{\sqrt{w}}\right) + O(|w|^{-1/2})\\= -\frac{2}{\pi}\sqrt{w} \lim_{x\to-\infty}\left(\frac{\pi}{2}- \arctan\left(\frac{\sqrt{w}}{\sqrt{-x}}\right)\right) + O(|w|^{-1/2})=-\sqrt{w}+O(|w|^{-1/2}).
\end{multline}

Plugging back into \eqref{eq_G_def_3}, we arrive at \eqref{eq_G_limits}:
$$
 \mathcal G(w)=-\sqrt{w}+o(1), \qquad \text{as }w\to\infty\text{ with }\mathrm{Re}(w)\ge 0.
$$

In order to prove \eqref{eq_G_Gaussian}, we again use \eqref{eq_G_def_3}. \eqref{eq_x59} and \eqref{eq_x58} imply that the sum of the first and the third terms is $-\sqrt{w}+o\left(w^{-1/4}\right)$ and it remains to analyze the second term. It is a sum of mean $0$ independent random variables, and the Central Limit Theorem applies. Hence, it remains to compute the asymptotic variance of the sum as $w\to\infty$, which is 
\begin{equation}
 \sum_{j=1}^{\infty} \frac{2}{\left(w+\left(\frac{3\pi j}{2}\right)^{2/3}\right)^2}=(2+o(1))\int\limits_0^{\infty} \frac{\dd x}{\left(w+\left(\frac{3\pi x}{2}\right)^{2/3}\right)^2}=(2+o(1)) \frac{1}{2\sqrt{w}}. \qedhere
\end{equation}
\end{proof}

\begin{proof}[Proof of Proposition \ref{Proposition_Transition_to_Gauss_appendix}] An equivalent statement with a different proof can be found in \citet[Theorem 4.1.1]{bloemendal2011finite}. Our proof is based on \eqref{eq_G_Gaussian}, which we restate as
\begin{equation}
\label{eq_x60}
 \mathcal G(w)=-\sqrt{w}+ w^{-1/4} \mathcal N(0,1)+o\left(w^{-1/4}\right), \qquad w\to + \infty.
\end{equation}
Using Definition \ref{Definition_Transition_function}, the main computation of the proof is to replace $\mathcal G(w)$ with $-\Theta$ and then solve \eqref{eq_x60}, viewed as an equation on unknown $w$, and treating $\Theta$ as a parameter. In this way we get the desired equivalent form of \eqref{eq_Transition_to_Gauss_ap}:
$$
 \mathcal T(\Theta)=w= \Theta^2 + 2  \Theta ^{1/2}\mathcal N(0,1) + o\left( \Theta^{1/2}\right), \qquad \Theta\to +\infty.
$$

In order to justify the validity of this computation, we use the monotonicity of $\mathcal G(w)$ on $[\aa_1,+\infty)$. The distributional limit of $\frac{\mathcal T(\Theta)-\Theta^2}{2 \sqrt{\Theta}}$ is obtained from the computation of the following probabilities for $t\in\mathbb R$, in which we used Definition \ref{Definition_Transition_function}:
$$
 \mathrm{Prob}\left(\mathcal T(\Theta)\le 2 t \sqrt{\Theta} +\Theta^2 \right)= \mathrm{Prob}\left( -\Theta\ge \mathcal G(2 t \sqrt{\Theta} +\Theta^2), \quad \aa_1\le  2 t \sqrt{\Theta} +\Theta^2 \right).
$$
The second condition $ \aa_1\le  2 t \sqrt{\Theta} +\Theta^2$ has probability approaching $1$ as $\Theta\to \infty$, and, therefore, can be dropped. For the first condition, we use \eqref{eq_x60} to transform it as $\Theta\to +\infty$:
\begin{multline*}
  \mathrm{Prob}\left(- \Theta\ge  -\sqrt{2 t \sqrt{\Theta} +\Theta^2}+ (2 t \sqrt{\Theta} +\Theta^2)^{-1/4} \mathcal N(0,1)+o\left((2 t \sqrt{\Theta} +\Theta^2)^{-1/4}\right)\right)
  \\=  \mathrm{Prob}\left( -\Theta\ge  -\Theta -  t \Theta^{-1/2}+ \Theta^{-1/2} \mathcal N(0,1)+o\left(\Theta^{-1/2}\right)\right)
  =  \mathrm{Prob}\bigl(  t \ge \mathcal N(0,1)+o\left(1\right)\bigr),
\end{multline*}
and the last probability clearly tends to the Gaussian distribution function as $\Theta\to\infty$.
\end{proof}

\subsection{A class of meromorphic functions} We first consider deterministic functions, then introduce randomness and establish theorems that aid in proving convergence to $\mathcal G(w)$.

\begin{definition}
 Given a real number $\gamma\in\mathbb R$,  a sequence of real numbers $x_1\ge x_2\ge x_3\ge \dots$ with $\lim_{n\to\infty}x_n=-\infty$, and a sequence of non-negative weights $\{w_j\}_{j=1}^{\infty}$, satisfying
 \begin{equation}
 \label{eq_MH_weight_bound}
  \sum_{j=1}^{\infty} \frac{w_j}{1+x_j^2}<\infty,
 \end{equation}
 we define a complex function
 \begin{equation}
 \label{eq_MH_function}
  f(z)=\gamma + \sum_{j=1}^{\infty} w_j\left(\frac{1}{z-x_j}+\frac{x_j}{1+x_j^2}\right), \qquad z\in \mathbb C\setminus \{x_j\}_{j=1}^{\infty}.
 \end{equation}
 We let $\MH$ denote the convex cone of all complex functions of this form. The minus in the notation $\MH$ indicates that $x_n\to-\infty$.
\end{definition}

Note that we allow some $w_j$ to vanish, so \eqref{eq_MH_function} may reduce to a finite sum.

\begin{lemma}
The sum \eqref{eq_MH_function} converges uniformly  on any compact subset of $\mathbb C\setminus \{x_j\}_{j=1}^{\infty}$; that is, it converges in the topology $\tomer$.
\end{lemma}
\begin{proof}
 For $z$ in a compact set $\mathfrak Z\subset \mathbb C\setminus\{x_j\}_{j=1}^{\infty}$, we have
 $$
  \left| w_j\left(\frac{1}{z-x_j}+\frac{x_j}{1+x_j^2}\right)\right|=\left|\frac{1+zx_j}{z-x_j} \cdot \frac{w_j}{1+x_j^2}\right|\le C(\mathfrak Z) \cdot \frac{w_j}{1+x_j^2}.
 $$
 Hence, the absolute and uniform convergence of \eqref{eq_MH_function} follows from \eqref{eq_MH_weight_bound}.
\end{proof}
Our arguments crucially use the sublinearity of functions in $\MH$:

\begin{lemma} \label{Lemma_sublinear}
 For any $f\in\MH$, we have
 $$
  \lim_{R\to\pm \infty} \frac{f(\ii R)}{R}=0.
 $$
\end{lemma}
\begin{proof}
 We have
 $$
  f(\ii R)-f(\ii)=  \ii(1-R) \sum_{j=1}^{\infty} \frac{w_j}{(\ii R-x_j)(\ii-x_j)}.
 $$
 For $R>1$, the magnitude of the  sum is bounded by \eqref{eq_MH_weight_bound} and each term goes to $0$ as $R\to\infty$. Hence, by the dominated convergence theorem the sum goes to $0$ and $f(\ii R)-f(\ii)=o(R)$. Since obviously also $f(\ii)=o(R)$ as $R\to\infty$, the conclusion follows.
\end{proof}

Now we state a distributional convergence theorem for random functions from $\MH$. Note that we make no assumptions on $\gamma_n$ or $\gamma$.

\begin{theorem} \label{Theorem_random_convergence}
Take  random functions $f_n$, $n=1,2,\dots$, and $f$ from $\MH$, corresponding to random $(\gamma_n, \{x_{j; n}\}_{j=1}^{\infty}, \{w_{j;n}\}_{j=1}^{\infty})$ and $(\gamma, \{x_j\}_{j=1}^{\infty}, \{w_j\}_{j=1}^{\infty})$.
Suppose that, in the sense of convergence in finite-dimensional distributions:
\begin{equation}
\label{eq_finite_dim_conv}
   \lim_{n\to\infty} x_{j;n} = x_j,\quad\text{and}\quad  \lim_{n\to\infty} w_{j;n} = w_j, \text{ for each }j=1,2,\dots,
\end{equation}
and there exists a deterministic function $\phi: \mathbb R_+ \to \mathbb C$ such that for any $\eps>0$ and  $R>0$
\begin{equation}
\label{eq_tail}
   \lim_{R\to +\infty} \mathrm{Prob}( |f(\ii R)-\phi(R)|>\eps)=  \lim_{R\to +\infty} \limsup_{n\to\infty}\mathrm{Prob}( |f_n(\ii R)-\phi(R)|>\eps)=0.
\end{equation}
Then there exists a coupling that places all $f_n$ and $f$ on the same probability space such that almost surely $f_n \tomer f$.
\end{theorem}
As noted earlier, this implies convergence in distribution $f_n \tomerD f$, and thus convergence in distribution of $(f_n(z_j))_{j=1}^k$ to $(f(z_j))_{j=1}^k$ at any finite set of points.

Theorem \ref{Theorem_random_convergence} is inspired by \citet{aizenman2015ubiquity}, see Theorem 3.1 and Section 6 there, as well as \citet[Section 1.4]{sodin2018critical}. There are, however, important differences: in \citet{aizenman2015ubiquity} the function $\phi(R)$ was constant -- this does not hold in our main application. Additionally, our topology of convergence is stronger than that of \citet{aizenman2015ubiquity}, which is crucial when solving equations of the form $f_n(z)=\theta$. Unlike \citet{sodin2018critical}, we do not make use of any results from complex analysis, beyond the basics.

\begin{remark}
 One can also deal with several sequences of functions $f^{[k]}_n(z)$ converging towards $f^{[k]}(z)$, $k=1,2,\dots,K$. The tail condition \eqref{eq_tail}, the conclusion of the theorem, and the proof remain exactly the same for such extension.
\end{remark}

\noindent In the rest of this subsection we prove Theorem \ref{Theorem_random_convergence}. We start with a deterministic statement.

\begin{lemma} \label{Lemma_functional_convergence}
Let $f_n$, $n=1,2,\dots$, and $f$  be deterministic functions from $\MH$, corresponding to $(\gamma^n, \{x_{j;n}\}_{j=1}^{\infty}, \{w_{j;n}\}_{j=1}^{\infty})$, and $(\gamma, \{x_j\}_{j=1}^{\infty}, \{w_j\}_{j=1}^{\infty})$ . Suppose that:
 \begin{equation}
   \lim_{n\to\infty} w_{j;n} = w_j,\quad\text{and}\quad  \lim_{n\to\infty} x_{j;n} = x_j, \text{ for each }j=1,2,\dots, \qquad \text{ and}
 \end{equation}
 \begin{equation}
   \lim_{n\to\infty} f_n(\ii)=f(\ii).
 \end{equation}
 Then $f_n \tomer f$, i.e., $f_n(z) \to f(z)$, uniformly over $z$ in compact subsets of $\mathbb C \setminus\{x_j\}_{j=1}^{\infty}$.
\end{lemma}
\begin{proof} We have
$$
 f_n(z)-f_n(\ii)= \sum_{j=1}^{\infty} w_{j;n}\left(\frac{1}{z-x_{j;n}}-\frac{1}{\ii-x_{j;n}}\right)= (\ii-z)\sum_{j=1}^{\infty} \frac{w_{j;n}}{(z-x_{j;n})(\ii-x_{j;n})}.
$$
Each term in the last series converges towards its counterpart for the series of $f(z)-f(\ii)$, and we need to produce a uniform tail bound. For that we note
$$
 \left|\sum_{j=M}^{\infty} \frac{w_{j;n}}{(z-x_{j;n})(\ii-x_{j;n})}\right|\le \sum_{j=M}^{\infty}\left| \frac{\ii +x_{j;n}}{z-x_{j;n}}\right| \cdot  \frac{w_{j;n}}{1+x_{j;n}^2} \le \left[1+\max_{j\ge M} \left\{\frac{|z|+1}{|z-x_{j;n}|}\right\}\right]  \sum_{j=M}^{\infty} \frac{w_{j;n}}{1+x_{j;n}^2}.
$$
The first factor stays uniformly bounded as $n\to\infty$, and it remains to show that the second factor tends to $0$ as $M\to\infty$ (uniformly in $n$). For that we observe
\begin{equation}
\label{eq_x15}
  \sum_{j=M}^{\infty} \frac{w_{j;n}}{1+x_{j;n}^2}=- \mathrm{Im} f_n(\ii)-\sum_{j=1}^{M-1} \frac{w_{j;n}}{1+x_{j;n}^2}.
\end{equation}
For an arbitrary $\eps>0$, we choose $M$ large enough, so that
$$
  \sum_{j=M}^{\infty} \frac{w_j}{1+x_j^2}=-\mathrm{Im} f(\ii)-\sum_{j=1}^{M-1} \frac{w_j}{1+x_j^2}< \eps.
$$
Then sending $n\to\infty$, in the finite sum in the right-hand side of \eqref{eq_x15}, we deduce existence of $n_0$ such that for all $n>n_0$,
\begin{equation}
\label{eq_x16}
  \sum_{j=M}^{\infty} \frac{w_{j;n}}{1+x_{j;n}^2}<2 \eps.
\end{equation}
Because $\sum_{j=1}^{\infty}  \frac{w_{j;n}}{1+x_{j;n}^2}<\infty$ for each $n$, at the expense of increasing $M$, we can guarantee that \eqref{eq_x16} holds for all $n=1,2,\dots$.
\end{proof}

The next step is to study the value $f_n(\ii)$ which appeared in the previous lemma.

\begin{lemma} \label{Lemma_value at_i}
 Under the conditions of Theorem \ref{Theorem_random_convergence}, the random variables $f_n(\ii)$ converge in distribution as $n\to\infty$ towards $f(\ii)$, jointly with the convergence \eqref{eq_finite_dim_conv}.
\end{lemma}
\begin{proof}
 We take two constants $Q>q>1$, which are both large. Consider
\begin{equation}
\label{eq_x25}
  f_n(\ii q)-f_n(\ii)=\sum_{j=1}^\infty  w_{j;n} \left[\frac{1}{\ii q-x_{j;n}}-\frac{1}{\ii-x_{j;n}}\right]=\ii(1-q)\sum_{j=1}^\infty  \frac{w_{j;n}}{\left(\ii q-x_{j;n}\right)\cdot \left(\ii-x_{j;n}\right)}.
\end{equation}
Using \eqref{eq_finite_dim_conv}, term-by-term the last series converges as $n\to\infty$ towards
$$
 f(\ii q)-f(\ii)=\ii(1-q)\sum_{j=1}^\infty   \frac{w_j}{\left(\ii q-x_j\right)\cdot \left(\ii-x_j\right)},
$$
which is a well-defined finite random variable due to $f\in\MH$ and \eqref{eq_MH_weight_bound}. In order to justify the interchange of the order of summation and taking the $n\to\infty$ limit, we need to produce an additional tail bound. We observe that for large $Q-q$ we have
$$
\left|\sum_{j:\, |x_{j;n}|>Q}   \frac{w_{j;n}}{\left(\ii q-x_{j;n}\right)\cdot \left(\ii-x_{j;n}\right)}\right|
 \le 2 \sum_{j:\, |x_{j;n}|>Q}   \frac{w_{j;n}}{Q^2 +x_{j;n}^2}\le - \frac{2}{Q} \mathrm{Im}\bigl[f_n(\ii Q)\bigr].
$$
We claim that by choosing first large $Q$, and then large $n$ we can make $ \frac{2}{Q} \mathrm{Im}\bigl[f_n(\ii Q)\bigr]$ arbitrarily small with probability arbitrary close to $1$. Indeed,
combining Lemma \ref{Lemma_sublinear} with the first limit of \eqref{eq_tail}, we conclude that $\phi(Q)=o(Q)$ for large $Q$. Then the second limit in \eqref{eq_tail} implies that for large $n$ also $f_n(\ii Q)=o(Q)$ with high probability.

On the other hand, using \eqref{eq_finite_dim_conv} and $\lim_{n\to\infty} |x_n|=\infty$ from the definition of $\MH$, we conclude that for large $n$, the part of the sum \eqref{eq_x25} with $|x_{j;n}|\le Q$ has only finitely many (with the number dependent only on $Q$) terms with high probability. It follows that we are allowed to pass to the limit $n\to\infty$ in \eqref{eq_finite_dim_conv} and we have proven that for each $q>1$
\begin{equation}
\label{eq_x21}
 \lim_{n\to\infty}\left(f_n(\ii q)-f_n(\ii);\,  \left(x_{j;n}, w_{j;n}\right)_{j=1}^{\infty}\right)\stackrel{d}{=} \Bigl(f(\ii q)-f(\ii);\, \bigl(x_j, w_j\bigr)_{j=1}^{\infty}\Bigr).
\end{equation}
In order to finish the proof of the lemma, it remains to get rid of $f_n(\ii q)$. For that we notice that by \eqref{eq_tail}, $f_n(\ii q)-f_n(\ii)\approx \phi(q)-f_n(\ii)$ and also $f(\ii q)-f(\ii)\approx \phi(q)- f(\ii)$, i.e., the differences of the left-hand sides and the right-hand sides are small with probability close to $1$ when $q$ and $n$ are large. Subtracting $\phi(q)$ from both sides, we are done. \end{proof}

\begin{proof}[Proof of Theorem \ref{Theorem_random_convergence}] By Lemma \ref{Lemma_value at_i}, in distribution,
\begin{equation}
\label{eq_x26}
 \lim_{n\to\infty}\left(f_n(\ii);\,  \left(x_{j;n}, w_{j;n}\right)_{j=1}^{\infty}\right)\stackrel{d}{=} \Bigl(f(\ii);\, \bigl(x_j, w_j\bigr)_{j=1}^{\infty}\Bigr).
\end{equation}
By the Skorokhod's representation theorem (see, e.g., \citet[Chapter 1, Theorem 6.7]{Billingsley}) one can construct a coupling, so that \eqref{eq_x26} becomes almost sure convergence. Then Lemma \ref{Lemma_functional_convergence} implies $f_n \tomer f$. 
\end{proof}

The setting of Theorem \ref{Theorem_random_convergence} is tailored to our main application -- convergence to $\mathcal G(w)$ -- which is why we assume $x_1 \ge x_2 \ge \dots$. However, we could instead consider functions of the form \eqref{eq_MH_function} constructed from a sequence $x_1 \le x_2 \le \dots$ with $x_n \to +\infty$; we denote the class of such functions by $\Omega_+$. Similarly, we could allow both types of sequences and work with sums of the form $f(z) = f^-(z) + f^+(z)$, where $f^{\pm} \in \Omega_{\pm}$. The following theorem can be proved by repeating the argument of Theorem \ref{Theorem_random_convergence} word for word.

\begin{theorem} \label{Theorem_random_convergence_2sided}
 In Theorem \ref{Theorem_random_convergence}, instead of $f_n,f\in \Omega_-$, we can take $f_n=f^+_n+f^-_n$, $f=f^++f^-$, with $f_n^\pm,f^\pm\in \Omega_{\pm}$ corresponding to $(\gamma_n^\pm, \{x_{j; n^\pm}\}_{j=1}^{\infty}, \{w_{j;n}^\pm\}_{j=1}^{\infty})$ and $(\gamma^\pm, \{x_j^\pm\}_{j=1}^{\infty}, \{w_j^\pm\}_{j=1}^{\infty})$. Replacing \eqref{eq_finite_dim_conv} with
 \begin{equation}
\label{eq_finite_dim_conv_2}
   \lim_{n\to\infty} x_{j;n}^\pm = x_j^\pm,\quad\text{and}\quad  \lim_{n\to\infty} w_{j;n}^\pm = w_j^\pm, \text{ for each }j=1,2,\dots,
\end{equation}
and keeping \eqref{eq_tail}, the conclusion of Theorem \ref{Theorem_random_convergence} continues to hold: almost surely $f_n \tomer f$.
\end{theorem}
\begin{remark}
 For $f\in \Omega_+$, the function $-f(z)$  lies in the Herglotz-Nevanlinna class, mapping the upper half-plane to itself. Our proofs
  extend to other sublinear Herglotz-Nevanlinna functions but exclude the linear
  $az$, $a>0$, since Lemma \ref{Lemma_sublinear} is essential.
\end{remark}

Theorems \ref{Theorem_random_convergence} and \ref{Theorem_random_convergence_2sided} can be applied in various situations in random matrix theory, where $x_{j;n}$ are eigenvalues converging  to various universal scaling limits (such as Airy, sine, or Bessel point processes). When all $w_{j;n}$ equal $1$, we deal with the Stieltjes transform or the log-derivative of the charateristic polynomial, connecting us to the vast literature on the latter, see, e.g., \citet{lambert2020strong,johnstone2025edge,collins2025edge} for the most recent results related to the edge limits. In our main application, instead, $w_{j;n}$ are i.i.d.~random variables.

\subsection{Convergence to $\mathcal G(w)$}

In this section we use Theorem \ref{Theorem_random_convergence} to derive sufficient conditions for convergence towards $\mathcal G(w)$ of Theorem \ref{Theorem_G_def}. For each $N=1,2,\dots$, we are given an $N$--tuple of random variables $\lambda_{1;N}\ge \lambda_{2;N}\ge\dots\ge\lambda_{N;N}$; we assume that for each $N$ these variables are defined on its own probability space, which therefore depends on $N$. In addition, for each $N$ we are given a sequence of i.i.d.\ Gaussian random variables $\xi_1,\xi_2,\dots,\xi_N$ on the same $N$--dependent probability space (we omit an additional $N$ from the index of $\xi_i$, because the distribution of $\xi_i$ does not depend on it), which are independent of $\{\lambda_{i;N}\}_{i=1}^N$. In addition, we are given a continuous non-negative function $h(x)\ge 0$. The data is assumed to satisfy the following condition, using the empirical Stieltjes transform notation:
$$
 m_N(z)=\frac{1}{N}\sum_{i=1}^N \frac{h(\lambda_{i;N})}{z-\lambda_{i;N}}.
$$
\begin{assumption} \label{Assumption_for_limit} There exists a function $m(z)$ and constants $\lambda_+\in\mathbb R$, $\s>0$,  $\m\in\mathbb R$, $\eps>0$, and $C>0$ such that
\begin{enumerate}[label={(\roman*)}]
 \item In the sense of convergence of finite-dimensional distributions:
 \begin{equation}
 \label{eq_Assumption_Airy_finite}
    \lim_{N\to\infty} \left\{ N^{2/3} \s^{2/3} \left(\lambda_{j;N}- \lambda_+\right)  \right\}_{j=1}^{N}\stackrel{d}{=} \{\aa_j\}_{j=1}^{\infty},
 \end{equation}
where $\aa_1\ge \aa_2\ge\dots$ is the Airy$_1$ point process.
\item As $R\to 0$ we have
\begin{equation}
\label{eq_Assumption_Stieltjes_imaginary}
m(\lambda_++\ii  R)=\m -\s \frac{1+\ii}{\sqrt{2}}\sqrt{R} + o\left(1\right).
\end{equation}

\item For  all $N=1,2,\dots$ and all $N^{-2/3}<R<N^{-2/3+\eps}$ we should have
\begin{equation}
\label{eq_Assumption_local_law}
 \E\left[\left|m_N(\lambda_++\ii R)-m(\lambda_++ \ii R)\right|^2\right]\le \frac{C}{N^2 R^2}.
\end{equation}
\end{enumerate}
\end{assumption}

\smallskip

\noindent While not required by the assumptions, it is typical to have
$$
 m(z)=\int_{\mathbb R} \frac{h(x)}{z-x}  \mu(x)\dd x,
$$
where $\mu(x)\dd x$ is an independent of $N$ probability measure $\mu(x)\dd x$, supported on a finite interval $[\lambda_-,\lambda_+]\subset \mathbb R$, and which is a limit of the empirical measures of $\lambda_{i;N}$.

The first condition in Assumption \ref{Assumption_for_limit} specifies the edge limit of $\lambda_{i;N}$; the second condition is related to the square root behavior $h(x)\mu(x)\approx \frac{\s}{\pi} \sqrt{\lambda_+-x}$ for $x$ close to $\lambda_+$ from the left; the third condition is the optimal local law meaning that the empirical measure of $\lambda_{i;N}$ is close to the limit given by $\mu(x)\dd x$. 

\begin{theorem} \label{Theorem_as_convergence}
Under Assumption \ref{Assumption_for_limit} and denoting
\begin{equation}\label{eq:def-GN}
 \mathcal G_N(w)=N^{1/3} \s^{-2/3}\left(\frac{1}{N}\sum_{j=1}^N \frac{h(\lambda_{j;N})\, \xi_j^2}{\lambda_++N^{-2/3} \s^{-2/3} w-\lambda_{j;N}}-\m\right),
\end{equation} there exists a coupling that places random variables $\bigl(\{\lambda_{j;N}\}_{j=1}^N, \{\xi_j\}_{j=1}^N\bigr)_{N=1}^{\infty}$ and the Airy-Green function $\mathcal G(w)$ on the same probability space, such that almost surely
\begin{equation}\label{eq_x17}
\mathcal G_N(w) \tomer h(\lambda_+) G(w), \qquad N\to\infty.
\end{equation}
\end{theorem}
\begin{corollary} \label{Corollary_dist_convergence}
For any finitely many deterministic $w_1,\dots,w_k\in\mathbb C$, the convergence in \eqref{eq_x17} holds in joint distribution over $w=w_1$, \dots, $w=w_k$.
\end{corollary}
\begin{corollary} \label{Corollary_root}
 For $\theta\in\mathbb R$, let $\tilde w_N$ denote the largest real solution of the equation $\mathcal G_N(w)=\theta$ and let $\tilde w$ denote the largest real solution of the equation $\mathcal G(w)=\theta$. Under Assumption \ref{Assumption_for_limit},
 \begin{equation}
  \lim_{N\to\infty} \tilde w_N\stackrel{d}{=} \tilde w.
 \end{equation}
\end{corollary}
\begin{remark} \label{Remark_joint_convergence}
 We might have more than one sequence $\xi_j$ and more than one $h(x)$. One can take finite $K$ and consider simultaneous edge limits for
 $$
  \frac{1}{N} \sum_{j=1}^N \frac{h^{(k)}(\lambda_{j;N})\cdot(\xi_j^{(k)})^2}{z-\lambda_{j;N}}, \qquad k=1,2,\dots,K,
 $$
 where $\{\xi_j^{(k)}\}$ are all i.i.d.\ Gaussian $\mathcal N(0,1)$ over $j$, but might be correlated over $k$. Theorem \ref{Theorem_as_convergence} and Corollaries \ref{Corollary_dist_convergence}, \ref{Corollary_root} have immediate extensions to such setting, where each of the $k$ functions converges to its own $\mathcal G^{(k)}(w)$, constructed by the same $\{\aa_j\}$, but distinct sequences $\{\xi_j\}=\{\xi_j^{(k)}\}$, which might be correlated over $k$. The convergence would be joint over $k=1,2,\dots,K$, while the proof remains exactly the same.
\end{remark}

\begin{proof}[Proof of Theorem \ref{Theorem_as_convergence}] We want to apply Theorem \ref{Theorem_random_convergence} with $n=N$, $f_n=\mathcal G_N$, $f=h(\lambda_+)\cdot \mathcal G$.

{\bf Condition 1.} We need to check that the functions $\mathcal G_N$ and $\mathcal G$ are in the class $\MH$. For $\mathcal G_N$ this is immediate. $\mathcal G$ also belongs to $\MH$, as can be seen by rewriting the definition:
\begin{align*}
 \mathcal G(z)&=\lim_{x\to-\infty} \left[\left(\sum_{j:\, \aa_j>x} \frac{\xi_j^2}{z-\aa_j}\right) - \frac{2}{\pi}\sqrt{-x} \right]\\&=
 \lim_{x\to-\infty} \sum_{j:\, \aa_j>x} \xi_j^2\left[\frac{1}{z-\aa_j}+\frac{\aa_j}{1+\aa_j^2}\right] +  \lim_{x\to-\infty}\left( \sum_{j:\, \aa_j>x}\frac{-\aa_j}{1+\aa_j^2}- \frac{2}{\pi}\sqrt{-x} \right)
\end{align*}
and noting that almost sure convergence of $\displaystyle \sum_{j=1}^{\infty} \frac{\xi_j^2}{1+\aa_j^2}$
follows from Lemma \ref{Lemma_Airy_asymptotics} and monotone convergence theorem (applied conditionally on the values of $\{\aa_j\}_{j=1}^{\infty}$).

{\bf Condition 2.} Limits \eqref{eq_finite_dim_conv} are included in assumption \eqref{eq_Assumption_Airy_finite} for the particle positions, while the distribution of $w_{j;n}=h(\lambda_{j;N})\xi_j^2$ converges to that of $h(\lambda_+)\xi_j^2$ by continuity of $h(x)$.

{\bf Condition 3.} It remains to check \eqref{eq_tail}. We set $\phi(R)=-\frac{1+\ii}{\sqrt{2}}\sqrt{R}$. Proposition \ref{Proposition_G_at_infinity} yields that $\mathcal G(\ii R)+\frac{1+\ii}{\sqrt{2}} \sqrt{R}$ goes to $0$ as $R\to\infty$ almost surely, and, therefore, also in probability, verifying the first limit in \eqref{eq_tail}. For the second limit we prove that
\begin{equation}
\label{eq_x22}
  \lim_{R\to+\infty} \limsup_{N\to\infty}\E\left|\mathcal G_N(\ii R)+\frac{1+\ii}{\sqrt{2}}\sqrt{R}\right|^2=0.
\end{equation}
By the Markov inequality, \eqref{eq_x22} is sufficient for establishing \eqref{eq_tail}. Taking the expectation with respect to $\xi_j$ first and using $\E(\xi_j^2-1)^2=2$, we have

\begin{align}
\label{eq_x23}
 \E&\left|\mathcal G_N(\ii R)+\frac{1+\ii}{\sqrt{2}}\sqrt{R}\right|^2=\E\left(\mathcal G_N(\ii R)+\frac{1+\ii}{\sqrt{2}}\sqrt{R}\right)\left(\overline{\mathcal G_N(\ii R)}+\frac{1-\ii}{\sqrt{2}}\sqrt{R}\right)
 \\
 &\notag =  \E\left|N^{1/3}\s^{-2/3} m_N\left(\lambda_++\frac{\ii  R}{N^{2/3} \s^{2/3}}\right)-\m N^{1/3}\s^{-2/3}+\frac{1+\ii}{\sqrt{2}}\sqrt{R}\right|^2
 \\
 &\notag \qquad\qquad
 + \E \sum_{j=1}^{N} \frac{2 h^2(\lambda_{j;N})}{\left|\ii R-N^{2/3} \s^{2/3}(\lambda_{j;N}-\lambda_+)\right|^2}
 \\
 &\notag=  \E\Biggl|N^{1/3}\s^{-2/3} m\left(\lambda_++\frac{\ii  R}{N^{2/3} \s^{2/3}}\right)-\m N^{1/3}\s^{-2/3}+\frac{1+\ii}{\sqrt{2}}\sqrt{R}
 \\&\notag
 \qquad\qquad
 +N^{1/3}\s^{-2/3} \left(m_N\left(\lambda_++\frac{\ii  R}{N^{2/3} \s^{2/3}}\right)-m\left(\lambda_++\frac{\ii  R}{N^{2/3} \s^{2/3}}\right)\right)\Biggr|^2
 \\
 &\notag
 \qquad\qquad
 + \frac{ N^{1/3}\s^{-2/3}}{\ii R} \E\left[m_N\left(\lambda_+-\frac{\ii  R}{N^{2/3} \s^{2/3}}\right)-m_N\left(\lambda_++\frac{\ii  R}{N^{2/3} \s^{2/3}}\right)\right].
\end{align}
The fourth line of \eqref{eq_x23} becomes small as $R\to\infty$ by \eqref{eq_Assumption_Stieltjes_imaginary}, in which one needs to rescale $R$. The fifth line becomes small by \eqref{eq_Assumption_local_law}, again with rescaled $R$. The sixth line becomes small as $R\to\infty$ by a combination of \eqref{eq_Assumption_Stieltjes_imaginary} and \eqref{eq_Assumption_local_law}.
\end{proof}

\begin{proof}[Proof of Corollary \ref{Corollary_dist_convergence}]
Take $W=\{w_1,w_2,\dots,w_k\}$. The distribution of each $\aa_j$ is absolutely continuous (e.g., because there is a well-defined density of the first correlation measure \eqref{eq_Airy_1st_corr}), and therefore almost surely no $\aa_j$ belong to $W$. Hence, $(\mathcal G_N(w_1),\dots,\mathcal G_N(w_k))$ converges towards $(\mathcal G(w_1),\dots,\mathcal G(w_k))$ almost surely and, therefore, also in distribution.
\end{proof}

\begin{proof}[Proof of Corollary \ref{Corollary_root}] Directly subtracting the values at two points inside the definition of $\mathcal G(w)$, we see that it is a strictly monotone decreasing function of real $w>\aa_1$. Proposition \ref{Proposition_G_at_infinity} implies that it varies  from $+\infty$ to $-\infty$ on $(\aa_1,+\infty)$. We conclude that for each $\theta\in\mathbb R$, there exists a unique random $\tilde w>\aa_1$, such that $\mathcal G(\tilde w)=\theta$.

Next, we consider the coupling of Theorem \ref{Theorem_as_convergence}. By Rouch\'{e}'s theorem (see e.g., \citet{Ahlfors_1979}) locally-uniform convergence of meromorphic functions implies the convergence of their zeros. Hence, $\tilde w_N\to\tilde w$ almost surely, and, therefore, also in distribution.
\end{proof}

\subsection{Universal bound for $\beta$--ensembles.} \label{Section_beta_ensembles_a}
To verify the conditions of Theorem \ref{Theorem_as_convergence}, we use the following universal bound valid for all $\beta$--ensembles. Take $\beta>0$, a function $V(x):[a,b]\to\mathbb R$, called a potential, and consider for each $N=1,2,\dots$, a probability measure on $\{\lambda_i\}_{i=1}^N$, such that $b\ge \lambda_1 > \lambda_2>\dots>\lambda_N\ge a$, with probability density proportional to
\begin{equation}
\label{eq_beta_ensemble}
 \prod_{1\le i<j\le N} (\lambda_i-\lambda_j)^{\beta} \prod_{i=1}^N \exp\left(-\frac{\beta N}{2} V_N(\lambda_i)\right).
\end{equation}
We will be interested in the case $\beta = 1$ and the following three options for $V_N(\lambda)$, although the theorems we use hold in much greater generality.
\begin{itemize}
 \item GOE ensemble. The distribution of the eigenvalues of $\mathcal E=\frac{1}{\sqrt{2 N}}(\mathcal Z+\mathcal Z^\T)$ with $\sigma^2=1$ as in Section \ref{Section_spiked_Wigner} takes the form \eqref{eq_beta_ensemble} with (see, e.g., \citet[Chapter 1]{forrest}, \citet[Chapter 4]{pastur2011eigenvalue}):
 $$
  a=-\infty, \quad b=+\infty, \quad V_N(\lambda)=\frac{\lambda^2}{2}.
 $$
\item LOE ensemble. The distribution of the eigenvalues of $\frac{1}{S} X X^\T$, where $S\ge N$ and $X$ is $N\times S$ matrix of i.i.d.~$\mathcal N(0,1)$, as in Sections \ref{Section_spiked_covariance} and \ref{Section_factor} in the case of no signals $\theta$, takes the form \eqref{eq_beta_ensemble} with (see, e.g., \citet[Chapter 1]{forrest}, \citet[Chapter 7]{pastur2011eigenvalue}):
 $$
  a=0, \quad b=+\infty, \quad V_N(\lambda)=-\frac{S-N-1}{N}\ln(\lambda)+\frac{S}{N}\lambda.
 $$

 \item JOE ensemble. The distribution of the squared sample canonical correlations between independent matrices $\U$ and $\V$ of dimensions $N\times S$ and $M\times S$, respectively, and filled with i.i.d.\ $\mathcal N(0,1)$, as in Section \ref{Section_spiked_CCA} in the case of no signals $\theta$, with $N\le M$, $M+N\le S$ takes the form \eqref{eq_beta_ensemble} with (see e.g. \citet[Section 3.6.1]{forrest})
  $$
   a=0,\quad b=1,\quad V_N(\lambda)=-\frac{M-N-1}{N}\ln(\lambda)-\frac{S-N-M-1}{N}\ln(1-\lambda).
  $$
 \end{itemize}
We let $\mu_V$ denote the equilibrium measure or limit shape for each of the ensembles, which is a deterministic measure, approximating the (random) empirical measure $\frac{1}{N}\sum_{i=1}^N \delta_{\lambda_i}$ for large $N$. Explicitly, its density is given by:
\begin{itemize}
 \item For GOE, $\mu_V$ is the semicirle law of density $\frac{1}{2\pi} \sqrt{4-x^2}\, \mathbf 1_{[-2,2]}\, \dd x$.
 \item For LOE, $\mu_V$ is the Marchenko-Pastur law of density $\frac{1}{2\pi} \frac{\sqrt{(\lambda_+-x)(x-\lambda_-)}}{\gamma^2 x}$, where $(1\pm\gamma)^2$ and the parameter $\gamma$ satisfies $\frac{N}{S}=\gamma^2+O(1)$ as $N\to\infty$.
 \item For JOE, $\mu_V$ is the Wachter law of density $\frac{\tau_N}{2\pi } \frac{\sqrt{(x-\lambda_-)(\lambda_+-x)}}{x (1-x)} \mathbf 1_{[\lambda_-,\lambda_+]}\, \dd x$, where $\lambda_\pm=\left(\sqrt{\tau_M^{-1}(1-\tau_N^{-1})}\pm \sqrt{\tau_N^{-1}(1-\tau_M^{-1})}\right)^2$ and the parameters $\tau_N$, $\tau_N$ satisfy $\frac{S}{N}=\tau_N+O(1)$, $\frac{S}{M}=\tau_M+O(1)$ as $N\to\infty$.
\end{itemize}
In all three cases $[\lambda_-,\lambda_+]$ denotes the support of $\mu_V$. Let us also introduce the empirical and limiting Stieltjes transforms:
\begin{equation}
\label{eq_x53}
 m_N(z)=\frac{1}{N}\sum_{i=1}^N \frac{1}{z-\lambda_i}, \qquad m_V(z)=\int_{\lambda_-}^{\lambda_+} \frac{\mu_V(\dd x)}{z-x}.
\end{equation}
In all situation of interest the function $m_V(z)$ is explicit and satisfies \eqref{eq_Assumption_Stieltjes_imaginary}. We will present the formulas for $m_V(z)$ when they are needed.
\begin{theorem} \label{Theorem_from_Paul}
 For each of GOE, LOE, JOE, there exist constants $\tilde \eta>0$ and $C>0$ (depending on $\gamma$, $\tau_N$, $\tau_M$ in a continuous way), such that for all $N\ge 1$, $q=1,2,\dots$, and $z=E+\ii \eta$ with $0<\eta<\tilde \eta$, $\lambda_--\tilde\eta<E<\lambda_++\tilde\eta$, we have
 \begin{equation}
 \label{eq_bound_from_Paul}
   \E \left| m_N(z)-m_V(z)\right|^q \le \left(\frac{C q^2}{N\eta}\right)^q.
 \end{equation}
\end{theorem}
This theorem under the name ``Optimal Local Law'' can be found in \citet[Proposition 3.5]{bourgade2022optimal}, \cite[Remark 7.8]{huang2024convergence}, and \citet{bourgade2024optimal2}, where general $\beta$-ensembles \eqref{eq_beta_ensemble} are analyzed. Related
statements can also be found in many other sources, e.g., \citet[Lemma B.2 on arXiv or Lemma S5.2 in the supplement to the published version]{bao2019canonical} and \citet[Theorem 2.14]{FanYang} contain slightly weaker bounds for the CCA setting. 

\bigskip

Theorem \ref{Theorem_from_Paul} readily implies the condition \eqref{eq_Assumption_local_law} of Assumption \ref{Assumption_for_limit} for the case $h(x)=1$. Several more forms of $h(x)$ will also be needed.

\begin{corollary}\label{Corollary_local_with_square_root}
 For LOE or JOE ensembles, the bound \eqref{eq_bound_from_Paul} extends to the case $h(x)=\sqrt{x}$, in the following form: there exist $\tilde \eta>0$ and $C>0$, such that for all  $N\ge 1$, and $z=E+\ii \eta$ with $0<\eta<\tilde \eta$, $\max(\lambda_--\tilde\eta,0)<E<\lambda_++\tilde\eta$, we have
 \begin{equation}
 \label{eq_bound_from_Paul_2}
   \E \left[\frac{1}{N}\sum_{i=1}^N \frac{\sqrt{\lambda_i}}{z-\lambda_i}-\int_{\lambda_-}^{\lambda_+} \frac{\sqrt{x}\mu_V(\dd x)}{z-x}\right]^2 \le \frac{C}{N^2\eta^2}+ \frac{C \ln^2(N)}{N^2}.
 \end{equation}
\end{corollary}
\begin{remark}
 It is plausible that the second term in the right-hand side of \eqref{eq_bound_from_Paul_2} can be dropped. Yet, the current form is sufficient for us, as it clearly implies \eqref{eq_Assumption_local_law}.
\end{remark}
\begin{proof}[Proof of Corollary \ref{Corollary_local_with_square_root}] Note a deterministic bound $\left|\frac{1}{z-\lambda_i}\right|\le \frac{1}{\eta}$. This bound allows us to discard events, whose probabilities are decaying fast with $N$. We will do this several times.

Let us first assume that $\lambda_->0$ (which happens whenever $p$ is bounded away from $0$). By the large deviations principle for the support of the empirical measure of eigenvalues (see, e.g.,  \citet[Proposition 2.1]{borot2013asymptotic}) for each $\delta>0$, with probability exponentially close to $1$ for large $N$ all the eigenvalues $\lambda_1,\dots,\lambda_N$ lie in the positive segment $[\lambda_--\delta,\lambda_++\delta]\subset (0,+\infty)$. We choose a contour $\Gamma$ enclosing both this segment and a point $z$, intersecting real axis in positive points, and lying inside the domain of validity of \eqref{eq_bound_from_Paul}. We write using \eqref{eq_x53} and residue expansion of the contour integral:
 \begin{align}
 \label{eq_x54} \frac{1}{N}\sum_{i=1}^N \frac{\sqrt{\lambda_i}}{z-\lambda_i}&=\frac{1}{2\pi\ii}\oint_\Gamma \frac{m_N(u) \sqrt{u}}{u-z} \dd u - m_N(z)\sqrt{z},
   \\
 \label{eq_x55}  \int_{\lambda_-}^{\lambda_+} \frac{\sqrt{x}\mu_V(\dd x)}{z-x}&=\frac{1}{2\pi\ii}\oint_\Gamma \frac{m_V(u) \sqrt{u}}{u-z} \dd u - m_V(z)\sqrt{z}.
 \end{align}
 Using the inequality $|a+b|^2\le 2|a|^2+2|b|^2$, it is sufficient to separately analyze the first and second terms. The second moment of the difference of the second terms in \eqref{eq_x54} and \eqref{eq_x55} is upper-bounded by \eqref{eq_bound_from_Paul}. For the first terms, we split the integration contour into two parts: $\Gamma_1$ is the part at distance at most $1/N$ from the real axis and $\Gamma_2$ is the remaining part away from the real axis. Hence, it remains to  upper-bound
 \begin{equation}
 \label{eq_x56}
   \E\left|\int_{\Gamma_1} \frac{(m_N(u)-m_V(u)) \sqrt{u}}{u-z} \dd u\right|^2+ \E\left|\int_{\Gamma_2} \frac{(m_N(u)-m_V(u)) \sqrt{u}}{u-z} \dd u\right|^2.
 \end{equation}
 For the first term in \eqref{eq_x56}, the integrand is upper-bounded by a deterministic constant, because $\Gamma_1$ is away from $[\lambda_--\delta,\lambda_++\delta]$ by construction and we ignore the event when eigenvalues are not in this segment, as explained at the beginning of the proof. The length of $\Gamma_1$ is of order $1/N$, and therefore the first term is upper-bounded by $C_1 N^{-2}$.

 For the second term in \eqref{eq_x56}, using $\E|ab|\le (\E|a|^2)^{1/2} (\E|b|^2)^{1/2}$ and the bound \eqref{eq_bound_from_Paul}, we produce an estimate:
 \begin{multline}
  \E \int_{\Gamma_2} \frac{(m_N(u)-m_V(u)) \sqrt{u}}{u-z} \dd u\int_{\Gamma_2} \frac{\overline{(m_N(v)-m_V(v)) \sqrt{v}}}{\overline{v}-\overline{z}} \dd \overline{v}\\ =   \iint_{\Gamma_2\times \Gamma_2} \E\left[(m_N(u)-m_V(u))\overline{(m_N(v)-m_V(v))}\right] \frac{ \sqrt{u}\overline{\sqrt{v}}}{(u-z)(\overline{v}-\overline{z})}  \dd u \dd \overline{v}\\ \le \frac{C_2}{N^2} \iint_{\Gamma_2\times \Gamma_2} \frac{1}{|\mathrm{Re}\, u| |\mathrm{Re}\, v|} |\dd u| |\dd v|\le C_3\frac{\ln^2(N)}{N^2}.
 \end{multline}
 Summing three upper bounds, we arrive at \eqref{eq_bound_from_Paul_2}.

 In the case $\lambda_-=0$ we need to be more careful, because if we argue in the same way as for $\lambda_->0$, then $\Gamma$ has to loop around $0$ and $\sqrt{u}$ is no longer holomorphic on the contour. Instead, we choose the contour $\Gamma$ passing directly through $0$ and enclosing $(0,\lambda_++\delta)$. Compared to the previous argument, it is no longer true that on the part of the contour at distance $\frac{1}{N}$ from 0 the integrand is upper-bounded by a constant. However, in this case the splitting of the contour $\Gamma$ is not required at all: the $1/\eta$ singularity of the bound \eqref{eq_bound_from_Paul} is compensated by $\sqrt{u}$ factor in the integrand, and the resulting expression $\eta^{-1/2}$ is integrable at $\eta=0$.
\end{proof}

\begin{corollary}\label{Corollary_local_with_square_root_2}
 For JOE ensembles, the bound \eqref{eq_bound_from_Paul} further extends to the cases $h(x)=\sqrt{1-x}$ and $h(x)=\sqrt{x(1-x)}$: there exist $\tilde \eta>0$ and $C>0$, such that for all  $N\ge 1$, and $z=E+\ii \eta$ with $0<\eta<\tilde \eta$, $\max(\lambda_--\tilde\eta,0)<E<\min(\lambda_++\tilde\eta,1)$, we have
 \begin{equation}
 \label{eq_bound_from_Paul_3}
   \E \left[\frac{1}{N}\sum_{i=1}^N \frac{h(\lambda_i)}{z-\lambda_i}-\int_{\lambda_-}^{\lambda_+} \frac{h(x)\mu_V(\dd x)}{z-x}\right]^2 \le \frac{C}{N^2\eta^2}+ \frac{C \ln^2(N)}{N^2}.
 \end{equation}
\end{corollary}
The proof for Corollary \ref{Corollary_local_with_square_root_2} is the same as for Corollary \ref{Corollary_local_with_square_root} and is omitted.

 We also need to check the same condition for finite-rank perturbations of $\beta$--ensembles \eqref{eq_beta_ensemble}, which, in fact, follows automatically as the following lemma asserts.

\begin{lemma} \label{Lemma_preservation_local_law} There exist two positive constants $e_1$ and $e_2$ such that the following holds. Suppose that $\lambda_i^N$, $1\le i \le N$, satisfies  \eqref{eq_Assumption_local_law} of Assumption \ref{Assumption_for_limit} with $C=C_1$ and $\mu_i^N$ interlaces with it, which means either
\begin{equation}
\label{eq_x28i}
 \mu_1^N\ge \lambda_1^N\ge \mu_2^N\dots\ge \mu_N^N\ge \lambda_N^N \qquad \text{or}\qquad  \lambda_1^N\ge \mu_1^N\ge \lambda_2^N\dots\ge \lambda_N^N\ge \mu_N^N.
\end{equation}
Then $\mu_i^N$, $1\le i \le N$, satisfies \eqref{eq_Assumption_local_law} with $C=e_1C_1+e_2$. This holds with $h(x)=1$; additionally with $h(x)=\sqrt{x}$ if the eigenvalues are positive; and with $h(x)=\sqrt{1-x}$ and $h(x)=\sqrt{x(1-x)}$ if the eigenvalues lie in the interval $[0,1]$.
\end{lemma}
\begin{proof} We only consider the case $h(x)=1$, as other cases are similar.
 We would like to upper bound the absolute value of the difference,
 $
  \frac{1}{N}\sum_{i=1}^N \frac{1}{z-\lambda_{i;N}}- \frac{1}{N}\sum_{i=1}^N \frac{1}{z-\mu_{i;N}}.
 $
  The real part of the difference is
 \begin{equation}
 \label{eq_x27}
  \frac{1}{N}\sum_{i=1}^N \frac{\mathrm{Re} z-\lambda_{i;N}}{(\mathrm{Re} z-\lambda_{i;N})^2+(\mathrm{Im} z)^2}- \frac{1}{N}\sum_{i=1}^N \frac{\mathrm{Re} z-\mu_{i;N}}{{(\mathrm{Re} z-\mu_{i;N})^2+(\mathrm{Im} z)^2}}.
 \end{equation}
 Note that the function $x\mapsto \frac{\mathrm{Re} z-x}{(\mathrm{Re} z-x)^2+(\mathrm{Im} z)^2}$ is monotone on each of the three segments $x\in(-\infty, \mathrm{Re} z-|\mathrm{Im} z|]$, $x\in (\mathrm{Re} z-|\mathrm{Im} z|, \mathrm{Re} z+|\mathrm{Im} z|)$, $x\in[\mathrm{Re} z+|\mathrm{Im} z|, +\infty)$. We split the first sum in \eqref{eq_x27} into three parts, corresponding to $\lambda_{i;N}$ in each of these three segments. Let $N_1$, $N_2$, $N_3$ be the numbers of terms in the corresponding parts. Using the interlacement \eqref{eq_x28i} we match the $N_1$ terms in the first segment to $N_1-1$ terms in the second sum in \eqref{eq_x27} corresponding to $\mu_{i;N}$ which are between those $N_1$ terms.
 By monotonicity, the difference between the two subsums is at most $\frac{1}{N}\max_{x\in\mathbb R}  \frac{\mathrm{Re} z-x}{(\mathrm{Re} z-x)^2+(\mathrm{Im} z)^2}=\frac{1}{2 N |\mathrm{Im} z|}$. Similarly, the $N_2$ terms in the second segment are matched to $N_2-1$ terms in the second sum in \eqref{eq_x27} again leading to the difference between the two subsums at most $\frac{1}{2 N |\mathrm{Im} z|}$. The same is true for the remaining $N_3$ terms in the third segment, matched to $N_3-1$ terms in the second sum in \eqref{eq_x27}. Note that as a result of this procedure, three terms in the second sum of \eqref{eq_x27} remained unmatched. Their contribution is upper bounded by $3\cdot \frac{1}{2 N |\mathrm{Im} z|}$. We conclude that
 \begin{equation*}
  \left|\mathrm{Re} \left(\frac{1}{N}\sum_{i=1}^N \frac{1}{z-\lambda_{i;N}}- \frac{1}{N}\sum_{i=1}^N \frac{1}{z-\mu_{i;N}}\right)\right|\le \frac{3}{ N |\mathrm{Im} z|}.
 \end{equation*}
 For the imaginary part, the argument is similar: we need to bound
 \begin{equation*}
  \frac{1}{N}\sum_{i=1}^N \frac{\mathrm{Im} z}{(\mathrm{Re} z-\lambda_{i;N})^2+(\mathrm{Im} z)^2}- \frac{1}{N}\sum_{i=1}^N \frac{\mathrm{Im} z}{{(\mathrm{Re} z-\mu_{i;N})^2+(\mathrm{Im} z)^2}},
 \end{equation*}
 which is done by noticing that the function $x\mapsto \frac{\mathrm{Re} z-x}{(\mathrm{Re} z-x)^2+(\mathrm{Im} z)^2}$ is monotone on each of the two segments $x\in(-\infty, \mathrm{Re} z]$, $x\in (\mathrm{Re} z, +\infty)$. Repeating the real part arguments, we get
 \begin{equation*}
  \left|\mathrm{Im} \left(\frac{1}{N}\sum_{i=1}^N \frac{1}{z-\lambda_{i;N}}- \frac{1}{N}\sum_{i=1}^N \frac{1}{z-\mu_{i;N}}\right)\right|\le \frac{3}{ N |\mathrm{Im} z|}.
 \end{equation*}
 Hence, using the inequality $|a+b|^2\le 2|a|^2+2|b|^2$, we get
\begin{multline*}
 \left|\frac{1}{N}\sum_{i=1}^N \frac{1}{\lambda_++\ii R-\mu_{i;N}} -m(\lambda_++ \ii R)\right|^2\\\le
 2 \left|\frac{1}{N}\sum_{i=1}^N \frac{1}{\lambda_++\ii R-\lambda_{i;N}} -m(\lambda_++ \ii R)\right|^2 + \frac{36}{N^2 R^2}.
\end{multline*}
Taking expectation and using \eqref{eq_Assumption_local_law} for $\lambda_{i;N}$, we are done.
\end{proof}

\section{Appendix B: Applications of general theorem to four models}
\label{Section_proof_of asymptotic_approximations}

In this section we use Theorem \ref{Theorem_random_convergence} to prove Theorem \ref{Theorem_main_convergence_statement}. The proofs for all four settings follow similar outlines, and we present the most detailed exposition for spiked Wigner matrices. The difficulty of the argument increases steadily from the Wigner to the CCA.

\subsection{Spiked Wigner matrices} \label{Section_spiked_Wigner_proof} Our first task is to introduce an equation governing the change of the eigenvalues of a symmetric matrix under a rank one perturbation.

Take $N\times N$ real symmetric matrix $\B$ with eigenvalues $\lambda_1\ge\dots\ge\lambda_N$ and normalized eigenvectors $\u_i$, $1\le i \le N$, satisfying $\langle \u_i, \u_j \rangle =\delta_{i=j}$. In addition, take a column-vector $\u^*$. For  $\theta\ne 0$, set
$$
 \A=\theta\cdot  \u^* (\u^*)^\T + \B.
$$
\begin{proposition} \label{Proposition_equation_spiked_Wigner}
 For each eigenvalue $a$ of $\A$, either
\begin{equation}
\label{eq_spiked_sym}
 \frac{1}{\theta }=\sum_{i=1}^N \frac{  \langle \u_i, \u^*\rangle^2}{a-\lambda_i};
\end{equation}
or $a=\lambda_j$ for $1\le j \le N$, where $\lambda_j$ is an eigenvalue of $\B$ of multiplicity one, $\langle \u^*,\u_j\rangle=0$, and the equation \eqref{eq_spiked_sym} holds with the $j$-th term excluded; or $a=\lambda_j$, where $\lambda_j$ is an eigenvalue of $\B$ of multiplicity larger than one.
\end{proposition}
We omit the proof, see, e.g., \citet{jones1978eigenvalue} or \citet{arbenz1988restricted}.
\begin{remark}
 We mostly use Proposition \ref{Proposition_equation_spiked_Wigner} in the generic situations of all distinct $\lambda_i$ and all non-zero $\langle \u_i, \u^*\rangle$. Hence \eqref{eq_spiked_sym} will hold.
\end{remark}

\begin{corollary} \label{Corollary_Wigner_interlacement}
 The eigenvalues of $\A$ and $\B$ interlace: if $\mu_1\ge \dots\ge\mu_N$ are the eigenvalues of $\A$, then $\mu_1\ge \lambda_1\ge \mu_2\ge\dots$ for $\theta>0$ and $\lambda_1\ge\mu_1\ge\lambda_2\ge\dots$ for $\theta<0$.
\end{corollary}
\begin{proof}
 Assuming that all $\lambda_i$ are distinct and for all $i$ we have $\langle \u_i, \u^*\rangle\ne 0$ (general case is obtained by a limit transition), note that \eqref{eq_spiked_sym} is a polynomial equation on $a$ of degree $N$, and therefore has $N$ roots which are  $\mu_1$, \dots, $\mu_N$. Tracking the sign changes of the difference between the RHS and the LHS on intervals $(-\infty,\lambda_N)$, $(\lambda_N,\lambda_{N-1})$,\dots, $(\lambda_1,+\infty)$, we localize the roots, so that they satisfy the desired interlacements.
\end{proof}

In the rest of the section we use \eqref{eq_spiked_sym} to produce an inductive proof of Theorem \ref{Theorem_main_convergence_statement} for the spiked Wigner model by growing $r$. On each step, we need to know the law of the scalar products $\langle \u_i, \u^*\rangle$ appearing in the equation. This is simplified by the following observation:

\begin{lemma} \label{Lemma_introduce_independence}
 In \eqref{eq_Spiked_Wigner}, replacing  the deterministic orthonormal vectors $\u_i^*$ by $r$ independent vectors uniformly distributed on the $N$--dimensional unit sphere (independent of $\mathcal E$) leaves Theorem \ref{Theorem_main_convergence_statement} unchanged.
\end{lemma}
\begin{proof}
 We first note that the asymptotics in Theorem \ref{Theorem_main_convergence_statement} does not depend on the choice of $\u_i^*$ in \eqref{eq_Spiked_Wigner} as long as they are orthogonal to each other. Indeed, any $r$ orthogonal vectors can be obtained from any other $r$ orthogonal vectors by an orthogonal transformation of the space. Such transformation does not change the eigenvalues of $\A$, and also does not change the probability distribution of the matrix $\mathcal E$ (which uses Gaussianity of its matrix elements).

 Now let $\v_1^*,\dots,\v_r^*$ be $r$ independent vectors uniformly distributed on the unit sphere. We consider the matrix $M_r=\sum_{i=1}^r \theta_i \v_i^* (\v_i^*)^\T$ and would like to decompose it as $M_r=\sum_{i=1}^r \theta_i' \w_i^* (\w_i^*)^\T$ with orthonormal vectors $\w_i^*$; $\v_i^*$ were not orthonormal. Clearly, $\theta_i'$ are non-zero eigenvalues of $M_r$ and $\w_i^*$ are corresponding eigenvectors. We claim that
 \begin{equation}
 \label{eq_x31}
  \theta_i=\theta'_i+O\left(\frac{1}{N}\right),\qquad N\to\infty.
 \end{equation}
 Note that \eqref{eq_x31} implies the statement of Lemma \ref{Lemma_introduce_independence}, because addition of $O(1/N)$ does not change any of the asymptotics statements of Theorem \ref{Theorem_main_convergence_statement}. Hence, it remains to prove \eqref{eq_x31}. This can be done by induction on $r$. Using \eqref{eq_spiked_sym} with $\B=M_{r-1}$, the non-zero eigenvalues of $M_r$ solve an equation
 \begin{equation}
 \label{eq_x32}
 \frac{1}{\theta_r }=\sum_{i=1}^{r-1} \frac{  \langle \w_i, \v^*_r\rangle^2}{a-\theta'_i}+ \sum_{i=r}^N \frac{  \langle \w_i, \v^*_r\rangle^2}{a},
 \end{equation}
 where $\w_i$ are eigenvectors of $M_{r-1}$ and $\theta'_i$ are non-zero eigenvalues. Representing the unit vector $\v^*_r$ as a vector with i.i.d.\ $\mathcal N(0,1)$ components divided by its length, using the Law of Large Numbers and the induction hypothesis, the equation is rewritten as
 \begin{equation}
 \label{eq_x33}
 \frac{1}{\theta_r }=\frac{1}{N}\sum_{i=1}^{r-1} \frac{\chi_i^2(1+O(\tfrac{1}{\sqrt{N}}))}{a-\theta_i+O(\tfrac{1}{N})}+ \frac{1+O(\tfrac{1}{N})}{a},
 \end{equation}
 where $\chi_i$ are i.i.d.\ $\mathcal N(0,1)$ random variables. \eqref{eq_x33} is a polynomial equation on $a$ of degree $r$, which clearly has $r$ roots of the form $\theta_i+O(\tfrac{1}{N})$, thus proving \eqref{eq_x31}.
\end{proof}

The induction can proceed in various orders of spike addition, and we choose to first add all subcritical and supercritical spikes, and then add the critical spike with index $q$ at the very end. Hence, as an intermediate statement we have the following:

\begin{proposition} \label{Proposition_Wigner_far_spikes}
 Consider the spiked Wigner model $\A= \sum_{i=1}^r \theta_i \cdot \u^*_i (\u^*_i)^\T + \mathcal E$ of Section \ref{Section_spiked_Wigner}, where $\mathcal E=\frac{1}{\sqrt{2 N}}(\mathcal Z+\mathcal Z^\T)$, with $\mathcal Z$ being $N\times N$ matrix of i.i.d.\ $\mathcal N(0,1)$, and $\theta_1>\theta_2>\dots>\theta_r$ split into two groups: $\theta_1,\dots,\theta_{q-1}>\theta^c=1$ and $\theta_q,\dots,\theta_r<\theta^c=1$. Then, in the sense of convergence in joint distribution and using \eqref{eq_Spiked_Wigner_params}:
  \begin{align}
 \label{eq_x29}\lim_{N\to\infty} \sqrt{N}(\lambda_i-\lambda(\theta_i))&\stackrel{d}{=} \mathcal N(0, V(\theta_i)),\qquad 1\le i \le q-1,\\
 \label{eq_x30} \lim_{N\to\infty} N^{2/3}(\lambda_{i}-2)&\stackrel{d}{=} \aa_{i-q+1},\qquad i\geq q,
 \end{align}
 where $ \mathcal N(0, V(\theta_i))$ are independent over $i$ and with  points of the Airy$_1$ point process  $\{\aa_j\}_{j\ge 1}$. In addition, \eqref{eq_Assumption_local_law} with $h(x)=1$ holds for the eigenvalues of $\A$.
\end{proposition}
\begin{proof}
  The final statement, \eqref{eq_Assumption_local_law} is proven by induction on $r$, starting from Theorem \ref{Theorem_from_Paul} for $r=0$, and using Lemma \ref{Lemma_preservation_local_law} with Corollary \ref{Corollary_Wigner_interlacement} for the induction step. Here
\begin{equation}
\label{eq_semicircle_repeat}
 \mu(x)\dd x=\frac{1}{2\pi} \sqrt{4-x^2}\, \mathbf 1_{[-2,2]}\, \dd x, \qquad m(z)=\frac{z-\sqrt{z^2-4}}{2},
\end{equation}
and the constants of \eqref{eq_Assumption_Stieltjes_imaginary} are computed as $\s=1$, $\m=1$, $\lambda_+=2$.

  Statements close to \eqref{eq_x29}, \eqref{eq_x30} are known from \citet{capitaine2012central,benaych2011fluctuations,knowles2013isotropic};  we sketch the proof in order to be self-contained.

  We proceed by induction on $r$; the base case $r=0$ is Proposition \ref{Proposition_Airy_Gauss}.
  We analyze the equation \eqref{eq_spiked_sym} with $\theta=\theta_1$ and $\lambda_i$ being eigenvalues of $\A= \sum_{i=2}^{r} \theta_i \cdot \u^*_i (\u^*_i)^\T + \mathcal E$. We start by considering the case $\theta_1>\theta_c$ (i.e., $q>1$) and comment on the changes for the case $q=1$ at the end. We first look at the interval $a\in [\lambda_1,+\infty)$. The right-hand side of the equation \eqref{eq_spiked_sym} is a monotone-decreasing function of $a$ in this interval, changing from $+\infty$ to $0$. Hence, there is a unique $\hat a\in (\lambda_1,+\infty)$ solving \eqref{eq_spiked_sym}. In order to locate this $\hat a$, we make an asymptotic expansion of the equation. Using Lemma \ref{Lemma_introduce_independence}, we can assume $\u^*$ in \eqref{eq_spiked_sym} to be a uniformly random vector on the unit sphere, independent from everything else. Then $\langle \u_i, \u^*\rangle$ are coordinates of a similar vector, because $\u_i$ are orthonormal. Hence, the equation \eqref{eq_spiked_sym} is recast as
  \begin{equation}
  \label{eq_x35}
   \frac{1}{\theta_1}= \sum_{j=1}^N \frac{  \zeta_j^2}{a-\lambda_j}, \qquad (\zeta_1,\dots,\zeta_N) \sim \text{Uniform on unit sphere }\mathbb S^{N-1}.
  \end{equation}
  We further approximate the RHS as $N\to\infty$ for $a>\lambda_j$, by writing it as
  \begin{equation}
  \label{eq_x34}
   \frac{1}{N}\sum_{j=1}^N \frac{1}{a-\lambda_j} + \sum_{j=1}^N \frac{  \zeta_j^2-\tfrac{1}{N}}{a-\lambda_j}.
  \end{equation}
  Let us first imagine that $(\lambda_1,\dots,\lambda_N)$ are eigenvalues of GOE, i.e., of the matrix ${\mathcal E=\frac{1}{\sqrt{2 N}}(\mathcal Z+\mathcal Z^\T)}$ from the statement of the proposition. Then, using the semicircle law (see Theorem \ref{Theorem_from_Paul} or \citet[Chapter 2 and Theorem 9.2]{Bai_Silverstein}), the first term in \eqref{eq_x34} becomes
  \begin{equation}
    \int_{-2}^2 \frac{1}{2\pi}\sqrt{4-x^2} \frac{\dd x}{a -x}+ O\left(\frac{1}{N}\right)=\frac{1}{2}\left(a-\sqrt{a^2-4}\right)+ O\left(\frac{1}{N}\right).
  \end{equation}
  For the second term, note that $ \zeta_j^2-\tfrac{1}{N}$, $j=1,2,\dots,N$, are weakly dependent mean $0$ random variables. Hence, conditionally on $\lambda_1,\dots,\lambda_N$, CLT applies and the sum is asymptotically Gaussian. The limit of the variance can be computed using
  $$
   \E \left( \zeta_j^2-\tfrac{1}{N}\right)^2=\frac{2}{N^2} + o\left(\frac{1}{N^2}\right),\qquad \E \left( \zeta_i^2-\tfrac{1}{N}\right)\left( \zeta_j^2-\tfrac{1}{N}\right)=-\frac{2}{N^3}+o\left(\frac{1}{N^3}\right),
  $$
  where the first identity comes from writing $\zeta_j^2\stackrel{d}{=} \frac{\xi_j^2}{\sum_{\ell=1}^N \xi_\ell^2}$ with i.i.d.~$\mathcal N(0,1)$ random variables $\xi_j$ and the second identity comes from combining the first one with $\E \left[\sum_{j=1}^N (\zeta_j^2-\frac{1}{N})\right]^2=\E [0]^2=0$. Hence,
  $$
   \E\left[\left(\sum_{j=1}^N \frac{  \zeta_j^2-\tfrac{1}{N}}{a-\lambda_j}\right)^2\middle|\lambda_1,\dots,\lambda_N \right]=\frac{2}{N^2}\sum_{j=1}^N \frac{1}{(a-\lambda_j)^2}-\frac{2}{N^3} \left[\sum_{j=1}^N \frac{1}{a-\lambda_j}\right]^2+o\left(\frac{1}{N}\right),
  $$
  and plugging in the semicircle law, we further approximate the variance as
  \begin{multline*}
   \frac{2}{N}\left(\int_{-2}^2 \frac{1}{2\pi}\sqrt{4-x^2} \frac{\dd x}{(z -x)^2}-\left[\int_{-2}^2 \frac{1}{2\pi}\sqrt{4-x^2} \frac{\dd x}{z -x}\right]^2 \right)+o\left(\frac{1}{N}\right)\\=
   \frac{1}{N}\left(-1+\frac{a}{\sqrt{a^2-4}}\right)-\frac{1}{2N}\left(a-\sqrt{a^2-4}\right)^2+o\left(\frac{1}{N}\right)
   =\frac{(a-\sqrt{a^2-4})^3}{4 N \sqrt{a^2-4}}+o\left(\frac{1}{N}\right).
  \end{multline*}
  We conclude that on the interval $(\lambda_1,+\infty)$, the equation \eqref{eq_x35} is approximated by
  \begin{equation}
  \label{eq_x36}
   \frac{1}{\theta_1}= \frac{1}{2}\left(a-\sqrt{a^2-4}\right)+ \frac{1}{\sqrt{N}}\sqrt{\frac{(a-\sqrt{a^2-4})^3}{4 \sqrt{a^2-4}}}\mathcal N(0,1)+ o\left(\frac{1}{\sqrt{N}}\right).
  \end{equation}
  Recall that this approximation was obtained assuming $\lambda_i$ to be eigenvalues of GOE. What we actually need for them is instead to be coming from a deformation of GOE, i.e., to be the eigenvalues of $\sum_{i=2}^{r} \theta_i \cdot \u^*_i (\u^*_i)^\T + \mathcal E$. The approximation \eqref{eq_x36} remains true for such a finite rank deformation, as follows from the interlacements of Corollary \ref{Corollary_Wigner_interlacement}, by repeating the arguments in the proof of Lemma \ref{Lemma_preservation_local_law}.

  Solving the equation \eqref{eq_x36} as $N\to\infty$, we get
  \begin{equation}
  \label{eq_x36_2}
     a=\theta_1+\frac{1}{\theta_1} + \mathcal N(0,1)  \frac{\sqrt{2}}{\sqrt{N}} \sqrt{\frac{\theta_1^2-1}{\theta_1^2}}+ o\left(\frac{1}{\sqrt{N}}\right),
  \end{equation}
  which matches \eqref{eq_x29} and proves the desired asymptotics for the largest eigenvalue.

  For the remaining eigenvalues, the idea is to show that the $(i+1)$st largest root of \eqref{eq_x35} is very close to $\lambda_i$ and then use the induction hypothesis. By Corollary \ref{Corollary_Wigner_interlacement}, the $(i+1)$st largest root is the unique root in the interval $(\lambda_{i+1},\lambda_i)$ and our task is to show that it is much closer to the right end-point of this segment rather than to the left end-point.

  {\bf Case 1:} $\theta_{i+1}>\theta^c$, so that $\lambda_i$ is bounded away from $\lambda_+$. Let us approximate \eqref{eq_x35} as $N\to\infty$ near $\lambda_i$. For the sum over $j\ne i$, the same arguments leading to \eqref{eq_x36} continue to hold and the equation turns into:
  \begin{equation}\label{eq_x37}
    \frac{1}{\theta_1}-\frac{1}{2}\bigl(a-\sqrt{a^2-4}\bigr)- \frac{1}{\sqrt{N}}\sqrt{\frac{(a-\sqrt{a^2-4})^3}{4 \sqrt{a^2-4}}}\mathcal N(0,1)= \frac{  \zeta_i^2}{a-\lambda_i}+ o\left(\frac{1}{\sqrt{N}}\right).
  \end{equation}
  Note that for $a$ close to $\lambda_i$, the value of $\frac{1}{2}\bigl(a-\sqrt{a^2-4}\bigr)$ is close to $1/\theta_{i+1}>1/\theta_1$ by induction assumption \eqref{eq_x29}. Hence, for large $N$ the left-hand side of \eqref{eq_x37} is negative and bounded away from zero. On the other hand $\zeta_i^2=O(\tfrac{1}{N})$. We conclude that $a$ should be smaller than $\lambda_i$, at distance $O(\tfrac{1}{N})$, in order for \eqref{eq_x37} to hold. Hence, this root $a$ satisfies \eqref{eq_x29}.

  {\bf Case 2:} $\theta_{i+1}<\theta^c$, so that $\lambda_i$ is close to $\lambda_+$. We approximate the right-hand side of \eqref{eq_x35} as $N\to\infty$ near $\lambda_+$ by writing $\zeta_j^2\stackrel{d}{=} \frac{\xi_j^2}{\sum_{\ell=1}^N \xi_\ell^2}$ with i.i.d.\ $\mathcal N(0,1)$ random variables $\xi_j$ and then using Theorem \ref{Theorem_as_convergence} with $\s=\m=1$, $h(x)=1$, whose assumptions hold by the induction hypothesis. The right-hand side of \eqref{eq_x35} has asymptotics $1 + N^{-1/3} \cdot \mathcal G(b)+ o(N^{-1/3})$ where the rescaled variable is $b=N^{2/3}(a-2)$. The interval $a\in (\lambda_{i+1},\lambda_i)$ turns asymptotically into $b\in (\aa_{i-q+2},\aa_{i-q+1})$, and the equation becomes $\mathcal G(b)=N^{1/3} (1/\theta_1-1)+o(N^{-1/3})$. Since $1/\theta_1<1/\theta^c=1$, we are looking for a value of $b\in (\aa_{i-q+2},\aa_{i-q+1})$, where $\mathcal G(b)$ would be large and negative. Clearly, then $b$ needs to be close to the right end-point of the segment, i.e., to $\aa_{i-q+1}$ and we achieve \eqref{eq_x30}.

  To finish the proof it remains to analize the case $\theta_1<\theta^c$, i.e., $q=1$. The argument then repeats the just presented Case 2, with the only difference being that we now look for a value of $b\in (\aa_{i-q+2},\aa_{i-q+1})$, where $\mathcal G(b)$ would be large and \emph{positive}. Then $b$ needs to be close to the left end-point of the segment\footnote{For $i=q=1$ the segment of interest is $(\aa_1,+\infty)$.}, which is $\aa_{i-q+2}=\aa_{i+1}$ and we arrive at \eqref{eq_x30}. \end{proof}

\begin{proof}[Proof of Theorem \ref{Theorem_main_convergence_statement} for the spiked Wigner model of Section \ref{Section_spiked_Wigner}] We note that the constants \eqref{eq_transition_parameters} simplify to $\kappa_1=\kappa_2=1$. We analyze the equation \eqref{eq_spiked_sym} for $\theta=\theta_q$, $(\lambda_i,\u_i)_{i=1}^N$ being eigenvalues and eigenvectors of
$$
\A= \sum_{\begin{smallmatrix}1\le i\le r\\ i\ne q\end{smallmatrix}} \theta_i \cdot \u^*_i (\u^*_i)^\T + \mathcal E,
$$
and $\u^*$ being a uniformly random unit vector independent from the rest. By Proposition \ref{Proposition_equation_spiked_Wigner} and Lemma \ref{Lemma_introduce_independence}  the $N$ solutions of this equation denoted $a_1\ge a_2\ge \dots\ge a_N$, are precisely the eigenvalues of Theorem \ref{Theorem_main_convergence_statement} (which were $\lambda_i$ there) and we need to establish \eqref{eq_Gaussian_limit} and \eqref{eq_Transition_limit}. We rely on Proposition \ref{Proposition_Wigner_far_spikes} for the asymptotics of $\lambda_1,\lambda_2,\dots$ and
therefore know that the largest ones satisfy \eqref{eq_Gaussian_limit} (equivalently, \eqref{eq_x29}), while the next ones converge to the points of the Airy$_1$ point process by \eqref{eq_x30}.

We first claim that
\begin{equation}
 \label{eq_x38} a_i - \lambda_i=O\left(\frac{1}{N}\right), \qquad 1\le i \le q-1.
\end{equation}
The proof of \eqref{eq_x38} is exactly the same as Case 1 in the proof of Proposition \ref{Proposition_Wigner_far_spikes}, i.e.,  using \eqref{eq_x37} (with $\theta_1$ replaced by $\theta_q$ this time). \eqref{eq_x38} combined with \eqref{eq_x29} implies the desired asymptotics \eqref{eq_Gaussian_limit} for $a_1,\dots,a_{q-1}$.

It remains to investigate $a_q$. Because $\u^*$ is uniformly random and $[\u_i]_{i=1}^N$ are orthonormal, the vector $\langle \u_j, \u^*\rangle^2$, $1\le j\le N$, has the same distribution as $\frac{\xi_j^2}{\sum_{\ell=1}^N \xi_\ell^2}$ with i.i.d.\ $\mathcal N(0,1)$ random variables $\xi_j$. Hence, recalling $\theta_q=\theta^c+N^{-1/3} \tilde \theta=1+N^{-1/3}\tilde \theta$, \eqref{eq_spiked_sym} becomes:
\begin{equation}
\label{eq_x39}
 \frac{1}{1+N^{-1/3}\tilde \theta }=\frac{1}{\sum_{\ell=1}^N \xi_\ell^2} \sum_{i=q}^{N} \frac{  \xi_i^2}{a-\lambda_i}+ \frac{1}{\sum_{\ell=1}^N \xi_\ell^2} \sum_{i=1}^{q-1} \frac{  \xi_i^2}{a-\lambda_i}.
\end{equation}
By Corollary \ref{Corollary_Wigner_interlacement}, $a_q$ is the unique root of \eqref{eq_x39} in the interval\footnote{If $q=1$, then we should set $\lambda_{q-1}=+\infty$.} $(\lambda_{q},\lambda_{q-1})$. In order to locate this root, we change the variables
\begin{equation}
\label{eq_x40}
 b=N^{2/3}(a-2),\qquad a=2+N^{-2/3}b,
\end{equation}
and investigate the asymptotics of \eqref{eq_x39} for finite $b$, i.e., for $a$ close to $\lambda_+=2$. Because $\lambda_1,\dots,\lambda_{q-1}$ are bounded away from $2$ by Proposition \ref{Proposition_Wigner_far_spikes}, the second sum in the right-hand side of \eqref{eq_x39} is $O(\tfrac{1}{N})$ and can be omitted. For the first sum we apply Theorem \ref{Theorem_as_convergence} with $\s=\m=1$, $h(x)=1$. Hence, \eqref{eq_x39} turns into
\begin{equation}
\label{eq_x40}
 \frac{1}{1+N^{-1/3}\tilde \theta }-1=N^{-1/3} \mathcal G(b)+ o\left(N^{-1/3}\right).
\end{equation}
Taylor expanding the left-hand side in small $N^{-1/3}\tilde \theta$ and using Corollary \ref{Corollary_root}, we conclude that the desired root $b$ converges towards the largest root of the equation $-\tilde \theta=\mathcal G(b)$. Comparing with Definition \ref{Definition_Transition_function}, we are done.
\end{proof}

\subsection{Spiked covariance model} \label{Section_spiked_covariance_proof}  We explain the new parts, but are brief on the technical details repeating the previous section. At the end  we also analyze $\theta_i\to\infty$ limits and explain how those connect to the strong signals. We start with an analogue of Proposition \ref{Proposition_equation_spiked_Wigner}.

Suppose that we are given $N\times S$ matrix $\U$, in which the rows are indexed by $i=0,1,\dots,N-1$ and $S\ge N$. We let $\lambda_*$ denote the squared length of the zeroth row of $\U$ (treated as an $S$--dimensional vector) and let $\v^*$ be the unit vector in the direction of this row.
Let $\widetilde \U$ denote the $(N-1)\times S$ matrix formed by rows $i=1,2,\dots,N-1$ of $\U$.

We would like to connect the singular values and singular vectors of $\U$ to: singular values and vectors of $\widetilde \U$,\,  $\lambda_*$,\, and $\v^*$.  We let $(\u_i,\v_i, \sqrt{\lambda_i})$, $1\le i \le N-1$, be the left singular vector (of $(N-1)\times 1$ dimensions), right singular vector (of $S\times 1$ dimensions), singular value triplets for $\widetilde \U$, which means that
$$
 \widetilde \U=\begin{pmatrix} \u_1; \u_2; \dots; \u_{N-1} \end{pmatrix} \begin{pmatrix} \sqrt{\lambda_1}& 0 & \dots \\ 0 & \sqrt{\lambda_2} & 0 & \\ & 0 & \ddots\\ & & 0 &\sqrt{\lambda_{N-1}}\end{pmatrix} \begin{pmatrix} \v_1^\T\\ \v_2^\T \\ \vdots \\ \v_{N-1}^\T\end{pmatrix}=\sum_{i=1}^{N-1} \sqrt{\lambda_i} \u_i \v_i^\T
$$
and $\langle \u_i, \u_j\rangle=\delta_{i=j}$, $\langle \v_i, \v_j\rangle=\delta_{i=j}$. We order the singular values so that $ \lambda_1 \ge \dots \ge \lambda_{N-1}\ge 0$.

\begin{proposition} \label{Proposition_PCA_equation} Suppose that $a\ge 0$ is an eigenvalue of $\U \U^\T$. Then either
\begin{equation}
\label{eq_PCA_equation}
 \lambda_*  \left(1+ \sum_{i=1}^{N-1} \frac{\lambda_i \langle \v^*, \v_i\rangle^2}{a-\lambda_i}\right)=a,
\end{equation}
or $a=\lambda_j$ for $1\le j \le N$, where $\sqrt{\lambda_j}$ is a singular value of multiplicity one, $ \langle \v^*, \v_j\rangle=0$, and \eqref{eq_PCA_equation} holds with the $j$th term excluded, or $a=\lambda_j$, where $\sqrt{\lambda_j}$ is a singular value of multiplicity greater than $1$.
\end{proposition}
We omit the proof, see, e.g.,  \citet[Appendix A]{BG_CCA} (which has $\lambda_i$ squared and $\v$'s replaced with $\u$'s).

In order to see how Proposition \ref{Proposition_PCA_equation} is relevant for the setting of Section \ref{Section_spiked_covariance}, we choose one index $1\le k \le r$ and a deterministic orthogonal matrix $O$, which maps $\u_k^*$ into the zeroth basis vector and orthogonal complement of $\u_k^*$ to the span of the basis vectors with labels $1,2,\dots,N-1$, so that
$$
 O \Omega O^\T=\begin{pmatrix} \theta_k& 0 & \dots& 0\\ 0  \\ \vdots & & \Omega'\\ 0 \end{pmatrix},
$$
where $\Omega'$ is $(N-1)\times (N-1)$ symmetric matrix with eigenvalues $\{\theta_i\}_{i\ne k}$ and $1$ of multiplicity $N-r$. Conjugating with $O$ does not change the eigenvalues of the sample covariance matrix  $\frac{1}{S} X X^\T$; it also does not change the fact that the columns of $X$ are i.i.d. On the other hand, after transformation by $O$, we can use Proposition \ref{Proposition_PCA_equation} for $\U=\Omega X$. We reach the following:

\begin{corollary} \label{Corollary_spiked_covariance_induction}
 For each $1\le k \le r$, the eigenvalues of $\frac{1}{S} X X^\T$ solve an equation in $a$
\begin{equation}
\label{eq_spiked_covariance_induction}
  \frac{1}{S}\frac{\sum_{i=N}^S \xi_i^2}{a}+ \frac{1}{S} \sum_{i=1}^{N-1} \frac{ \xi_i^2}{a-\lambda_i}=\frac{1}{\theta_k},
\end{equation}
where $\lambda_1^2\ge\dots\ge\lambda_{N-1}^2$ are eigenvalues of $\frac{1}{S} Y Y^\T$, with $Y$ being $(N-1)\times S$ matrix with i.i.d.\ Gaussian columns of covariance $\Omega'$, and $\xi_1,\dots,\xi_S$ are i.i.d.\ $\mathcal N(0,1)$ independent from $Y$.
\end{corollary}
\begin{proof} After we rotate by $O$ defined above, the zeroth row of $O X$ is a vector with i.i.d.\ $\sqrt{\theta_k} \mathcal N(0,1)$ random variables as its components, independent from the remaining rows of $X$. Hence, the scalar products of the zeroth row with orthonormal vectors $\v_i$ are again i.i.d.\ $\sqrt{\theta_k} \mathcal N(0,1)$ and we can denote them $\sqrt{\theta_k} \xi_i$. The values of $\langle \v^*,\v_i\rangle$ in \eqref{eq_PCA_equation} differ from these scalar product by the normalization of $\v^*$, i.e., they are
$$
 \frac{\theta_k \xi_i^2}{\sum_{l=1}^S \theta_k \xi_l^2}= \frac{\xi_i^2}{\sum_{l=1}^S  \xi_l^2}.
$$
On the other hand, the value of $\lambda^*$ in \eqref{eq_PCA_equation} is the squared length of the zeroth row, $\|\v^*\|^2=\sum_{l=1}^S \theta_k \xi_l^2$. Finally, Proposition \ref{Proposition_PCA_equation} deals with $\U \U^\T$, while in the corollary the matrix $X X^\T$ is divided by $S$ and, therefore, we should rescale by $S$ both $\lambda_i$ and $a$. Hence, dividing \eqref{eq_PCA_equation} by $a$ and $\theta_k$, we get an equivalent form of  \eqref{eq_spiked_covariance_induction}:
$$
  \frac{1}{a}\left(\frac{\sum_{i=1}^S \xi_i^2}{S}+ \frac{1}{S} \sum_{i=1}^{N-1} \frac{\lambda_i \xi_i^2}{a-\lambda_i}\right)=\frac{1}{\theta_k}. \qedhere
$$
\end{proof}
\begin{corollary} \label{Corollary_spiked_covariance_interlacement}
 In Corollary \ref{Corollary_spiked_covariance_induction}, the eigenvalues $a_1\ge \dots\ge a_N$ of $\frac{1}{S} X X^\T$  interlace with those of $\frac{1}{S} Y Y^\T$: $a_1\ge \lambda_1\ge a_2\ge \dots\ge \lambda_{N-1}\ge a_N$.
\end{corollary}
\begin{proof}
 After multiplying by denominators, \eqref{eq_spiked_covariance_induction} is a degree $N$ polynomial equation and we locate its roots by keeping track of the signs of the difference between left-hand and right-hand sides on intervals $(\lambda_1,+\infty)$, $(\lambda_2,\lambda_1)$, \dots, $(\lambda_{N-1},\lambda_N)$.
\end{proof}

\begin{proposition} \label{Proposition_covariance_far_spikes}
 Consider the spiked covariance model $\Omega=\sigma^2 I_N +  (\theta_i-\sigma^2) \cdot \u^*_i (\u^*_i)^\T$ of Section \ref{Section_spiked_covariance} with $\sigma^2=1$, and $\theta_1>\theta_2>\dots>\theta_r\ge 0$ split into two groups: $\theta_1,\dots,\theta_{q-1}>\theta^c=1+\gamma$ and $\theta_q,\dots,\theta_r<\theta^c=1+\gamma$. Then, with $\frac{N}{S}=\gamma^2+O\left(\tfrac{1}{N}\right)$, in the sense of convergence in joint distribution and using \eqref{eq_Spiked_covariance_params}:
  \begin{align}
 \label{eq_x41}\lim_{N\to\infty} \sqrt{N}(\lambda_i-\lambda(\theta_i))&\stackrel{d}{=} \mathcal N(0, V(\theta_i)),\qquad 1\le i \le q-1,\\
 \label{eq_x42} \lim_{N\to\infty} N^{2/3}\frac{\lambda_{i}-(1+\gamma)^2}{\gamma(1+\gamma)^{4/3}}&\stackrel{d}{=} \aa_{i-q+1}, \qquad i\geq q,
 \end{align}
 where $ \mathcal N(0, V(\theta_i))$ are independent over $i$ and with  points of the Airy$_1$ point process  $\{\aa_j\}_{j\ge 1}$.
 In addition, \eqref{eq_Assumption_local_law} with $h(x)=1$ holds for the eigenvalues of $\frac{1}{S}X X^\T$.
\end{proposition}
\begin{proof} The final claim, \eqref{eq_Assumption_local_law} with $h(x)=1$ follows by induction on $r$ from Theorem \ref{Theorem_from_Paul} for $r=0$, and using Lemma \ref{Lemma_preservation_local_law} with Corollary \ref{Corollary_spiked_covariance_interlacement} for the induction step. Here $\lambda_\pm=(1\pm\gamma)^2$,
\begin{equation}
\label{eq_Marchenko_pastur}
 \mu(x)\dd x=\frac{1}{2\pi} \frac{\sqrt{(\lambda_+-x)(x-\lambda_-)}}{\gamma^2 x}\, \mathbf 1_{[\lambda_-,\lambda_+]}\, \dd x, \qquad m(z)=\frac{z+\gamma^2-1-\sqrt{(z-\lambda_+)(z-\lambda_-)}}{2\gamma^2 z},
\end{equation}
which are the Marchenko-Pastur law and its Stieljes transform, respectively, and the constants of \eqref{eq_Assumption_Stieltjes_imaginary} are computed to be $\s=\frac{\sqrt{\lambda_+-\lambda_-}}{2\gamma^2\lambda_+}=\frac{1}{\gamma^{3/2}(1+\gamma)^2}$, $\m=\frac{(1+\gamma)^2+\gamma^2-1}{2\gamma^2 (1+\gamma)^2}=\frac{1}{\gamma(1+\gamma)}$,

 Statements close to \eqref{eq_x41}, \eqref{eq_x42} are known from \citet{paul2007asymptotics}, \citet{bai2008central}, \citet{bloemendal2016principal}. Alternatively, the proof can be obtained by the same argument as in Proposition \ref{Proposition_Wigner_far_spikes}, by induction on $r$ with the base case $r=0$ given in \cite{johnstone2001distribution,soshnikov2002note}  and the step based on Corollary \ref{Corollary_spiked_covariance_induction}. We only highlight the key computation, which is an analogue of \eqref{eq_x36} and \eqref{eq_x36_2}.

 We rewrite \eqref{eq_spiked_covariance_induction} as
 \begin{equation}
\label{eq_x43}
 \frac{1}{S}\cdot \frac{S-N+1}{a}+ \frac{1}{S} \sum_{i=1}^{N-1} \frac{1}{a-\lambda_i}+\frac{1}{S}\frac{\sum_{i=N}^S (\xi_i^2-1)}{a}+ \frac{1}{S} \sum_{i=1}^{N-1} \frac{ \xi_i^2-1}{a-\lambda_i}=\frac{1}{\theta_k},
\end{equation}
and analyze its asymptotics, assuming that $\lambda_1,\dots,\lambda_N$ are eigenvalues of $\frac{1}{S} X X^\T$, where $X$ is $(N-1)\times S$ matrix of i.i.d.\ $\mathcal N(0,1)$. Then, using the Marchenko-Pastur law (see Theorem \ref{Theorem_from_Paul} or \citet[Chapter 3 and Theorem 9.10]{Bai_Silverstein}), formulas for the limit \eqref{eq_Marchenko_pastur}, and CLT, the first two terms in \eqref{eq_x43} become deterministic as $S\to\infty$, while the third and fourth terms become independent Gaussians. Hence, \eqref{eq_x43} becomes
 \begin{equation}
\label{eq_x44}
 \frac{1-\gamma^2}{a}+ \gamma^2 m(a)+\frac{\sqrt{2}}{\sqrt{N}} \mathcal N(0,1) \sqrt{\frac{\gamma^2(1-\gamma^2)}{a^2}-\gamma^4 m'(a)}+ o\left(\frac{1}{\sqrt{N}}\right)=\frac{1}{\theta_k}.
\end{equation}
Plugging the formula for $m(z)$ from \eqref{eq_Marchenko_pastur}, we get
 \begin{multline}
\label{eq_x45}
\frac{a+1-\gamma^2-\sqrt{(a-\lambda_+)(a-\lambda_-)}}{2 a}\\+
\frac{\gamma}{a \sqrt{N}} \mathcal N(0,1) \sqrt{1-\gamma^2
 -\frac{(1-\gamma^2)^2-a(1+\gamma^2)}{\sqrt{(a-\lambda_+)(a-\lambda_-)}}}+ o\left(\frac{1}{\sqrt{N}}\right)=\frac{1}{\theta_k}.
\end{multline}
Treating the last identity as an equation on $a$, we solve it asymptotically as $N\to\infty$, getting:
\begin{equation}
 a=\theta_k\left(1+\frac{\gamma^2}{\theta_k-1}\right)+ \mathcal N(0,1)\frac{\gamma\sqrt{2}}{ \sqrt{N}} \theta_k   \sqrt{ 1-\frac{\gamma^2}{(\theta_k-1)^2}}+o\left(\frac{1}{\sqrt{N}}\right),
\end{equation}
which matches \eqref{eq_Spiked_covariance_params}.
\end{proof}

\begin{proof}[Proof of Theorem \ref{Theorem_main_convergence_statement} for the spiked covariance model of Section \ref{Section_spiked_covariance}] The constants \eqref{eq_transition_parameters} simplify to $\kappa_1=\gamma(1+\gamma)^{4/3}$, $\kappa_2=\frac{1}{\gamma (1+\gamma)^{2/3}}$. Note that $\kappa_1$ matches the denominator in \eqref{eq_x42}, as it should; it also equals $\s^{-2/3}$, where $\s$ is the constant in \eqref{eq_Assumption_Stieltjes_imaginary}, as computed after \eqref{eq_Marchenko_pastur}. Hence, Assumption \ref{Assumption_for_limit} will be satisfied.

We analyze the equation \eqref{eq_spiked_covariance_induction} for $k=q$. In this situation the asymptotics of $(\lambda_i)_{i=1}^N$ is given to us by Proposition \ref{Proposition_covariance_far_spikes}. Arguing as in the previous section, the $q-1$ largest roots of the equation are close to $\lambda_1,\dots,\lambda_{q-1}$, resulting in \eqref{eq_Gaussian_limit}. In order to establish \eqref{eq_Transition_limit}, we need to approximate \eqref{eq_spiked_covariance_induction} for $a$ close to $\lambda_+=(1+\gamma)^2$ and locate the root of the equation in the $(\lambda_q,\lambda_{q-1})$ interval. We change the variables
\begin{equation}
\label{eq_x46}
 b=N^{2/3}\frac{a-\lambda_+}{\kappa_1}=N^{2/3}\frac{a-(1+\gamma)^2}{\gamma(1+\gamma)^{4/3}},\qquad a=(1+\gamma)^2+N^{-2/3}\gamma(1+\gamma)^{4/3} b,
\end{equation}
and apply Theorem \ref{Theorem_as_convergence} with $h(x)=1$, converting \eqref{eq_spiked_covariance_induction} into (recall that $\theta^c=1+\gamma$ and $\m=\frac{1}{\gamma(1+\gamma)}$, as computed after \eqref{eq_Marchenko_pastur}):
$$
  \frac{S+1-N}{S}\cdot \frac{1}{(1+\gamma)^2}+ \frac{N}{S}\left[\frac{1}{\gamma(1+\gamma)}+N^{-1/3} \frac{1}{\gamma(1+\gamma)^{4/3}} \mathcal G(b)\right] =\frac{1}{1+\gamma+N^{-1/3}\tilde\theta} +o\left(N^{-1/3}\right).
$$
Recalling that $\frac{N}{S}=\gamma^2+O\left(\tfrac{1}{N}\right)$, we convert the last equation into
$$
  \frac{1-\gamma}{1+\gamma}+ \frac{\gamma}{(1+\gamma)}+N^{-1/3} \frac{\gamma}{(1+\gamma)^{4/3}} \mathcal G(b) =\frac{1}{1+\gamma} \left(1-N^{-1/3}\frac{\tilde\theta}{1+\gamma}\right) +o\left(N^{-1/3}\right).
$$
The finite order term cancel out and we finally get after multiplying by $N^{1/3}$ the equation
$$
 \frac{\gamma}{(1+\gamma)^{4/3}} \mathcal G(b) = - \tilde\theta \frac{1}{(1+\gamma)^2}+o(1)\quad  \Longleftrightarrow \quad \mathcal G(b) = - \tilde\theta \frac{1}{\gamma(1+\gamma)^{2/3}}+o(1).
$$
Recognizing the constant $\kappa_2$  and comparing with Definition \ref{Definition_Transition_function}, we arrive at \eqref{eq_Transition_limit}.
\end{proof}

We now analyse $\theta_i\to\infty$ limits. The following is immediate from \eqref{eq_Spiked_covariance_params}.

\begin{corollary} \label{corollary_spiked_covariance_strong}
 The parameters in \eqref{eq_Gaussian_limit} for the spiked covariance model satisfy:
 \begin{equation}
 \label{eq_x86}
   \frac{\lambda(\theta)}{\theta}=1+O\left(\frac{1}{\theta}\right), \qquad  \frac{V(\theta)}{\theta^2}=2\gamma^2+O\left(\frac{1}{\theta}\right), \qquad \theta\to\infty.
 \end{equation}
\end{corollary}

Corollary \ref{corollary_spiked_covariance_strong} agrees with results in the literature (e.g.\ \citet{wang2017asymptotics,cai2020limiting,jiang2021generalized,fan2024tests}) for the cases when $\theta$ grows with $N$, and shows that the approximation of Theorem \ref{Theorem_main_convergence_statement} remains valid in such regimes. For instance, in the strong signal regime, $\theta_i= N\tilde \theta_i$ as $N\to\infty$, formally inserting $\theta_i= N\tilde  \theta_i$ into \eqref{eq_Spiked_covariance_params}, dividing by $\gamma N \tilde \theta_i$ and using \eqref{eq_x86}, we get
$
\sqrt{S}\left(\frac{\lambda_i}{\theta_i}-1\right)\xrightarrow[]{d}  \mathcal N(0,2),
$
which precisely matches \citet[Theorem 1]{fan2024tests}.

\subsection{Factor models} \label{Section_factor_proof} An analogue of Propositions \ref{Proposition_equation_spiked_Wigner} and \ref{Proposition_PCA_equation} is more complicated here.

\begin{proposition} \label{Proposition_Factor_master}
 For $N\le S$, let $X$ be an $N\times S$ matrix of the form
 \begin{equation} \label{eq_x49}
  X=\sqrt{\theta  S} \cdot \u^* (\v^*)^\T + Y,\qquad \text{ with }\qquad Y=\sum_{i=1}^N \sqrt{\lambda_i S} \cdot \u_i (\v_i)^\T,
 \end{equation}
  where $\u^*$ and $\u_i$ are $N$--dimensional unit vectors; $\u_i$, $1\le i \le N$, are pairwise orthogonal; $\v^*$ and $\v_i$ are $S$--dimensional unit vectors;  $\v_i$, $1\le i \le N$, are pairwise orthogonal. If $a$ is a non-zero squared singular value of $X/\sqrt{S}$, i.e., an eigenvalue of $\frac{1}{S} X X^\T$, then either it solves
  \begin{equation}
  \label{eq_Factor_equation}
  \left(1-\sqrt{\theta}\sum\limits_{i=1}^N\frac{\sqrt{\lambda_i}  \langle \u^*,\u_{i}\rangle \langle \v^*, \v_i \rangle }{a-\lambda_i}\right)^2 =a \theta \left(\sum\limits_{i=1}^N\frac{ \langle  \u^*, \u_i\rangle^2
 }{a-\lambda_i}\right)\left(\sum\limits_{j=1}^S \frac{  \langle \v^*, \v_j \rangle^2}{a -  \lambda_j}\right),
 \end{equation}
 where $\v_{N+1}$, \dots $\v_{S}$ are arbitrary vectors complimenting $\v_1,\dots,\v_N$ to an orthonormal basis, $\langle \cdot,\cdot\rangle$ is the scalar product, and we set $\lambda_j=0$ for $N<j\le S$.
  Or $a=\lambda_i$ for $1\le i \le N$, where $\lambda_i$ has multiplicity one, $\langle \u^*,\u_i\rangle=\langle \v^*,\v_i\rangle=0$, and the equation \eqref{eq_Factor_equation} holds with the $i$-th terms excluded; or $a=\lambda_i$, where $\lambda_i$ has multiplicity larger than one.
\end{proposition}
\begin{remark}
 A related statement is \citet[Lemma 4.1]{benaych2012singular}. If for some $i$, $\lambda_i=0$, $\u^*=\u_i$,  $\langle \u^*, \u_{i'}\rangle=0$  for all $i\ne i'$, then \eqref{eq_Factor_equation} turns into \eqref{eq_PCA_equation}.
\end{remark}
\begin{proof}[Proof of Proposition \ref{Proposition_Factor_master}]
 We only deal with distinct $\lambda_i$ and assuming for each index $i$ either $\langle \u^*,\u_i\rangle\ne 0$ or $\langle \v^*,\v_i\rangle\ne 0$. Other cases are obtained by continuous deformations.

 Let $a$ be a squared singular value of $X/\sqrt{S}$ corresponding to vectors $\hat\u$ and $\hat \v$. Then $a=pq/S$, where $p$ and $q$ solve:
$$
 \begin{cases} X\hat \v=p\hat \u,\\ X^\T \hat \u = q \hat \v.
 \end{cases}
 \quad \Longleftrightarrow \quad
   \begin{cases} \langle X\hat \v, \u_i\rangle =p\langle \hat \u,\u_i\rangle,& 1\le i \le N,\\ \langle X^\T \hat \u, \v_j \rangle = q \langle \hat \v, \hat \v_j\rangle, & 1\le j \le S.
 \end{cases}
$$
Denoting $\alpha_i=\langle \hat \u,\u_i\rangle$ and $\beta_j= \langle \hat \v, \hat \v_j\rangle$, we rewrite these $N+S$ equations as:
$$
 \begin{cases} \langle X \sum_{j'=1}^S \beta_{j'} \v_{j'}, \u_i\rangle=
 \sqrt{\theta  S}  \sum_{j'=1}^S \beta_{j'} \langle  \v^*,\v_{j'} \rangle \langle  \u^*, \u_i\rangle +  \sqrt{\lambda_i S} \beta_{i}=p\alpha_i,& 1\le i \le N,
 \\
  \langle X^\T \sum_{i'=1}^N \alpha_{i'} \u_{i'}, \v_j \rangle
  = \sqrt{\theta S} \sum_{i'=1}^N \alpha_{i'} \langle \u^*,\u_{i'}\rangle \langle \v^*, \v_j \rangle
  +  \sqrt{\lambda_j S} \alpha_{j} = q \beta_j, & 1\le j \le N,
  \\
   \langle X^\T \sum_{i'=1}^N \alpha_{i'} \u_{i'}, \v_j \rangle= \sqrt{\theta S} \sum_{i'=1}^N \alpha_{i'} \langle \u^*,\u_{i'}\rangle \langle \v^*, \v_j \rangle = q \beta_j, & N+1\le j \le S.
 \end{cases}
$$
Combining the equations corresponding to the same $i=j$ and solving as two linear equations in two variables $\alpha_i$, $\beta_i$, we get
\begin{equation}
\label{eq_x48}
 \begin{cases} \alpha_i= \sqrt{\theta  S}\frac{q  \langle  \u^*, \u_i\rangle  \tilde \beta
 +\sqrt{\lambda_i S}  \langle \v^*, \v_i \rangle \tilde\alpha }{pq-\lambda_i S}, & 1\le i \le N,\\
   \beta_{i}= \sqrt{\theta  S}  \frac{\sqrt{\lambda_i S}\langle  \u^*, \u_i\rangle\tilde\beta + p  \langle \v^*, \v_i \rangle\tilde \alpha}{p q -  \lambda_i S}, & 1\le i \le S,
 \end{cases}
\end{equation}
where in the last formula for $N<i\le S$ we should use $\lambda_i=0$ and
$$
  \tilde \alpha=  \sum_{i'=1}^N \alpha_{i'} \langle \u^*,\u_{i'}\rangle, \qquad \tilde \beta=\sum_{j'=1}^S \beta_{j'} \langle  \v^*,\v_{j'} \rangle.
$$
We find the values of $\tilde \alpha$ and $\tilde \beta$, by plugging \eqref{eq_x48} back into their definitions, getting:
$$
\begin{cases}
    \tilde \alpha= \tilde\alpha \sqrt{\theta  S} \sum\limits_{i=1}^N\frac{\sqrt{\lambda_i S}  \langle \u^*,\u_{i}\rangle \langle \v^*, \v_i \rangle }{pq-\lambda_i S}+ \tilde \beta \sqrt{\theta  S}\sum\limits_{i=1}^N\frac{q  \langle  \u^*, \u_i\rangle^2
 }{pq-\lambda_i S} , \\
  \tilde \beta= \tilde \alpha \sqrt{\theta  S}\sum\limits_{j=1}^S \frac{p  \langle \v^*, \v_j \rangle^2}{p q -  \lambda_j S}+\tilde \beta  \sqrt{\theta  S} \sum\limits_{j=1}^S \frac{\sqrt{\lambda_j S} \langle  \u^*, \u_j\rangle\langle  \v^*,\v_{j} \rangle  }{p q -  \lambda_j S}.
\end{cases}
$$
The system of two homogeneous linear equations has a non-zero solution if and only if the determinant of the matrix coefficients is zero, which is precisely the condition \eqref{eq_Factor_equation}.
\end{proof}

\begin{corollary} \label{Corollary_factor_interlacement}
 In Proposition \ref{Proposition_Factor_master}, let $a_1\ge \dots\ge a_N$  be the eigenvalues of $\frac{1}{S} X X^\T$ and let  $\lambda_1\ge \dots\ge \lambda_N$ be the eigenvalues of $\frac{1}{S} Y Y^\T$. Then there exists another set of $N$ eigenvalues, $\mu_1\ge \dots\ge\mu_N$, such that
 \begin{equation}
  a_1\ge \mu_1\ge a_2\ge\dots\ge a_N\ge \mu_N, \qquad \text{and} \qquad \lambda_1\ge \mu_1\ge \lambda_2\ge\dots\ge \lambda_N\ge \mu_N.
 \end{equation}
\end{corollary}
\begin{proof} Using \eqref{eq_x49}, we write
 $$
 XX^\T=
 \theta S \cdot \u^*  (\u^*)^\T + \sqrt{\theta S} \cdot \left[ \u^* (\v^*)^\T Y^*+ Y  \v^* (\u^*)^\T\right]+ Y Y^*,
 $$
 which implies that $\frac{1}{S} X X^\T$ is a sum of $\frac{1}{S} Y Y^\T$ and a rank two symmetric matrix. In the (non-orthogonal) basis $(\u^*, Y\v^*)$, this matrix has the form
 $$
 \begin{pmatrix}\theta S +c\sqrt{\theta S}  & c\theta S+d\sqrt{\theta S}\\ \sqrt{\theta S} & c\sqrt{\theta S} \end{pmatrix}, \qquad c=\langle \u^*, Y \v^*\rangle, \quad d=\langle Y\v^*,Y\v^*\rangle.
 $$
This matrix has one positive and one negative eigenvalue. because its determinant is $(c^2-d)\theta S<0$, since $\u^*$ is a unit vector. It remains to use Corollary \ref{Corollary_Wigner_interlacement} twice.
\end{proof}

\begin{lemma} \label{Lemma_introduce_independence_factors}
In the factor model
$X = \sum_{i=1}^r \sqrt{\theta_i} \sqrt{S} \cdot \u^*_i (\v^*_i)^\T + \mathcal E$
 of \eqref{eq_factor_model}, replacing $\u_i^*$ and $\v_i^*$ with $2r$ independent vectors uniformly distributed on the $N$-- and $S$--dimensional unit spheres, respectively (independent of each other and $\mathcal E$) leaves Theorem \ref{Theorem_main_convergence_statement} unchanged.
\end{lemma}
\begin{proof}
 The asymptotics in Theorem \ref{Theorem_main_convergence_statement} does not depend on the choice of $\{\u_i^*\}_{i=1}^r$ and $\{\v_i^*\}_{i=1}^r$ in \eqref{eq_factor_model} as long as they form two orthonormal systems.
  Indeed, any $r$ orthonormal vectors can be obtained from any other $r$ orthonormal vectors by an orthogonal transformation. Such transformation does not change the eigenvalues of $\frac{1}{S}X X^\T$, and also does not change the probability distribution of the matrix $ \mathcal E$ (which uses Gaussianity of its matrix elements).

 Now let $\tilde \u_1^*,\dots,\tilde \u_r^*$ be $r$ independent vectors uniformly distributed on the $N$--dimensional unit sphere and let $\tilde \v_1^*,\dots,\tilde \v_r^*$ be $r$ independent vectors uniformly distributed on the $S$--dimensional unit sphere. We consider the matrix $M_r=\sum_{i=1}^r \sqrt{\theta_i}\sqrt{S} \tilde \u_i^* (\tilde \v_i^*)^\T$ and would like to decompose it as $M_r=\sum_{i=1}^r \sqrt{\theta_i'}{\sqrt S} \u_i^* (\v_i^*)^\T$ with orthonormal vectors $\u_i^*$ and $\v_i^*$; $\tilde \u_i^*$ and $\tilde \v_i^*$ were not orthonormal. Clearly, $\sqrt{\theta_i'}$ are non-zero singular values of $M_r$ and $\u_i^*$, $\v_i^*$ are corresponding left and right  singular vectors. We claim that
 \begin{equation}
 \label{eq_x47}
  \theta_i=\theta'_i+O\left(\frac{1}{N}+\frac{1}{S}\right),\qquad N,S\to\infty.
 \end{equation}
 Note that \eqref{eq_x47} implies the statement of Lemma \ref{Lemma_introduce_independence_factors}, because addition of $O(1/N)$ does not change any of the asymptotics statements of Theorem \ref{Theorem_main_convergence_statement}. Hence, it remains to prove \eqref{eq_x47}. This can be done by induction on $r$. Using \eqref{eq_Factor_equation} with $Y=M_{r-1}$, the non-zero eigenvalues of $M_r$ solve an equation
   \begin{multline}
  \label{eq_Factor_equation_2}
  \left(1-\sqrt{\theta_r}\sum\limits_{i=1}^{r-1}\frac{\sqrt{\theta'_i}  \langle \tilde \u^*_r,\u_{i}\rangle \langle \tilde \v^*_r, \v_i \rangle }{a-\theta'_i}\right)^2\\ =a \theta_r \left(\sum\limits_{i=1}^{r-1}\frac{ \langle  \tilde \u^*_r, \u_i\rangle^2
 }{a-\theta'_i}+\sum\limits_{i=r}^N\frac{ \langle  \tilde \u^*_r, \u_i\rangle^2
 }{a}\right)\left(\sum\limits_{j=1}^{r-1} \frac{  \langle \tilde \v^*_r, \v_j \rangle^2}{a -  \theta'_j}+\sum\limits_{j=r}^S \frac{  \langle \tilde \v^*_r, \v_j \rangle^2}{a }\right),
 \end{multline}
 where $(\u_i,\v_i,\theta'_i)$ are singular vectors and values of $M_{r-1}$ for $1\le i \le {r-1}$, and $\u_i$, $\v_i$ with $i\ge r$ complement those to othonormal bases of $N$-- and $S$--dimensional spaces, respectively. We note that after multiplying by the denominators $a$, $a-\theta'_1$, \dots, $a-\theta'_{r-1}$, \eqref{eq_Factor_equation_2} becomes a polynomial equation of degree $r$ (there is a cancelation between the left and right hand sides which guarantees that denominators $(a-\theta'_i)^2$ do not appear) and therefore it has $r$ roots. Let us locate these roots.

 Representing the unit vectors $\tilde \u^*_r$ and $\tilde \v^*_r$ as vectors with i.i.d.\ $\mathcal N(0,1)$ components divided by their lengths, using the Law of Large Numbers and the induction hypothesis, the equation is rewritten as $N,S\to\infty$ as
 \begin{multline}
  \label{eq_Factor_equation_3}
  \left(1-\sqrt{\frac{\theta_r}{NS}}\sum\limits_{i=1}^{r-1}\frac{\sqrt{\theta_i+O\left(\frac{1}{N}+\frac{1}{S}\right)}  \chi_i \xi_i \left(1+O\left(\frac{1}{\sqrt N}+\frac{1}{\sqrt{S}}\right)\right)  }{a-\theta_i+O\left(\frac{1}{N}+\frac{1}{S}\right)}\right)^2 =a \theta_r  \\ \times \left(\frac{1}{N}\sum\limits_{i=1}^{r-1}\frac{\chi_i^2
 \left(1+O\left(\frac{1}{\sqrt N}\right)\right)}{a-\theta_i+O\left(\frac{1}{N}+\frac{1}{S}\right)}+\frac{1+O\left(\frac{1}{N}\right)}{a}\right)\left(\frac{1}{S}\sum\limits_{j=1}^{r-1} \frac{ \xi_i^2\left(1+O\left(\frac{1}{\sqrt{S}}\right)\right)}{a -  \theta_j+O\left(\frac{1}{N}+\frac{1}{S}\right)}+ \frac{1+O\left(\frac{1}{S}\right)}{a }\right),
 \end{multline}
 where $\chi_i$ and $\xi_i$ are i.i.d.\ $\mathcal N(0,1)$ random variables. \eqref{eq_Factor_equation_3} clearly has $r$ roots of the form $\theta_i+O(\tfrac{1}{N}+\tfrac{1}{S})$, $1\le i \le r$, thus proving \eqref{eq_x47}.
\end{proof}

\begin{proposition} \label{Proposition_factor_far_spikes}
 Consider the factor model $X = \sum_{i=1}^r \sqrt{\theta_i S} \cdot \u^*_i (\v^*_i)^\T + \mathcal E$ of Section \ref{Section_spiked_covariance} with $\sigma^2=1$, and $\theta_1>\theta_2>\dots>\theta_r\ge 0$ split into two groups: $\theta_1,\dots,\theta_{q-1}>\theta^c=\gamma$ and $\theta_q,\dots,\theta_r<\theta^c=\gamma$. Then, with $\frac{N}{S}=\gamma^2+O\left(\tfrac{1}{N}\right)$,  in the sense of convergence in joint distribution and using \eqref{eq_Factor_params}:
  \begin{align}
 \label{eq_x50}\lim_{N\to\infty} \sqrt{N}(\lambda_i-\lambda(\theta_i))&\stackrel{d}{=} \mathcal N(0, V(\theta_i)),\qquad 1\le i \le q-1,\\
 \label{eq_x51} \lim_{N\to\infty} N^{2/3}\frac{\lambda_{i}-(1+\gamma)^2}{\gamma(1+\gamma)^{4/3}}&\stackrel{d}{=} \aa_{i-q+1}, \qquad i\geq q,
 \end{align}
 where $ \mathcal N(0, V(\theta_i))$ are independent over $i$ and with points of the Airy$_1$ point process  $\{\aa_j\}_{j\ge 1}$. In addition, \eqref{eq_Assumption_local_law} holds for the eigenvalues of $\frac{1}{S}X X^\T$ with $h(x)=1$ and with $h(x)=\sqrt{x}$.
\end{proposition}
\begin{proof}   The final statement, \eqref{eq_Assumption_local_law}  with $h(x)=1$ and $h(x)=\sqrt{x}$ follows by induction on $r$ from Theorem \ref{Theorem_from_Paul} and Corollary \ref{Corollary_local_with_square_root} for $r=0$, and using Lemma \ref{Lemma_preservation_local_law} with Corollary \ref{Corollary_factor_interlacement} for the induction step. Here $\lambda_\pm=(1\pm\gamma)^2$, and for $h(x)=1$
\begin{equation}
\label{eq_Marchenko_pastur_2}
 \mu(x)\dd x=\frac{1}{2\pi} \frac{\sqrt{(\lambda_+-x)(x-\lambda_-)}}{\gamma^2 x}\, \mathbf 1_{[\lambda_-,\lambda_+]}\, \dd x, \qquad m(z)=\frac{z+\gamma^2-1-\sqrt{(z-\lambda_+)(z-\lambda_-)}}{2\gamma^2 z},
\end{equation}
which are the same as in \eqref{eq_Marchenko_pastur}, and the constants of \eqref{eq_Assumption_Stieltjes_imaginary} are computed to be $\s=\frac{\sqrt{\lambda_+-\lambda_-}}{2\gamma^2\lambda_+}=\frac{1}{\gamma^{3/2}(1+\gamma)^2}$, $\m=\frac{(1+\gamma)^2+\gamma^2-1}{2\gamma^2 (1+\gamma)^2}=\frac{1}{\gamma(1+\gamma)}$.

 Statements similar to \eqref{eq_x50} are known from \citet{onatski2012asymptotics} and \citet{benaych2012singular}; \eqref{eq_x51} is harder to locate in the literature. The proof can be obtained by the same argument as in Proposition \ref{Proposition_Wigner_far_spikes}, by induction on $r$ with the base case $r=0$ given in \cite{johnstone2001distribution,soshnikov2002note}  and the step based on Proposition \ref{Proposition_Factor_master}. We only highlight the key computation, which is an analogue of \eqref{eq_x36} and \eqref{eq_x36_2}.

 In view of Lemma \ref{Lemma_introduce_independence_factors}, and representing uniformly random unit vectors as Gaussian vectors with i.i.d.\ components divided by their norms, we rewrite \eqref{eq_Factor_equation} as
 \begin{equation}
  \label{eq_Factor_equation_4}
  \left(1-\sqrt{\frac{\theta}{\sum_{i=1}^N \xi_i^2 \sum_{j=1}^S \eta_j^2}}\sum\limits_{i=1}^N\frac{\sqrt{\lambda_i}  \xi_i \eta_i}{a-\lambda_i}\right)^2 =\frac{a \theta}{\sum_{i=1}^N \xi_i^2 \sum_{j=1}^S \eta_j^2} \left(\sum\limits_{i=1}^N\frac{ \xi_i^2
 }{a-\lambda_i}\right)\left(\sum\limits_{j=1}^S \frac{  \eta_j^2}{a -  \lambda_j}\right),
 \end{equation}
 where $\xi_i$ and $\eta_j$ are i.i.d.\ $\mathcal N(0,1)$. We analyze the asymptotics of \eqref{eq_Factor_equation_4} assuming that $\lambda_1,\dots,\lambda_N$ are eigenvalues of $\frac{1}{S} \mathcal E \mathcal E^\T$, where $\mathcal E$ is $N\times S$ matrix of i.i.d.\ $\mathcal N(0,1)$. We multiply  \eqref{eq_Factor_equation_4} by $\frac{\sum_{i=1}^N \xi_i^2 \sum_{j=1}^S \eta_j^2}{N S }$ and then rewrite separating the terms of different orders as $N\to\infty$:
 \begin{multline}
  \label{eq_Factor_equation_5}
  \frac{\sum_{i=1}^N 1 \sum_{j=1}^S 1}{N S}+\sum_{i=1}^N \frac{\xi_i^2-1}{N} + \sum_{j=1}^S \frac{\eta_j^2-1}{S} -2 \sqrt{\frac{\theta}{NS}}\sum\limits_{i=1}^N\frac{\sqrt{\lambda_i}  \xi_i \eta_i}{a-\lambda_i} + o\left(\frac{1}{\sqrt{N}}\right)
   \\=\frac{a \theta}{NS} \left[\sum\limits_{i=1}^N\frac{1
 }{a-\lambda_i}  \sum\limits_{j=1}^S \frac{1}{a -  \lambda_j} + \sum\limits_{i=1}^N\frac{1
 }{a-\lambda_i}  \sum\limits_{j=1}^S \frac{  \eta_j^2-1}{a -  \lambda_j} + \sum\limits_{j=1}^S \frac{1}{a -  \lambda_j} \sum\limits_{i=1}^N\frac{ \xi_i^2-1
 }{a-\lambda_i}\right].
 \end{multline}
 Using the Marchenko-Pastur law (Theorem \ref{Theorem_from_Paul} or \citet[Chapter 3 and Theorem 9.10]{Bai_Silverstein}), formulas \eqref{eq_Marchenko_pastur}, the notation $
  \tilde m(a)=  \left[\gamma^2m(a)+ \frac{1-\gamma^2}{a}\right]
 $, and $\frac{N}{S}=\gamma^2+O\left(\frac{1}{N}\right)$ the equation \eqref{eq_Factor_equation_5} becomes:
 \begin{multline}
  \label{eq_Factor_equation_6}
  1-a\theta\, m(a) \tilde m(a)+\sum_{i=1}^N \frac{\xi_i^2-1}{N} \left[1-\frac{a \theta}{a-\lambda_i} \frac{1}{S}\sum\limits_{j=1}^S \frac{1}{a -  \lambda_j}\right]\\
  + \sum_{j=1}^S \frac{\eta_j^2-1}{S}\left[1-\frac{a \theta}{a-\lambda_j} \frac{1}{N} \sum\limits_{i=1}^N\frac{1
 }{a-\lambda_i} \right] -2 \sqrt{\frac{\theta}{NS}}\sum\limits_{i=1}^N\frac{\sqrt{\lambda_i}  \xi_i \eta_i}{a-\lambda_i} = o\left(\frac{1}{\sqrt{N}}\right).
 \end{multline}
 Since $\xi^2_i-1$, $\eta_j^2-1$, $\xi_i \eta_i$ are three mean $0$ i.i.d.\ in $i$ sequences, we are in a position to apply the Central Limit Theorem; note that all pairwise covariances between these random variables vanish (using $\E (\xi_i^2-1) \xi_i \eta_i=0$), and therefore different sums give rise to independent Gaussian limits.
Using $\E(\xi_i^2-1)^2=\E (\xi_j^2-1)^2=2$, $\E \xi_i^2\eta_i^2=1$ and applying CLT, we further transform \eqref{eq_Factor_equation_6} into
 \begin{multline}
  \label{eq_Factor_equation_7}
  1-a\theta\, m(a) \tilde m(a)+\frac{\mathcal{N}(0,1)}{\sqrt{N}}\Biggl[\frac{2}{N}\sum_{i=1}^N \left(1-\frac{a \theta}{a-\lambda_i} \frac{1}{S}\sum\limits_{j=1}^S \frac{1}{a -  \lambda_j}\right)^2+\\
  + \frac{2\gamma^2}{S}\sum_{j=1}^S \left(1-\frac{a \theta}{a-\lambda_j} \frac{1}{N} \sum\limits_{i=1}^N\frac{1
 }{a-\lambda_i} \right)^2 +4 \frac{\theta}{S}\sum\limits_{i=1}^N\frac{\lambda_i} {(a-\lambda_i)^2}\Biggr]^{1/2} = o\left(\frac{1}{\sqrt{N}}\right).
 \end{multline}
 Applying the Marchenko-Pastur law, we finally get
 \begin{multline}
  \label{eq_Factor_equation_8}
  1-a\theta\, m(a) \tilde m(a)+\frac{\mathcal{N}(0,1)}{\sqrt{N}}\Bigl[[2-4a\theta \, m(a)\tilde m(a)-2 a^2 \theta^2
 \, m'(a)\tilde m^2(a)]\\
  +[2\gamma^2-4\gamma^2 a\theta\, m(a)\tilde m(a)-2\gamma^2 a^2\theta^2\,\tilde m'(a) m^2(a)]-[4 \gamma^2\theta m(a) +4 \gamma^2 a\,\theta m'(a)]\Bigr]^{1/2} = o\left(\frac{1}{\sqrt{N}}\right).
 \end{multline}
In the leading order, the solution $a$ is found from the equation $1-a\theta\, m(a) \tilde m(a)=0$.
Plugging the formula for $m(z)$ from \eqref{eq_Marchenko_pastur}, this becomes
\begin{multline} \label{eq_x52}
 1= a\theta\, \frac{a-\sqrt{(a-\lambda_+)(a-\lambda_-)}+\gamma^2-1}{2\gamma^2 a}\cdot \frac{a-\sqrt{(a-\lambda_+)(a-\lambda_-)}+1-\gamma^2}{2 a}
 \\= \theta\, \frac{a-1-\gamma^2-\sqrt{(a-\lambda_+)(a-\lambda_-)}}{2\gamma^2 },
\end{multline}
which is equivalent to
\begin{equation}
\label{eq_x80}
 a=\theta + 1+\gamma^2+\frac{\gamma^2}{\theta}=\frac{1}{\theta}(\gamma^2+\theta)(1+\theta),
\end{equation}
which matches the formula for $\lambda(\theta)$ in \eqref{eq_Factor_params}.

Further, continuing to treat \eqref{eq_Factor_equation_8} as an equation on $a$, we develop the second order expansion of the solution as $N\to\infty$. For that we note the following simplifications whenever $a$ is given by \eqref{eq_x80}:

\begin{equation} \label{eq_x81}
  \sqrt{(a-\lambda_+)(a-\lambda_-)}= \theta -\frac{\gamma^2}{\theta}, \qquad m\left(a\right)= \frac{1}{\theta+\gamma^2}, \qquad \tilde m(a)= \frac{1}{1 + \theta}.
\end{equation}
\begin{equation} \label{eq_x82}
 m'(a)
 =\frac{ -\theta^2}{(\gamma^2+\theta)^2\left(\theta^2 -\gamma^2\right)}, \qquad
 \tilde m'(a)
 = \frac{-\theta^2}{ \left(\theta^2 -\gamma^2\right) (1+\theta)^2}.
\end{equation}
Hence, if we assume that $a$ is close to \eqref{eq_x80} and define small $\Delta a$ through
\begin{equation}
 a=\theta + 1+\gamma^2+\frac{\gamma^2}{\theta}+\Delta a,
\end{equation}
then expanding in $\Delta a$ using \eqref{eq_x81}, \eqref{eq_x82}, we get
$$
 1-a\theta\, m(a) \tilde m(a)
= \Delta a \frac{\theta }{ \theta^2 -\gamma^2}+O(\Delta a^2).
$$
In addition, the expression after $\mathcal N(0,1)$ under $[\cdot]^{1/2}$ in \eqref{eq_Factor_equation_8} simplifies upon plugging $a$ from \eqref{eq_x80}: using \eqref{eq_x81}, \eqref{eq_x82}, we get
$$
  2\gamma^2 \, \frac{1+\gamma^2+2\theta}{\theta^2 -\gamma^2}.
$$
Hence, in terms of $\Delta a$ the equation \eqref{eq_Factor_equation_8} becomes
$$
 \Delta a \frac{\theta }{ \theta^2 -\gamma^2}+\frac{\mathcal{N}(0,1)}{\sqrt{N}}\left[2\gamma^2 \, \frac{1+\gamma^2+2\theta}{\theta^2 -\gamma^2}\right]^{1/2}  +O(\Delta a^2)+ o\left(\frac{1}{\sqrt{N}}\right)=0.
$$
Its solution gives the desired  \eqref{eq_x50}, matching the formula for $V(\theta)$ in \eqref{eq_Factor_params}.
\end{proof}

\begin{proof}[Proof of Theorem \ref{Theorem_main_convergence_statement} for the factor model] The constants \eqref{eq_transition_parameters} are $\kappa_1=\gamma(1+\gamma)^{4/3}$, $\kappa_2=\frac{1}{\gamma (1+\gamma)^{2/3}}$ -- the same as for the spiked covariance model. Assumption \ref{Assumption_for_limit} is satisfied.

We analyze the equation \eqref{eq_Factor_equation} for $\theta=\theta_q$ in the form of \eqref{eq_Factor_equation_4}, i.e., we study \begin{equation}
  \label{eq_Factor_equation_9}
  \left(1-\sqrt{\frac{\theta_q}{\sum_{i=1}^N \xi_i^2 \sum_{j=1}^S \eta_j^2}}\sum\limits_{i=1}^N\frac{\sqrt{\lambda_i}  \xi_i \eta_i}{a-\lambda_i}\right)^2 =\frac{a \theta_q}{\sum_{i=1}^N \xi_i^2 \sum_{j=1}^S \eta_j^2} \left(\sum\limits_{i=1}^N\frac{ \xi_i^2
 }{a-\lambda_i}\right)\left(\sum\limits_{j=1}^S \frac{  \eta_j^2}{a -  \lambda_j}\right),
 \end{equation}
where the asymptotics of $(\lambda_i)_{i=1}^N$ is given  by Proposition \ref{Proposition_factor_far_spikes} and $\xi_i$ and $\eta_j$ are i.i.d.\ $\mathcal N(0,1)$.  As in the previous sections, the $q-1$ largest roots of the equation are close to $\lambda_1,\dots,\lambda_{q-1}$, resulting in \eqref{eq_Gaussian_limit}. To establish \eqref{eq_Transition_limit}, we  approximate  \eqref{eq_Factor_equation_9} for $a$ close to $\lambda_+=(1+\gamma)^2$ and locate the root of the equation in $(\lambda_q,\lambda_{q-1})$ interval. We change the variables
\begin{equation}
\label{eq_x46}
 b=N^{2/3}\frac{a-\lambda_+}{\kappa_1}=N^{2/3}\frac{a-(1+\gamma)^2}{\gamma(1+\gamma)^{4/3}},\qquad a=(1+\gamma)^2+N^{-2/3}\gamma(1+\gamma)^{4/3} b,
\end{equation}
recall that $\theta_c=\gamma$, $\theta_q=\gamma+N^{-1/3}\tilde\theta $,
and aim to apply Theorem \ref{Theorem_as_convergence} to \eqref{eq_Factor_equation_9}.  We can approximate $\sum_{i=1}^N \xi_i^2\approx N$ and $ \sum_{j=1}^S \eta_j^2\approx S$, because the relative errors of order $N^{-1/2}$ in these approximations are smaller than $N^{-1/3}$, the  scale of our interest. \eqref{eq_Factor_equation_9} becomes
\begin{multline}
  \label{eq_Factor_equation_10}
  \left(1-\sqrt{\frac{\gamma+ N^{-1/3}\tilde\theta }{NS}}\sum\limits_{i=1}^N\frac{\sqrt{\lambda_i}  \xi_i \eta_i}{a-\lambda_i}\right)^2 \\ =\frac{\bigl((1+\gamma)^2+N^{-2/3}\gamma(1+\gamma)^{4/3} b\bigr)\, (\gamma+N^{-1/3}\tilde\theta )}{NS} \left(\sum\limits_{i=1}^N\frac{ \xi_i^2
 }{a-\lambda_i}\right)\left(\sum\limits_{j=1}^S \frac{  \eta_j^2}{a -  \lambda_j}\right).
\end{multline}
In the right-hand side we apply Theorem \ref{Theorem_as_convergence} with $h(x)=1$. For the left-hand side we write
$$
\xi_i\eta_i=\frac{1}{2} \left[\left(\frac{\xi_i+\eta_i}{\sqrt{2}}\right)^2-\left(\frac{\xi_i-\eta_i}{\sqrt 2}\right)^2\right]
$$
and apply Theorem \ref{Theorem_as_convergence} twice with $h(x)=\sqrt{x}$. In view of Remark \ref{Remark_joint_convergence}, the convergence is joint over all four applications of Theorem \ref{Theorem_as_convergence}, and we get a limit expressed in terms of the four (correlated through the choices of $\xi_j$ sequences) copies of $\mathcal G(w)$. As noted after \eqref{eq_Marchenko_pastur_2}, we use $\s=\frac{1}{\gamma^{3/2}(1+\gamma)^2}$, $\m=\frac{1}{\gamma(1+\gamma)}$; we also recall $\frac{N}{S}\to \gamma^2+O\left(\tfrac{1}{N}\right)$. As a result, after dropping all $o(N^{-1/3})$ terms, \eqref{eq_Factor_equation_10} becomes:
\begin{multline}
  \label{eq_Factor_equation_11}
  1- N^{-1/3} \gamma^{1/2} (1+\gamma)^{-4/3} \sqrt{(1+\gamma)^2}\left(\mathcal G^{(1)}(b)-\mathcal G^{(2)}(b)\right) +o(N^{-1/3})\\  = o(N^{-1/3})+ N^{-1/3} \tilde\theta  (1+\gamma)^2 \gamma^2 \frac{1}{\gamma(1+\gamma)} \left( \frac{1}{\gamma(1+\gamma)}
   +  \frac{\frac{1}{\gamma^2}-1}{(1+\gamma)^2}\right)
   \\+(1+\gamma)^2 \gamma^3\left(\frac{1}{\gamma(1+\gamma)}+  \frac{N^{-1/3}}{\gamma (1+\gamma)^{4/3}}\mathcal G^{(3)}(b) \right)\left(\frac{1}{\gamma(1+\gamma)}+  \frac{N^{-1/3}}{\gamma (1+\gamma)^{4/3}}\mathcal G^{(4)}(b)  +  \frac{\frac{1}{\gamma^2}-1}{(1+\gamma)^2}\right).
\end{multline}
The order $1$ terms cancel out and multiplying \eqref{eq_Factor_equation_11} by $N^{1/3}$, we finally get:
\begin{equation}
  \label{eq_Factor_equation_12}
 - \tilde\theta  \frac{1}{\gamma} + o(1)=
 \frac{1}{ (1+\gamma)^{1/3}}\mathcal G^{(3)}(b)  + \frac{\gamma}{ (1+\gamma)^{1/3}}\mathcal G^{(4)}(b)+\gamma^{1/2}  \frac{\left(\mathcal G^{(1)}(b)-\mathcal G^{(2)}(b)\right)}{(1+\gamma)^{1/3}}.
\end{equation}
At this step an algebraic miracle happens, leading to the appearance of the familiar $\mathcal T(\Theta)$:

\smallskip

{\bf Claim.} The right-hand side of \eqref{eq_Factor_equation_12} is the same random function as $(1+\gamma)^{2/3} \mathcal G(b)$.

\smallskip

\noindent For the proof, recall that noises entering into  $\mathcal G^{(k)}$, $k=1,2,3,4$, are $\left(\frac{\xi_j+\eta_j}{\sqrt{2}}\right)^2$, $\left(\frac{\xi_j-\eta_j}{\sqrt{2}}\right)^2$, $\xi_j^2$, $\eta_j^2$, respectively. Hence, in the linear combination of $\mathcal G^{(k)}$ of \eqref{eq_Factor_equation_12}, the noises combine into
$$
  \frac{1}{ (1+\gamma)^{1/3}}\xi_j^2  + \frac{\gamma}{ (1+\gamma)^{1/3}}\eta_j^2+\gamma^{1/2}  \frac{2\xi_j \eta_j}{(1+\gamma)^{1/3}}= \frac{\bigl(\xi_j+\gamma^{1/2}\eta_j\bigr)^2}{(1+\gamma)^{1/3}}\stackrel{d}{=} (1+\gamma)^{2/3} \bigl[\mathcal N(0,1)\bigr]^2.
$$
Using the claim,  the solution of the equation \eqref{eq_Factor_equation_12} converges to the solution of
$- \tilde\theta  \frac{1}{\gamma(1+\gamma)^{2/3}}= \mathcal G(b).$
Recognizing  $\kappa_2$ and comparing with Definition \ref{Definition_Transition_function}, we arrive at \eqref{eq_Transition_limit}.
\end{proof}

\subsection{Canonical Correlation Analysis} The proof of Theorem \ref{Theorem_main_convergence_statement} for CCA of Section \ref{Section_spiked_CCA} follows the same outline as in the previous three sections. However, an analogue of Propositions \ref{Proposition_equation_spiked_Wigner}, \ref{Proposition_PCA_equation}, \ref{Proposition_Factor_master} becomes even more complicated, which introduces new challenges in the proofs. The setting is symmetric under $N\leftrightarrow M$ and we assume $N\le M$. The reader may consult \citet{BG_review} for general information on CCA.

Suppose that we are given $N\times S$ matrix $\U$ and $M\times S$ matrix $\V$, in which the first rows are denoted $\u^*$ and $\v^*$, respectively. The remaining rows form matrices $\widetilde \U$ and $\widetilde \V$ of sizes $(N-1)\times S$ and $(M-1)\times S$, respectively. We would like to connect sample canonical correlations between $\U$ and $\V$ to $\u^*$, $\v^*$, and the sample canonical correlations and variables between $\widetilde \U$ and $\widetilde \V$. The latter are denoted $c_1,\dots,c_{N-1}$, $\u_1,\dots,\u_{N-1}$, $\v_1,\dots,\v_{M-1}$, where $1\ge c_1\ge \dots\ge c_{N-1}\ge 0$, $\u_i$ are $S$--dimensional vectors forming an orthonormal basis in the space spanned by rows of $\widetilde \U$, $\v_j$ are $S$--dimensional vectors forming an orthonormal basis in the space spanned by the rows of $\widetilde \V$, and $\langle \u_i,\v_j\rangle=c_i \delta_{i=j}$, $1 \le i \le N-1$, $1\le j \le M-1$. For the convenience of notation, we also introduce numbers $c_N=c_{N+1}=\dots=c_{M-1}=0$.

\begin{proposition}\label{Proposition_CCA_master_equation}
 For each squared sample canonical correlation $a$ between $\U$ and $\V$, either
\begin{equation}
\label{eq_CCA_master}
\begin{split}
&\Bigg[\langle \u^*, \v^*\rangle + \sum_{j=1}^{M-1}  \frac{\langle\u^*, \v_j\rangle(c_j \langle \v^*, \u_j\rangle -a \langle \v^*, \v_j\rangle)}{a-c_j^2}
 \, -a  \sum_{i=1}^{N-1}    \frac{\langle\u^*, \u_i\rangle(\langle \v^*, \u_i\rangle - c_i \langle \v^*, \v_i\rangle)}{a - c_i^2}   \Bigg]^2\\
&\qquad= a \left[-\langle\u^*,\u^*\rangle+\sum_{j=1}^{M-1}\frac{  \langle \u^*, \v_j \rangle^2 -2 c_j  \langle\u^*, \v_j\rangle  \langle \u^*, \u_j\rangle }{a-c_j^2} +a\sum_{i=1}^{N-1}   \frac{\langle \u^*, \u_i\rangle^2}{a - c_i^2} \right] \\
&\qquad\qquad\times \left[-\langle\v^*,\v^*\rangle+\sum_{i=1}^{N-1}  \frac{\langle \v^*, \u_i\rangle^2 - 2  c_i \langle \v^*, \u_i\rangle \langle \v^*, \v_i\rangle}{a -c_i^2} +a\sum_{j=1}^{M-1} \frac{\langle\v^*,\v_j\rangle^2}{a - c_j^2}\right],
\end{split}
\end{equation}
 or $a=c_i^2$ for $1\le i \le {N-1}$, where $c_i$ has multiplicity one, $\langle \u^*,\u_i\rangle=\langle \u^*,\v_i\rangle=\langle \v^*,\u_i\rangle=\langle \v^*,\v_i\rangle=0$, and the equation \eqref{eq_CCA_master} holds with the $i$-th terms excluded; or $a=c_i^2$, where $c_i$ has multiplicity larger than one.
\end{proposition}
\noindent We omit the proof, see  \citet[Appendix A]{BG_CCA}, with $N=K$.

\begin{corollary} \label{Corollary_CCA_master_restated} In the setting of Section \ref{Section_spiked_CCA}, for each $1\le k \le r$, the squared sample canonical correlations of $\U$ and $\V$ solve an equation in variable $a$
\begin{multline}
\label{eq_CCA_master_rest}
 \Biggl[ \sum_{i=1}^S \xi_i \eta_i + \sum_{j=1}^{M-1}  \frac{((1-c_j^2)^{\frac{1}{2}} \xi_{j+N-1}+c_j \xi_j)((1-a)c_j \eta_j -a (1-c_j^2)^{\frac{1}{2}} \eta_{j+N-1})}{a-c_j^2} \\ -a  \sum_{i=1}^{N-1}    \frac{(1-c_i^2)^{1/2}\xi_i\bigl((1-c_i^2)^{1/2}\eta_i - c_i \eta_{i+N-1}\bigr)}{a - c_i^2}   \Biggr]^2\\= a \left[- \sum_{i=1}^S \xi_i^2+\sum_{j=1}^{M-1}\frac{  \bigl((1-c_j^2)^{1/2}\xi_{j+N-1}+c_j \xi_j\bigr)^2 -2 c_j  \xi_j \bigl((1-c_j^2)^{1/2}\xi_{j+N-1}+c_j \xi_j\bigr) }{a-c_j^2} +a\sum_{i=1}^{N-1}   \frac{\xi_i^2}{a - c_i^2} \right] \\ \times \left[- \sum_{i=1}^S \eta_i^2+\sum_{i=1}^{N-1}  \frac{\eta_i^2 - 2  c_i \eta_i((1-c_i^2)^{1/2}\eta_{i+N-1}+c_i \eta_i\bigr)}{a -c_i^2} +a\sum_{j=1}^{M-1} \frac{((1-c_j^2)^{1/2}\eta_{j+N-1}+c_j \eta_j\bigr)^2}{a - c_j^2}\right],
\end{multline}
where $(\xi_i,\eta_i)$ are i.i.d.\ Gaussian mean $0$ two-dimensional vectors with covariance matrix
\begin{equation}
\label{eq_x61}
 \begin{pmatrix} \E \xi_i^2 & \E \xi_i \eta_i \\ \E \eta_i \xi_i & \E \eta_i^2\end{pmatrix}= \begin{pmatrix} C_{uu} & C_{uv} \\ C_{uv} & C_{vv}\end{pmatrix}, \qquad \frac{C_{uv}^2}{C_{uu} C_{vv}}=\theta_k,
\end{equation}
and $c_i$ are squared sample canonical correlations between $\tilde \U$, $\tilde \V$, which are $(N-1)\times S$ and $(M-1)\times S$ matrices, independent from $(\xi_i,\eta_i)$ and produced trough the same mechanism as $\U$, $\V$, but with $r$ smaller by $1$ and $\theta_k$ removed from $\{\theta_1,\theta_2,\dots,\theta_r\}$.
\end{corollary}
\begin{remark} \label{Remark_CCA_cov}
 Due to invariance of \eqref{eq_CCA_master_rest} under multiplications of $\xi_i$ or $\eta_j$ by constants, there is no loss of generality in assuming $C_{uu}=C_{vv}=1$.
\end{remark}
\begin{proof}[Proof of Corollary \ref{Corollary_CCA_master_restated}] We fix $1 \le k \le r$ and connect Proposition \ref{Proposition_CCA_master_equation} to  Section \ref{Section_spiked_CCA}. We choose two deterministic invertible matrices $\widetilde A$ of size $N\times N$ and $\widetilde B$ of size $M\times M$, so that:
\begin{itemize}
 \item The first coordinate of $\widetilde A \u$ is independent from all the coordinates other than the first in $\widetilde A\u$ and in $\widetilde B \v$.
 \item The first coordinate of $\widetilde B \v$ is independent from all the coordinates other than the first in $\widetilde A\u$ and in $\widetilde B \v$.
 \item The variances of the first coordinates of $\widetilde A \u$ and $\widetilde B \v$ are $1$.
 \item The squared correlation coefficient between the first coordinates of $\widetilde A \u$ and $\widetilde B \v$ is $\theta_k$.
\end{itemize}
The existence of such $\widetilde A$ and $\widetilde B$ is the basic statement in CCA, equivalent to \eqref{eq_CCA_model} -- the first coordinates of $\widetilde A \u$ and $\widetilde B \v$ are canonical variables for the canonical correlation $\theta_k$.

Multiplying the matrices $\U$ and $\V$ by $\widetilde A$ and $\widetilde B$, respectively, does not change the squared sample canonical correlations, and brings them to the form of Proposition \ref{Proposition_CCA_master_equation}. It remains to explain that \eqref{eq_CCA_master} is the same as \eqref{eq_CCA_master_rest}. The correlation structure between the last $N-1$ rows of $\widetilde A\U$ and the last $M-1$ rows of $\widetilde B \V$ is as in \eqref{eq_CCA_model} but with $\theta_k$ removed from $\{\theta_1,\theta_2,\dots,\theta_r\}$. Hence $c_i^2$ in \eqref{eq_CCA_master} match their description in Corollary \ref{Corollary_CCA_master_restated}. The first rows of $\widetilde A\U$ and $\widetilde B\V$ are $S$--dimensional vectors with i.i.d.\ components, and with correlation structure as in \eqref{eq_x61}. We should further understand the joint distribution of the four arrays of random scalar products appearing in \eqref{eq_CCA_master}.
\begin{equation}
\label{eq_x62}
 \langle\u^*, \u_i\rangle,\quad  \langle\v^*, \u_i\rangle, \quad 1\le i \le N-1,\qquad   \langle\u^*, \v_j\rangle, \quad \langle\v^*, \v_j\rangle,\quad 1\le j\le M-1.
\end{equation}
If all $\u_i$,$\v_j$ were orthonormal, we could use that scalar products of i.i.d.\ mean $0$ Gaussian vectors with orthonormal basis are again i.i.d.\ Gaussian vectors. However, $\langle \u_i,\v_i\rangle\ne 0$. In order to fix this, we introduce new vectors $\v'_j$, $1\le j \le M-1$, which are:
$$
  \frac{\v_1-c_1\u_1}{\sqrt{1-c_1^2}},\, \frac{\v_2-c_2\u_2}{\sqrt{1-c_2^2}},\dots, \frac{\v_{N-1}-c_{N-1}\u_{N-1}}{\sqrt{1-c_{ N-1}^2}}, \qquad \v_N,\dots,\v_{M-1}.
$$
Since $\{\u_i\}$ are orthonormal,  $\{\v_j\}$ are orthonormal, and $\langle \u_i,\v_j\rangle=c_i \delta_{i=j}$, we conclude that the $N+M-2$ vectors $\u_1,\dots,\u_{N-1}$, $\v'_1,\dots,\v'_{M-1}$ are also orthonormal.

Since $\u^*, \v^*$ are Gaussian and independent from all $\u_j$, $\v_j$, the scalar products $\langle\u^*, \u_i\rangle, \langle\v^*, \u_i\rangle, \langle\u^*, \v'_j\rangle, \langle\v^*, \v'_j\rangle$ are also Gaussian and independent from all $\u_j$ and $\v_j$. Among these scalar products, he only non-zero covariances are:
$$
 \E [\langle\u^*, \u_i\rangle]^2=\E[\langle\v^*, \u_i\rangle]^2=\E[\langle\u^*, \v'_j\rangle]^2=\E[\langle\v^*, \v'_j\rangle]^2=1.
$$
$$
  \E \langle\u^*, \u_i\rangle \langle\v^*, \u_i\rangle=  \E \langle\u^*, \v'_j\rangle \langle\v^*, \v'_j\rangle=\sqrt{\theta_k}.
$$
Let us choose an orthonormal basis $\w_1,\dots,\w_S$ of the $S$--dimensional space, such that the first $N+M-2$ vectors are $\u_1,\dots,\u_{N-1}$, $\v'_1,\dots,\v'_{M-1}$ and the rest are arbitrary (the choice is independent from $\u^*$ and $\v^*$). Define
$
 \xi_i=\langle \u^*, \w_i\rangle$, $\eta_i=\langle \v^*,\w_i\rangle$.
Then $(\xi_i,\eta_i)$, $1\le i \le S$, are i.i.d.\ vectors with correlation structure \eqref{eq_x61}. We express various scalar products in \eqref{eq_CCA_master} through $\xi_i$ and $\eta_i$:
$
 \langle \u^*, \v^*\rangle=\sum_{i=1}^S \xi_i \eta_i, \quad \langle\u^*,\u^*\rangle= \sum_{i=1}^S \xi_i^2, \quad \langle\v^*,\v^*\rangle=\sum_{i=1}^S \eta_i^2;
$
$
  \langle\u^*, \u_i\rangle=\xi_i,\,\,\, \langle\v^*, \u_i\rangle=\eta_i;
$
$
   \langle\u^*, \v_j\rangle=\sqrt{1-c_j^2}\,\xi_{j+N-1}+c_j \xi_j,\,\,\, \langle\v^*, \v_j\rangle=\sqrt{1-c_j^2}\,\eta_{j+N-1}+c_j \eta_j.
$
Plugging these expressions into \eqref{eq_CCA_master}, we get \eqref{eq_CCA_master_rest}.
\end{proof}

We record a version of Corollaries \ref{Corollary_Wigner_interlacement}, \ref{Corollary_spiked_covariance_interlacement}, \ref{Corollary_factor_interlacement}, which is slightly more complicated, leading to an additional step in the proofs, on which we comment in Remark \ref{Remark_CCA_add_step}.

\begin{corollary} \label{Corollary_CCA_interlacing} Let $N$ roots of \eqref{eq_CCA_master} be denoted $a_1,\dots, a_N$. Then all $a_i$ are real numbers between $0$ and $1$. If we arrange $a_i$ in the decreasing order, then there exists another sequence of ${N-1}$ real numbers $y_1\ge y_2\ge \dots\ge y_{{N-1}}$, such that two interlacing conditions hold:
\begin{equation}
\label{eq_CCA_interlacing}
 a_1\ge y_1 \ge a_2 \ge \dots \ge y_{N-1}\ge a_{N},\qquad \text{ and } \qquad
 y_1 \ge c_1^2 \ge y_2 \ge \dots \ge y_{N-1} \ge c_{N-1}^2.
\end{equation}
In the notation of the proof of Corollary \ref{Corollary_CCA_master_restated}, the numbers $y_1,\dots,y_{N-1}$ are squared sample canonical correlations between $\widetilde \U=$(last $N-1$ rows of $\widetilde A \U$) and $\V$ .
\end{corollary}
\begin{proof} This is a version of Lemma A.6 in \citet{BG_CCA}. $a_1,\dots, a_N$ are eigenvalues of the product of projections on $\U$ and on $\V$ in $S$--dimensional space, i.e., $P_\U P_\V$ (or $P_\V P_\U P_\V$ or $P_{\U} P_\V P_\U$), while $y_1,\dots,y_{N-1}$ are eigenvalues of the product of projections on smaller $\widetilde \U$ and $\V$. Then \eqref{eq_CCA_interlacing} are two instances interlacing inequalities between eigenvalues of a matrix and its principal submatrix, as in \citet[Corollry III.1.5]{bhatia_matrix}.
\end{proof}

Next, we prove an analogue of Propositions \ref{Proposition_Wigner_far_spikes}, \ref{Proposition_covariance_far_spikes}, \ref{Proposition_factor_far_spikes}.

\begin{proposition} \label{Proposition_CCA_far_spikes}
 Consider the CCA model, with $\theta_1>\theta_2>\dots>\theta_r\ge 0$ split into  $\theta_1,\dots,\theta_{q-1}>\theta^c=\frac{1}{\sqrt{(\tau_M-1)(\tau_N-1)}}$ and $\theta_q,\dots,\theta_r<\theta^c=\frac{1}{\sqrt{(\tau_M-1)(\tau_N-1)}}$. Then, with $\frac{S}{N}=\tau_N+O\left(\tfrac{1}{N}\right)$, $\frac{S}{M}=\tau_M+O\left(\tfrac{1}{N}\right)$, using \eqref{eq_CCA_params}, in joint distribution we have:
  \begin{align}
 \label{eq_x63}&\lim_{N\to\infty} \sqrt{N}(\lambda_i-\lambda(\theta_i))\stackrel{d}{=} \mathcal N(0, V(\theta_i)),\qquad 1\le i \le q-1,\\
 \label{eq_x64} &\lim_{N\to\infty} N^{2/3}\left[\tfrac{(\sqrt{\tau_N-1}+ \sqrt{\tau_M-1})^{4/3}(\sqrt{\tau_N-1}\sqrt{\tau_M-1}-1)^{4/3}}{\tau_N^{5/3}\tau_M (\tau_N-1)^{1/6}(\tau_M-1)^{1/6}}\right]^{-1}(\lambda_{i}-
\lambda_+)\stackrel{d}{=} \aa_{i-q+1}, \qquad i\geq  q,
 \end{align}
 where $ \mathcal N(0, V(\theta_i))$ are independent over $i$ and with  points of the Airy$_1$ point process  $\{\aa_j\}_{j\ge 1}$. In addition, \eqref{eq_Assumption_local_law} hold for the sample squared canonical correlations between $\U$ and $\V$ with $h(x)=1$, $h(x)=\sqrt{x}$, $h(x)=\sqrt{1-x}$, and $h(x)=\sqrt{x(1-x)}$.
\end{proposition}

\begin{proof} The final statement, \eqref{eq_Assumption_local_law} with various choices of $h(x)$, follows by induction on $r$ from Theorem \ref{Theorem_from_Paul}, Corollary \ref{Corollary_local_with_square_root},  Corollary \ref{Corollary_local_with_square_root_2} for $r=0$, and using Lemma \ref{Lemma_preservation_local_law} with Corollary \ref{Corollary_CCA_interlacing} for the induction step. Further, various parameters are:
\begin{align}
\label{eq_Wachter}
&\lambda_\pm=\left(\sqrt{\tau_M^{-1}(1-\tau_N^{-1})}\pm \sqrt{\tau_N^{-1}(1-\tau_M^{-1})}\right)^2=\frac{\bigl(\sqrt{\tau_N-1}\pm \sqrt{\tau_M-1}\bigr)^2}{\tau_N \tau_M},\\ 
 \notag &\mu(x)\dd x=\frac{\tau_N}{2\pi } \frac{\sqrt{(x-\lambda_-)(\lambda_+-x)}}{x (1-x)} \mathbf 1_{[\lambda_-,\lambda_+]}\, \dd x,
  \\ \notag &m(z)= \frac{\tau_M^{-1}+\tau_N^{-1}-z+ \sqrt{ (z-\lambda_-)(z-\lambda_+)}}{2 \tau_N^{-1} z(z-1)} + \frac{1}{z},
\end{align}
which are the Wachter law and its Stieljes transform, respectively (see, e.g., \citet[Section 3]{BG_review}), and the constants of \eqref{eq_Assumption_Stieltjes_imaginary} are computed\footnote{Note that $1-\lambda_+= \frac{(\sqrt{\tau_N-1}\sqrt{\tau_M-1}-1)^2}{\tau_N\tau_M}$.} to be
\begin{align}
\s&=\frac{\sqrt{\lambda_+-\lambda_-}}{2\tau_N^{-1}\lambda_+(1-\lambda_+)}
=\tau_N^{5/2}\tau_M^{3/2}\frac{\sqrt[4]{ (\tau_N-1)(\tau_M-1)}}{(\sqrt{\tau_N-1}+ \sqrt{\tau_M-1})^2(\sqrt{\tau_N-1}\sqrt{\tau_M-1}-1)^2},
\\
\label{eq_x73} \m&=\frac{\tau_M^{-1}+\tau_N^{-1}-\lambda_+}{2 \tau_N^{-1} \lambda_+(\lambda_+-1)} + \frac{1}{\lambda_+}=\frac{\tau_N^{-1}-\tau_M^{-1}+(1-2 \tau_N^{-1})\lambda_+}{2 \tau_N^{-1} \lambda_+(1-\lambda_+)}=
\frac{\tau_M-\tau_N+(\tau_N -2 )\tau_M\lambda_+}{2 \tau_M \lambda_+(1-\lambda_+)}.
\end{align}
 Statements close to \eqref{eq_x63}, \eqref{eq_x64} are known from \citet{bao2019canonical}, \citet{yang2022limiting}. Alternatively, the proof can be obtained by the argument of Proposition \ref{Proposition_Wigner_far_spikes}, by induction on $r$ with case $r=0$ given in \citet{Johnstone_Jacobi,HanPanYang} and the step based on Corollary \ref{Corollary_CCA_master_restated}. We only highlight the key computation, which is an analogue of \eqref{eq_x36}, \eqref{eq_x36_2}.

 We recall \eqref{eq_CCA_master_rest} and analyze its asymptotic behavior for $a>\lambda_+$ and assuming that $\{c_i^2\}$ are squared sample canonical correlations between independent $\U$ and $\V$ of sizes $(N-1)\times S$ and $(M-1)\times S$, respectively, and with i.i.d.\ $\mathcal N(0,1)$ random variables as their elements. In this situation the joint distribution of $\{c_i^2\}$ is the JOE ensemble (see Section \ref{Section_beta_ensembles_a}), and their empirical measure converges to the Wachter law \eqref{eq_Wachter} as $N\to\infty$. According to Remark \ref{Remark_CCA_cov}, we use \eqref{eq_x61} with $C_{uu}=C_{vv}=1$, $C_{uv}=\sqrt{\theta_k}$; we also recall $c_j=0$ for $j\ge N$. We divide \eqref{eq_CCA_master_rest} by $N^2$ and split its sums into their expectations (conditional on $\{c_i\}$) and the mean $0$ parts. The left-hand side of \eqref{eq_CCA_master_rest} is transformed into:

\begin{multline}
\label{eq_x65}
 \Biggl[ \sqrt{\theta_k}\left(\frac{S-M+1}{N}  - \sum_{i=1}^{N-1}    \frac{a(1-a)}{N(a - c_i^2)} -  a     \frac{N-1}{N}\right)
 \\
  + \sum_{i=1}^S \frac{\xi_i \eta_i-\sqrt{\theta_k}}{N} + \sum_{j=1}^{N-1}\left[  \frac{((1-c_j^2)^{\frac{1}{2}} \xi_{j+N-1}+c_j \xi_j)((1-a)c_j \eta_j -a (1-c_j^2)^{\frac{1}{2}} \eta_{j+N-1})}{N(a-c_j^2)}+\frac{\sqrt{\theta_k}}{N}\right]\\-\sum_{j=N}^{M-1} \frac{\xi_{j+N-1}\eta_{j+N-1}-\sqrt{\theta_k}}{N}
   -a  \sum_{i=1}^{N-1}    \frac{(1-c_i^2)^{1/2}\xi_i\bigl((1-c_i^2)^{1/2}\eta_i - c_i \eta_{i+N-1}\bigr)-(1-c_i^2)\sqrt{\theta_k}}{N(a - c_i^2)}   \Biggr]^2,
\end{multline}
where the first line is of constant order and deterministic as $N\to\infty$; the second and third lines are of order $N^{-1/2}$ and become Gaussian as $N\to\infty$. The $[\cdot]$ factor in the third line of \eqref{eq_CCA_master_rest} is similarly transformed into:
\begin{multline}
\label{eq_x66}
 \left[ \frac{2N-S-2}{N}+\frac{ M-N  }{N}\cdot \frac{1}{a} +\sum_{i=1}^{N-1}   \frac{1-a}{N(a - c_i^2)}\right]
  - \sum_{i=1}^S \frac{\xi_i^2-1}{N}\\ +\sum_{j=1}^{M-1}\frac{  \bigl((1-c_j^2)^{1/2}\xi_{j+N-1}+c_j \xi_j\bigr)^2 -2 c_j  \xi_j \bigl((1-c_j^2)^{1/2}\xi_{j+N-1}+c_j \xi_j\bigr) -1 +2 c^2_j }{N(a-c_j^2)} +a\sum_{i=1}^{N-1}   \frac{\xi_i^2-1}{N(a - c_i^2)}.
\end{multline}
The fourth line of \eqref{eq_CCA_master_rest} is transformed into:
\begin{multline}
\label{eq_x67}
\left[\frac{M+N-S-2}{N}+\sum_{i=1}^{N-1}  \frac{1-a}{N(a -c_i^2)} \right]
- \sum_{i=1}^S \frac{\eta_i^2-1}{N}\\+\sum_{i=1}^{N-1}  \frac{\eta_i^2 - 2  c_i \eta_i((1-c_i^2)^{1/2}\eta_{i+N-1}+c_i \eta_i\bigr)-1 +2  c_i^2 }{N(a -c_i^2)} +a\sum_{j=1}^{M-1} \frac{((1-c_j^2)^{1/2}\eta_{j+N-1}+c_j \eta_j\bigr)^2-1}{N(a - c_j^2)}.
\end{multline}
Hence, using the convergence towards the Wachter law (see Theorem \ref{Theorem_from_Paul} or \citet[Sections 4.4, 9.13, 9.14, 9.15]{Bai_Silverstein} or \cite{dumitriu2012global}), in the leading order the equation \eqref{eq_CCA_master_rest} becomes as $N\to\infty$:
\begin{multline}
\label{eq_x68}
 \theta_k \bigl[\left(\tau_N -\tau_N \tau_M^{-1}  -  a(1-a) m(a) -  a \right)\bigr]^2\\ =a
 \bigl[ 2-\tau_N + \frac{\tau_{N} \tau_M^{-1}-1}{a} +(1-a)m(a)\bigr]\bigl[\tau_N \tau_M^{-1} +1 -\tau_N +(1-a) m(a) \bigr]+ O\bigl(N^{-1/2}\bigr).
\end{multline}
Plugging $m(a)$ from \eqref{eq_Wachter}, simplifying (see \citet[Lemma B.6]{BG_CCA} for such a computation), and solving the resulting equation in $a$, we express the solution as $a=\frac{\bigl(  (\tau_N-1)\theta  + 1 \bigr) \bigl(  (\tau_M-1) \theta + 1\bigr)}{\theta \tau_N \tau_M }+O(N^{-1/2})$, thus matching the expression for $\lambda(\theta)$ in \eqref{eq_CCA_params}.

Further, similarly to the proofs of Propositions \ref{Proposition_Wigner_far_spikes}, \ref{Proposition_covariance_far_spikes}, and \ref{Proposition_factor_far_spikes}, applying CLT to the second-order terms in \eqref{eq_x65}, \eqref{eq_x66}, \eqref{eq_x67}, we replace $O(N^{-1/2})$ in \eqref{eq_x68} with $N^{-1/2}$ times $\mathcal N(0,1)$ times an explicit function of $a$. Then solving \eqref{eq_x68} again, we arrive at the statement on the asymptotic Gaussianity of the solution $a$, which matches \eqref{eq_CCA_params}.
\end{proof}

\begin{remark} \label{Remark_CCA_add_step} In Proposition \ref{Proposition_Wigner_far_spikes} we located the largest eigenvalues by analyzing \eqref{eq_spiked_sym} for $a$ near each $\lambda(\theta_i)$ and near $\lambda_+$. We used the interlacements of Corollary \ref{Corollary_Wigner_interlacement} to make sure that we did not miss any eigenvalues. Similarly in the above proof, by analyzing \eqref{eq_CCA_master_rest}, we can locate solution $a$ near each $\lambda(\theta_i)$ and near $\lambda_+$. The argument for not missing other eigenvalues needs to be modified: if we use \eqref{eq_CCA_interlacing} na\"ively, then, say, in the $r=1$ case it would allow two eigenvalues $\lambda_1\ge \lambda_2$ larger than $c_1^2$, rather than only one in Corollary \ref{Corollary_Wigner_interlacement}.

We still use \eqref{eq_CCA_interlacing}, but compare with $y_1,\dots,y_{N-1}$ instead of $c_1^2,\dots,c_N^2$. In the induction step of the argument, the eigenvalues $y_1,y_2,\dots$, satisfy the same induction assumption and the same asymptotics as $c_1^2, c_2^2,\dots$ (the only difference is that $N$ got decreased by $1$). With this update, the same arguments, as before, go through. The same applies below.
\end{remark}

\begin{proof}[Proof of Theorem \ref{Theorem_main_convergence_statement} for the  CCA setting of Section \ref{Section_spiked_CCA}] The constants \eqref{eq_transition_parameters} evaluate to:
\begin{eqnarray}
 \,\, V'(\theta^c)&=& \frac{4\sqrt{\tau_N-1}\sqrt{\tau_M-1}\bigl(\sqrt{\tau_N-1}\sqrt{\tau_M-1}-1\bigr)^2\bigl(\sqrt{\tau_N-1}+\sqrt{\tau_M-1}\bigr)^2}{ \tau_N^3\tau_M^2}, \\
 \lambda''(\theta^c)&=&2\frac{(\tau_N-1)^{3/2}(\tau_M-1)^{3/2}}{ \tau_N\tau_M}~, \\
\label{eq_x77}
  \kappa_1&=&\frac{1}{2}  {\frac{[V'(\theta^c)]^{\frac23}}{\lambda''(\theta^c)^{\frac13}}}=  \frac{\bigl(\sqrt{\tau_N-1}\sqrt{\tau_M-1}-1\bigr)^{4/3}\bigl(\sqrt{\tau_N-1}+\sqrt{\tau_M-1}\bigr)^{4/3}}{ \tau_N^{5/3}\tau_M(\tau_N-1)^{1/6}(\tau_M-1)^{1/6}}, \\
\label{eq_x78}
   \kappa_2&=& {\frac{[\lambda''(\theta^c)]^{\frac23}}{V'(\theta^c)^{\frac13}}}=   \frac{ (\tau_N-1)^{5/6}(\tau_M-1)^{5/6} \tau_N^{1/3}}{\bigl(\sqrt{\tau_N-1}\sqrt{\tau_M-1}-1\bigr)^{2/3}\bigl(\sqrt{\tau_N-1}+\sqrt{\tau_M-1}\bigr)^{2/3}}.
\end{eqnarray}
This matches the constants in Proposition \ref{Proposition_CCA_far_spikes} and, hence, Assumption \ref{Assumption_for_limit} will be satisfied.
We analyze the equation \eqref{eq_CCA_master_rest}    for $k=q$, i.e., we study
\begin{equation}
\label{eq_CCA_master_rest_2}\begin{split}
&\Biggl[  \sum_{i=1}^S \xi_i \eta_i + \sum_{j=1}^{M-1}  \frac{((1-c_j^2)^{\frac{1}{2}} \xi_{j+N-1}+c_j \xi_j)((1-a)c_j \eta_j -a (1-c_j^2)^{\frac{1}{2}} \eta_{j+N-1})}{a-c_j^2} \\
&\qquad\qquad -a  \sum_{i=1}^{N-1}    \frac{(1-c_i^2)^{1/2}\xi_i\bigl((1-c_i^2)^{1/2}\eta_i - c_i \eta_{i+N-1}\bigr)}{a - c_i^2}   \Biggr]^2\\
&\qquad= a \Biggl[- \sum_{i=1}^S \xi_i^2
+a\sum_{i=1}^{N-1}   \frac{\xi_i^2}{a - c_i^2}  \\
&\qquad\qquad +\sum_{j=1}^{M-1}\frac{  \bigl((1-c_j^2)^{1/2}\xi_{j+N-1}+c_j \xi_j\bigr)^2 -2 c_j  \xi_j \bigl((1-c_j^2)^{1/2}\xi_{j+N-1}+c_j \xi_j\bigr) }{a-c_j^2} \Biggr] \\
&\qquad \times \Biggl[- \sum_{i=1}^S \eta_i^2+a\sum_{j=1}^{M-1} \frac{((1-c_j^2)^{1/2}\eta_{j+N-1}+c_j \eta_j\bigr)^2}{a - c_j^2}   \\
&\qquad\qquad +\sum_{i=1}^{N-1}  \frac{\eta_i^2 - 2  c_i \eta_i((1-c_i^2)^{1/2}\eta_{i+N-1}+c_i \eta_i\bigr)}{a -c_i^2}\Biggr],
\end{split}
\end{equation}
where $(\xi_i,\eta_i)$ are mean $0$, variance $1$ Gaussian i.i.d.\ in $i$ random variables with the squared correlation coefficient of $\xi_i$ and $\eta_i$ equal to $\theta_q$,  the asymptotics of $(c_i^2)_{i=1}^{N-1}$ is given to us by Proposition \ref{Proposition_CCA_far_spikes}; $c_j=0$ for $j\ge N$.  Arguing as in the previous sections, the $q-1$ largest roots of the equation are close to $\lambda_1,\dots,\lambda_{q-1}$, resulting in \eqref{eq_Gaussian_limit}. In order to establish \eqref{eq_Transition_limit}, we need to investigate  \eqref{eq_CCA_master_rest_2} for $a$ close to $\lambda_+=\frac{(\sqrt{\tau_N-1}+\sqrt{\tau_N-1})^2}{\tau_N \tau_M}$ and locate the root of the equation in the $(\lambda_q,\lambda_{q-1})$ interval. For this computation, we can approximate $\sum_{i=1}^S \xi_i^2\approx S$, $ \sum_{j=1}^S \eta_j^2\approx S$, and $\sum_{i=1}^S\xi_i \eta_i\approx S \sqrt{\theta_q}$ because the relative errors in these approximations are of order $N^{-1/2}$, which is smaller than $N^{-1/3}$ scale of our interest. Similarly, the sums $\sum_{j=N}^{M-1}$ can be replaced with their expectations. Hence, up to $O(N^{-1/2})$ error, \eqref{eq_CCA_master_rest_2} becomes

\begin{multline}
\label{eq_CCA_master_rest_3}
 \Biggl[ \frac{S+N-M}{N}\sqrt{\theta_q} + \sum_{j=1}^{N-1}  \frac{((1-c_j^2)^{\frac{1}{2}} \xi_{j+N-1}+c_j \xi_j)((1-a)c_j \eta_j -a (1-c_j^2)^{\frac{1}{2}} \eta_{j+N-1})}{N(a-c_j^2)} \\ -a  \sum_{i=1}^{N-1}    \frac{(1-c_i^2)^{1/2}\xi_i\bigl((1-c_i^2)^{1/2}\eta_i - c_i \eta_{i+N-1}\bigr)}{N(a - c_i^2)}   \Biggr]^2\\= a \Biggl[-\frac{S}{N}+\frac{M-N }{N a}+\sum_{j=1}^{N-1}\frac{  \bigl((1-c_j^2)^{1/2}\xi_{j+N-1}+c_j \xi_j\bigr)^2 -2 c_j  \xi_j \bigl((1-c_j^2)^{1/2}\xi_{j+N-1}+c_j \xi_j\bigr) }{N(a-c_j^2)} \\ +a\sum_{i=1}^{N-1}   \frac{\xi_i^2}{N(a - c_i^2)} \Biggr] \\ \times \left[\frac{-S+M-N}{N}+\sum_{i=1}^{N-1}  \frac{\eta_i^2 - 2  c_i \eta_i((1-c_i^2)^{1/2}\eta_{i+N-1}+c_i \eta_i\bigr)}{N(a -c_i^2)} +a\sum_{j=1}^{N-1} \frac{((1-c_j^2)^{1/2}\eta_{j+N-1}+c_j \eta_j\bigr)^2}{N(a - c_j^2)}\right].
\end{multline}
We change the variables $b=N^{2/3}\frac{a-\lambda_+}{\kappa_1}$, $a=\lambda_++N^{-2/3} \kappa_1 b$,
recall that $\theta^c=\frac{1}{\sqrt{(\tau_N-1)(\tau_M-1)}}$, $\theta_q=\frac{1}{\sqrt{(\tau_N-1)(\tau_M-1)}}+N^{-1/3}\tilde\theta $,
and apply Theorem \ref{Theorem_as_convergence} to \eqref{eq_CCA_master_rest_3}. Arguing exactly as for the factor model in the previous section,
we need several applications of the theorem, leading to several functions $\mathcal G^{(1)}(b)$, $\mathcal G^{(2)}(b)$, \dots, which are then recombined together. The leading deterministic terms recombine into the same expression as \eqref{eq_x68} evaluated at $a=\lambda_+$. Recalling that $m(\lambda_+)=\m$, as computed in the proof of Proposition \ref{Proposition_CCA_far_spikes}, the first two lines of \eqref{eq_CCA_master_rest_3}, up to $o(N^{-1/3})$ error, become:
\begin{multline}
\label{eq_x69}
 \biggl[\sqrt{(\tau_N-1)^{-1/2}(\tau_M-1)^{-1/2}+N^{-1/3}\tilde\theta } \, \left(\tau_N -\tau_N \tau_M^{-1}  -  \lambda_+(1-\lambda_+) \m -  \lambda_+ \right)
 \\ + \frac{N^{-1/3}}{\kappa_1} \sum_{j=1}^\infty  \tfrac{\bigl((1-\lambda_+)^{\frac{1}{2}} \check \xi_{j}+\sqrt{\lambda_+} \xi_j\bigr)\bigl((1-\lambda_+)\sqrt{\lambda_+} \eta_j -\lambda_+ (1-\lambda_+)^{\frac{1}{2}} \check \eta_{j}\bigr)\, - \, \lambda_+ (1-\lambda_+)^{1/2}\xi_j\bigl((1-\lambda_+)^{1/2}\eta_j - \sqrt{\lambda_+} \check \eta_{j}\bigr)}{b-\aa_j}\biggr]^2
 \\
 = (\tau_N-1)^{-1/2}(\tau_M-1)^{-1/2} \bigl(\tau_N -\tau_N \tau_M^{-1}  -  \lambda_+(1-\lambda_+) \m -  \lambda_+\bigr)^2 \biggl[1+ 4\tilde \theta N^{-1/3} (\tau_N-1)^{1/2}(\tau_M-1)^{1/2}
 \\
 + 2 N^{-1/3} \frac{ (\tau_N-1)^{1/2}(\tau_M-1)^{1/2} (1-\lambda_+)\sqrt{\lambda_+}}{\kappa_1  \bigl(\tau_N -\tau_N \tau_M^{-1}  -  \lambda_+(1-\lambda_+) \m -  \lambda_+\bigr)} \sum_{j=1}^\infty  \tfrac{\check \xi_{j}\bigl((1-\lambda_+)^{1/2} \eta_j -\sqrt{\lambda_+} \check \eta_{j}\bigr)\, }{b-\aa_j}+o(N^{-1/3})\biggr],
\end{multline}
where the sum $\sum\limits_{j=1}^\infty$ is a shortcut for linear combinations of several functions $\mathcal G(b)$ which should be formally understood as in \eqref{eq_G_definition}; and  $(\xi_j, \eta_j, \check \xi_j, \check \eta_j)$ are i.i.d.\ Gaussian random vectors which are coordinate-wise  $\mathcal N(0,1)$, $\E \xi_j \eta_j=\E \check \xi_j \check \eta_j=\sqrt{\theta^c}$, and all the other covariances are zero.
Similarly, the third and fourth lines of \eqref{eq_CCA_master_rest_3} turn, up to $o(N^{-1/3})$ error, into:
\begin{multline} \label{eq_x70}
 \lambda_+
 \left(2-\tau_N + \frac{\tau_{N} \tau_M^{-1}-1}{\lambda_+} +(1-\lambda_+)\m\right) \biggl[1+ \frac{N^{-1/3}}{\kappa_1 (2-\tau_N - \frac{\tau_{N} \tau_M^{-1}-1}{\lambda_+} +(1-\lambda_+)\m)}
 \\ \times\sum_{j=1}^{\infty}\frac{  \bigl((1-\lambda_+)^{1/2}\check \xi_{j}+\sqrt{\lambda_+} \xi_j\bigr)^2 -2 \sqrt{\lambda_+}  \xi_j \bigl((1-\lambda_+)^{1/2}\check \xi_{j}+\sqrt{\lambda_+}\xi_j\bigr)+ \lambda_+ \xi_j^2 }{b-\aa_j}  \biggr]
 \\= \lambda_+
 \left(2-\tau_N + \frac{\tau_{N} \tau_M^{-1}-1}{\lambda_+} +(1-\lambda_+)\m\right) \biggl[1+ \frac{N^{-1/3}(1-\lambda_+) }{\kappa_1 (2-\tau_N - \frac{\tau_{N} \tau_M^{-1}-1}{\lambda_+} +(1-\lambda_+)\m)}
 \sum_{j=1}^{\infty}\frac{  \check \xi_{j}^2}{b-\aa_j}  \biggr].
\end{multline}
The last line of \eqref{eq_CCA_master_rest_3} becomes, up to $o(N^{-1/3})$ error:
\begin{multline} \label{eq_x71}
 \bigl(\tau_N \tau_M^{-1} +1 -\tau_N +(1-\lambda_+) \m \bigr) \biggl[1+ \frac{N^{-1/3}}{\kappa_1 (\tau_N \tau_M^{-1} +1 -\tau_N +(1-\lambda_+) \m)}
 \\ \times  \sum_{j=1}^{\infty}  \frac{\eta_j^2 - 2  \sqrt{\lambda_+} \eta_j((1-\lambda_+)^{1/2}\check \eta_{j}+\sqrt{\lambda_+} \eta_j\bigr)+\lambda_+((1-\lambda_+)^{1/2}\check \eta_{j}+\sqrt{\lambda_+}\eta_j\bigr)^2}{b -\aa_j}  \biggr]
 \\
 =  \bigl(\tau_N \tau_M^{-1} +1 -\tau_N +(1-\lambda_+) \m \bigr) \biggl[1+ \frac{N^{-1/3} (1-\lambda_+)}{\kappa_1 (\tau_N \tau_M^{-1} +1 -\tau_N +(1-\lambda_+) \m)}  \sum_{j=1}^{\infty}  \frac{((1-\lambda_+)^{1/2}\eta_j  -\sqrt{\lambda_+}\check \eta_{j})^2}{b -\aa_j}  \biggr].
\end{multline}
Equating $\eqref{eq_x69}=\eqref{eq_x70}\cdot \eqref{eq_x71}$, noting that the leading term cancels (this is precisely the equation relating $\theta^c$ with $\lambda_+$, cf.\ \eqref{eq_x68} and the paragraph after it) and multiplying by $N^{1/3}$, we get:
\begin{multline} \label{eq_x72}
  -\tilde \theta \kappa_1 (\tau_N-1)^{1/2}(\tau_M-1)^{1/2}= 2 \tfrac{ (\tau_N-1)^{1/2}(\tau_M-1)^{1/2} (1-\lambda_+)\sqrt{\lambda_+}}{ \tau_N -\tau_N \tau_M^{-1}  -  \lambda_+(1-\lambda_+) \m -  \lambda_+} \sum_{j=1}^\infty  \tfrac{\check \xi_{j}\bigl((1-\lambda_+)^{1/2} \eta_j -\sqrt{\lambda_+} \check \eta_{j}\bigr)\, }{b-\aa_j}
  \\ -  \tfrac{(1-\lambda_+) \lambda_+ }{\tau_{N} \tau_M^{-1}-1+\lambda_+(2-\tau_N)  +\lambda_+(1-\lambda_+)\m)}
 \sum_{j=1}^{\infty}\frac{  \check \xi_{j}^2}{b-\aa_j}- \tfrac{ (1-\lambda_+)}{\tau_N \tau_M^{-1} +1 -\tau_N +(1-\lambda_+) \m}  \sum_{j=1}^{\infty}  \frac{((1-\lambda_+)^{1/2}\eta_j  -\sqrt{\lambda_+}\check \eta_{j})^2}{b -\aa_j}+ o(1).
\end{multline}
Similarly to the factor model in the previous section, at this step an algebraic miracle happens, leading to the appearance of exactly the same function $\mathcal T(\Theta)$ in the asymptotics. In order to see that we simplify the right-hand side of \eqref{eq_x72}. Let us plug the value of $\m$ from \eqref{eq_x73} and analyze the coefficient of $\frac{1}{b-\aa_j}$ in \eqref{eq_x72}, which is:
\begin{multline}
\label{eq_x74}
  \frac{4 \tau_M (\tau_N-1)^{1/4}(\tau_M-1)^{1/4}\sqrt{\lambda_+} (1-\lambda_+)}{ 2 \tau_M \tau_N    - \tau_M-\tau_N-\tau_N  \tau_M\lambda_+ }  \check \xi_{j}\bigl((1-\lambda_+)^{1/2} \eta_j -\sqrt{\lambda_+} \check \eta_{j}\bigr)
  \\
  +  \frac{2 \tau_M \lambda_+ (1-\lambda_+)  }{\tau_{M}-\tau_N+(\tau_N-2)\tau_M\lambda_+}
 \check \xi_{j}^2
  +\frac{2\tau_M \lambda_+(1-\lambda_+)}{\tau_N-\tau_M+\tau_N  (\tau_M-2)\lambda_+}   ((1-\lambda_+)^{1/2}\eta_j  -\sqrt{\lambda_+}\check \eta_{j})^2.
\end{multline}
We observe a complete square in the last equation, which would follow from the identity:
\begin{multline}
 \left[ \frac{4 \tau_M (\tau_N-1)^{1/4}(\tau_M-1)^{1/4}\sqrt{\lambda_+} (1-\lambda_+)}{ 2 \tau_M \tau_N    - \tau_M-\tau_N-\tau_N  \tau_M\lambda_+ }\right]^2\\ \stackrel{?}{=} 4 \left[\frac{2 \tau_M \lambda_+ (1-\lambda_+)  }{\tau_{M}-\tau_N+(\tau_N-2)\tau_M\lambda_+}\right]\left[\frac{2\tau_M \lambda_+(1-\lambda_+)}{\tau_N-\tau_M+\tau_N  (\tau_M-2)\lambda_+}\right].
\end{multline}
The last identity is equivalent to
\begin{multline*}
  \frac{ (\tau_N-1)^{1/2}(\tau_M-1)^{1/2} }{[ 2 \tau_M \tau_N    - \tau_M-\tau_N-\tau_N  \tau_M\lambda_+ ]^2} \stackrel{?}{=} \frac{ \lambda_+   }{[\tau_{M}-\tau_N+(\tau_N-2)\tau_M\lambda_+][\tau_N-\tau_M+\tau_N  (\tau_M-2)\lambda_+]}.
\end{multline*}
Plugging the value of $\lambda_+$ from \eqref{eq_Wachter}, we need to show:

\begin{multline}
  \frac{ \sqrt{\tau_N-1}\sqrt{\tau_M-1} }{[ 2 \tau_M \tau_N    - \tau_M-\tau_N-\bigl(\sqrt{\tau_N-1}+ \sqrt{\tau_M-1}\bigr)^2 ]^2}\\  \stackrel{?}{=} \frac{ \bigl(\sqrt{\tau_N-1}+ \sqrt{\tau_M-1}\bigr)^2  }{[\tau_{M}\tau_N-\tau_N^2+(\tau_N-2)\bigl(\sqrt{\tau_N-1}+ \sqrt{\tau_M-1}\bigr)^2][\tau_M\tau_N-\tau_M^2+ (\tau_M-2)\bigl(\sqrt{\tau_N-1}+ \sqrt{\tau_M-1}\bigr)^2]},
\end{multline}
which can be directly seen to be true by transforming the denominators:
{\small
\begin{align} \label{eq_x76}
 2 \tau_M \tau_N    - \tau_M-\tau_N-\bigl(\sqrt{\tau_N-1}+ \sqrt{\tau_M-1}\bigr)^2=2 \sqrt{\tau_N-1}\sqrt{\tau_M-1}(\sqrt{\tau_N-1}\sqrt{\tau_M-1}-1)&,
 \\
 \notag
 \tau_{M}\tau_N-\tau_N^2+(\tau_N-2)\bigl(\sqrt{\tau_N-1}+ \sqrt{\tau_M-1}\bigr)^2= 2\sqrt{\tau_N-1}(\sqrt{\tau_{N}-1}\sqrt{\tau_M-1}-1)(\sqrt{\tau_M-1}+\sqrt{\tau_N-1})&,
 \\
 \notag \tau_{N}\tau_M-\tau_M^2+(\tau_M-2)\bigl(\sqrt{\tau_N-1}+ \sqrt{\tau_M-1}\bigr)^2= 2\sqrt{\tau_M-1}(\sqrt{\tau_{N}-1}\sqrt{\tau_M-1}-1)(\sqrt{\tau_M-1}+\sqrt{\tau_N-1})&.
\end{align}
}
 Therefore, crucially, \eqref{eq_x74} is the square of a mean $0$ Gaussian
random variable. The variance of this random variable equals the expectation of \eqref{eq_x74}. Recalling that $\check \xi_j$, $\eta_j$, $\check \eta_j$ are $\mathcal N(0,1)$ with covariances
$$\E \check \xi_j \check \eta_j=\sqrt{\theta^c}= (\tau_N-1)^{-1/4}(\tau_M-1)^{-1/4}, \qquad \E \eta_j \check \xi_j=\E \eta_j \check \eta_j=0,
$$
using \eqref{eq_Wachter} and \eqref{eq_x76} we compute this expectation to be:
\begin{multline}
\label{eq_x75}
 2\tau_M (1-\lambda_+) \lambda_+ \Biggl[ \tfrac{-2 }{ 2 \tau_M \tau_N    - \tau_M-\tau_N-\tau_N  \tau_M\lambda_+ }
  +  \tfrac{ 1}{\tau_{M}-\tau_N+(\tau_N-2)\tau_M\lambda_+}
  +\tfrac{1}{\tau_N-\tau_M+\tau_N  (\tau_M-2)\lambda_+}  \Biggr]
\\
=
\tau_M^{-1} \tau_N^{-2} (\sqrt{\tau_N-1}\sqrt{\tau_M-1}-1) \bigl(\sqrt{\tau_N-1}+ \sqrt{\tau_M-1}\bigr)^2
\\ \times \biggl[ \tfrac{-2 }{ \sqrt{\tau_N-1}\sqrt{\tau_M-1}}
  +  \tfrac{ \tau_N}{\sqrt{\tau_N-1}(\sqrt{\tau_M-1}+\sqrt{\tau_N-1})}
  +\tfrac{\tau_M}{\sqrt{\tau_M-1}(\sqrt{\tau_M-1}+\sqrt{\tau_N-1})}  \biggr]
  \\
=
\tau_M^{-1} \tau_N^{-2} \frac{(\sqrt{\tau_N-1}\sqrt{\tau_M-1}-1)^2 (\sqrt{\tau_M-1}+\sqrt{\tau_N-1})^2}{\sqrt{\tau_N-1}\sqrt{\tau_M-1}}.
\end{multline}
We conclude that the equation \eqref{eq_x72} can be rewritten as
\begin{equation}
  -\tilde \theta \kappa_1   \tau_M \tau_N^2 \frac{(\tau_N-1)(\tau_M-1)}{(\sqrt{\tau_N-1}\sqrt{\tau_M-1}-1)^2 (\sqrt{\tau_M-1}+\sqrt{\tau_N-1})^2}=
  \mathcal G(b) +o(1).
\end{equation}
Plugging the value of $\kappa_1$ from \eqref{eq_x77}, we get
\begin{equation}
\label{eq_x79}
  -\tilde \theta   \frac{\tau_N^{1/3} (\tau_N-1)^{5/6}(\tau_M-1)^{5/6}}{(\sqrt{\tau_N-1}\sqrt{\tau_M-1}-1)^{2/3} (\sqrt{\tau_M-1}+\sqrt{\tau_N-1})^{2/3}}=
  \mathcal G(b)+o(1).
\end{equation}
Using \eqref{eq_x78}, recognizing $\kappa_2$ and comparing with Definition \ref{Definition_Transition_function}, we arrive at \eqref{eq_Transition_limit}.
\end{proof}

\linespread{1}
\bibliographystyle{abbrvnat}
\bibliography{Critical_biblio}

\begin{thebibliography}{122}
\providecommand{\natexlab}[1]{#1}
\providecommand{\url}[1]{\texttt{#1}}
\expandafter\ifx\csname urlstyle\endcsname\relax
  \providecommand{\doi}[1]{doi: #1}\else
  \providecommand{\doi}{doi: \begingroup \urlstyle{rm}\Url}\fi

\bibitem[Abramowitz and Stegun(1972)]{abramowitz1968handbook}
M.~Abramowitz and I.~A. Stegun.
\newblock \emph{Handbook of mathematical functions with formulas, graphs, and
  mathematical tables, Tenth Printing}.
\newblock US Government printing office, 1972.

\bibitem[Ahlfors(1979)]{Ahlfors_1979}
L.~V. Ahlfors.
\newblock \emph{Complex analysis}.
\newblock McGraw-Hill, New York, 3rd edition, 1979.

\bibitem[Aizenman and Warzel(2015)]{aizenman2015ubiquity}
M.~Aizenman and S.~Warzel.
\newblock On the ubiquity of the {C}auchy distribution in spectral problems.
\newblock \emph{Probability Theory and Related Fields}, 163:\penalty0 61--87,
  2015.

\bibitem[Alt et~al.(2025)Alt, Bourgade, and Pain]{bourgade2024optimal2}
J.~Alt, P.~Bourgade, and M.~Pain.
\newblock Local law for $\beta$--ensembles with a hard edge.
\newblock \emph{In preparation}, 2025.

\bibitem[Anderson(2003)]{anderson1958introduction}
T.~W. Anderson.
\newblock \emph{Introduction to multivariate statistical analysis, 3rd
  edition}.
\newblock John Wiley \& Sons, 2003.

\bibitem[Andreou et~al.(2019)Andreou, Gagliardini, Ghysels, and
  Rubin]{andreou2019inference}
E.~Andreou, P.~Gagliardini, E.~Ghysels, and M.~Rubin.
\newblock Inference in group factor models with an application to
  mixed-frequency data.
\newblock \emph{Econometrica}, 87\penalty0 (4):\penalty0 1267--1305, 2019.

\bibitem[Arbenz et~al.(1988)Arbenz, Gander, and Golub]{arbenz1988restricted}
P.~Arbenz, W.~Gander, and G.~H. Golub.
\newblock Restricted rank modification of the symmetric eigenvalue problem:
  Theoretical considerations.
\newblock \emph{Linear Algebra and its Applications}, 104:\penalty0 75--95,
  1988.

\bibitem[Bai(2003)]{bai2003inferential}
J.~Bai.
\newblock Inferential theory for factor models of large dimensions.
\newblock \emph{Econometrica}, 71\penalty0 (1):\penalty0 135--171, 2003.

\bibitem[Bai and Ng(2002)]{bai2002determining}
J.~Bai and S.~Ng.
\newblock Determining the number of factors in approximate factor models.
\newblock \emph{Econometrica}, 70\penalty0 (1):\penalty0 191--221, 2002.

\bibitem[Bai and Ng(2023)]{bai2023approximate}
J.~Bai and S.~Ng.
\newblock Approximate factor models with weaker loadings.
\newblock \emph{Journal of Econometrics}, 235\penalty0 (2):\penalty0
  1893--1916, 2023.

\bibitem[Bai and Silverstein(2010)]{Bai_Silverstein}
Z.~Bai and J.~W. Silverstein.
\newblock \emph{Spectral Analysis of Large Dimensional Random Matrices}.
\newblock Springer, 2nd edition, 2010.

\bibitem[Bai and Yao(2008)]{bai2008central}
Z.~Bai and J.-F. Yao.
\newblock Central limit theorems for eigenvalues in a spiked population model.
\newblock \emph{Annales de l'IHP Probabilit{\'e}s et statistiques},
  44:\penalty0 447--474, 2008.

\bibitem[Bai et~al.(2022)Bai, Hou, Hu, Jiang, and Zhang]{bai2022limiting}
Z.~Bai, Z.~Hou, J.~Hu, D.~Jiang, and X.~Zhang.
\newblock Limiting canonical distribution of two large-dimensional random
  vectors.
\newblock In \emph{Methodology and Applications of Statistics: A Volume in
  Honor of CR Rao on the Occasion of his 100th Birthday}, pages 213--238.
  Springer, 2022.

\bibitem[Bai et~al.(1988)Bai, Silverstein, and Yin]{bai1988note}
Z.~D. Bai, J.~W. Silverstein, and Y.~Q. Yin.
\newblock A note on the largest eigenvalue of a large dimensional sample
  covariance matrix.
\newblock \emph{Journal of multivariate analysis}, 26\penalty0 (2):\penalty0
  166--168, 1988.

\bibitem[Baik and Silverstein(2006)]{baik2006eigenvalues}
J.~Baik and J.~W. Silverstein.
\newblock Eigenvalues of large sample covariance matrices of spiked population
  models.
\newblock \emph{Journal of multivariate analysis}, 97\penalty0 (6):\penalty0
  1382--1408, 2006.

\bibitem[Baik et~al.(2005)Baik, Ben~Arous, and P{\'e}ch{\'e}]{baik2005phase}
J.~Baik, G.~Ben~Arous, and S.~P{\'e}ch{\'e}.
\newblock Phase transition of the largest eigenvalue for nonnull complex sample
  covariance matrices.
\newblock \emph{Annals of Probability}, pages 1643--1697, 2005.

\bibitem[Baik et~al.(2021)Baik, Collins-Woodfin, Le~Doussal, and
  Wu]{baik2021spherical}
J.~Baik, E.~Collins-Woodfin, P.~Le~Doussal, and H.~Wu.
\newblock Spherical spin glass model with external field.
\newblock \emph{Journal of Statistical Physics}, 183\penalty0 (2):\penalty0 31,
  2021.

\bibitem[Bailey et~al.(2016)Bailey, Kapetanios, and
  Pesaran]{bailey2016exponent}
N.~Bailey, G.~Kapetanios, and M.~H. Pesaran.
\newblock Exponent of cross-sectional dependence: Estimation and inference.
\newblock \emph{Journal of Applied Econometrics}, 31\penalty0 (6):\penalty0
  929--960, 2016.

\bibitem[Bailey et~al.(2021)Bailey, Kapetanios, and
  Pesaran]{bailey2021measurement}
N.~Bailey, G.~Kapetanios, and M.~H. Pesaran.
\newblock Measurement of factor strength: Theory and practice.
\newblock \emph{Journal of Applied Econometrics}, 36\penalty0 (5):\penalty0
  587--613, 2021.

\bibitem[Bao et~al.(2019)Bao, Hu, Pan, and Zhou]{bao2019canonical}
Z.~Bao, J.~Hu, G.~Pan, and W.~Zhou.
\newblock Canonical correlation coefficients of high-dimensional gaussian
  vectors: Finite rank case.
\newblock \emph{Annals of Statistics}, 47\penalty0 (1):\penalty0 612--640,
  2019.

\bibitem[Barigozzi and Hallin(2024)]{barigozzi2024dynamic}
M.~Barigozzi and M.~Hallin.
\newblock The dynamic, the static, and the weak factor models and the analysis
  of high-dimensional time series.
\newblock \emph{arXiv preprint arXiv:2407.10653}, 2024.

\bibitem[Bejan(2005)]{Bejan}
A.~Bejan.
\newblock Largest eigenvalues and sample covariance matrices. {T}racy-{W}idom
  and {P}ainlev\'{e} ii: computational aspects and realization in s-plus with
  applications.
\newblock \emph{Preprint:
  \url{http://users.stat.umn.edu/~jiang040/downloadpapers/largesteigen/largesteigen.pdf}},
  2005.

\bibitem[Benaych-Georges and Nadakuditi(2012)]{benaych2012singular}
F.~Benaych-Georges and R.~R. Nadakuditi.
\newblock The singular values and vectors of low rank perturbations of large
  rectangular random matrices.
\newblock \emph{Journal of Multivariate Analysis}, 111:\penalty0 120--135,
  2012.

\bibitem[Benaych-Georges et~al.(2011)Benaych-Georges, Guionnet, and
  Maida]{benaych2011fluctuations}
F.~Benaych-Georges, A.~Guionnet, and M.~Maida.
\newblock Fluctuations of the extreme eigenvalues of finite rank deformations
  of random matrices.
\newblock \emph{Electronic Journal of Probability}, 16:\penalty0 1621--1662,
  2011.

\bibitem[Bhatia(1997)]{bhatia_matrix}
R.~Bhatia.
\newblock \emph{Matrix analysis}.
\newblock Springer, 1997.

\bibitem[Billingsley(1999)]{Billingsley}
P.~Billingsley.
\newblock \emph{Convergence of probability measures}.
\newblock John Wiley \& Sons Inc., New York, 1999.

\bibitem[Bloemendal(2011)]{bloemendal2011finite}
A.~Bloemendal.
\newblock \emph{Finite rank perturbations of random matrices and their
  continuum limits, PhD Thesis}.
\newblock University of Toronto, 2011.

\bibitem[Bloemendal and Vir{\'a}g(2013)]{bloemendal2013limits}
A.~Bloemendal and B.~Vir{\'a}g.
\newblock Limits of spiked random matrices {I}.
\newblock \emph{Probability Theory and Related Fields}, 156:\penalty0 795--825,
  2013.

\bibitem[Bloemendal and Vir{\'a}g(2016)]{bloemendal2016limits}
A.~Bloemendal and B.~Vir{\'a}g.
\newblock Limits of spiked random matrices {II}.
\newblock \emph{The Annals of Probability}, pages 2726--2769, 2016.

\bibitem[Bloemendal et~al.(2016)Bloemendal, Knowles, Yau, and
  Yin]{bloemendal2016principal}
A.~Bloemendal, A.~Knowles, H.-T. Yau, and J.~Yin.
\newblock On the principal components of sample covariance matrices.
\newblock \emph{Probability theory and related fields}, 164\penalty0
  (1-2):\penalty0 459--552, 2016.

\bibitem[Bornemann(2009)]{bornemann2009numerical}
F.~Bornemann.
\newblock On the numerical evaluation of distributions in random matrix theory:
  a review.
\newblock \emph{arXiv preprint arXiv:0904.1581}, 2009.

\bibitem[Borot and Guionnet(2013)]{borot2013asymptotic}
G.~Borot and A.~Guionnet.
\newblock Asymptotic expansion of $\beta$ matrix models in the one-cut regime.
\newblock \emph{Communications in Mathematical Physics}, 317:\penalty0
  447--483, 2013.

\bibitem[Bourgade et~al.(2022)Bourgade, Mody, and Pain]{bourgade2022optimal}
P.~Bourgade, K.~Mody, and M.~Pain.
\newblock Optimal local law and central limit theorem for $\beta$-ensembles.
\newblock \emph{Communications in Mathematical Physics}, 390\penalty0
  (3):\penalty0 1017--1079, 2022.

\bibitem[Bykhovskaya and Gorin(2022)]{BG1}
A.~Bykhovskaya and V.~Gorin.
\newblock Cointegration in large {VAR}s.
\newblock \emph{Annals of {S}tatistics}, 50\penalty0 (3):\penalty0 1593--1617,
  2022.

\bibitem[Bykhovskaya and Gorin(2024)]{BG_review}
A.~Bykhovskaya and V.~Gorin.
\newblock Canonical correlation analysis: review.
\newblock \emph{arXiv preprint arXiv:2411.15625}, 2024.
\newblock \doi{10.48550/arXiv.2411.15625}.

\bibitem[Bykhovskaya and Gorin(2025)]{BG_CCA}
A.~Bykhovskaya and V.~Gorin.
\newblock High-dimensional canonical correlation analysis.
\newblock \emph{arXiv preprint arXiv:2306.16393}, 2025.

\bibitem[Bykhovskaya et~al.(2025)Bykhovskaya, Gorin, and
  Kiss]{vignette_largevars}
A.~Bykhovskaya, V.~Gorin, and E.~Kiss.
\newblock Largevars: an {R} package for testing large {VAR}s for the presence
  of cointegration.
\newblock \emph{arXiv preprint arXiv:2509.06295}, 2025.
\newblock \url{https://github.com/eszter-kiss/Largevars}.

\bibitem[Cai et~al.(2020)Cai, Han, and Pan]{cai2020limiting}
T.~T. Cai, X.~Han, and G.~Pan.
\newblock Limiting laws for divergent spiked eigenvalues and largest nonspiked
  eigenvalue of sample covariance matrices.
\newblock \emph{The Annals of Statistics}, 48\penalty0 (3):\penalty0
  1255--1280, 2020.

\bibitem[Capitaine et~al.(2009)Capitaine, Donati-Martin, and
  F{\'e}ral]{capitaine2009largest}
M.~Capitaine, C.~Donati-Martin, and D.~F{\'e}ral.
\newblock The largest eigenvalues of finite rank deformation of large {W}igner
  matrices: Convergence and nonuniversality of the fluctuations.
\newblock \emph{The Annals of Probability}, pages 1--47, 2009.

\bibitem[Capitaine et~al.(2012)Capitaine, Donati-Martin, and
  F{\'e}ral]{capitaine2012central}
M.~Capitaine, C.~Donati-Martin, and D.~F{\'e}ral.
\newblock Central limit theorems for eigenvalues of deformations of {W}igner
  matrices.
\newblock In \emph{Annales de l'IHP Probabilit{\'e}s et statistiques},
  volume~48, pages 107--133, 2012.

\bibitem[Chudik et~al.(2011)Chudik, Pesaran, and Tosetti]{chudik2011weak}
A.~Chudik, M.~Pesaran, and E.~Tosetti.
\newblock Weak and strong cross-section dependence and estimation of large
  panels.
\newblock \emph{Econometrics Journal}, 14:\penalty0 C45--C90, 2011.

\bibitem[Cochrane(2011)]{cochrane2011presidential}
J.~H. Cochrane.
\newblock Presidential address: Discount rates.
\newblock \emph{The Journal of finance}, 66\penalty0 (4):\penalty0 1047--1108,
  2011.

\bibitem[Collins-Woodfin and Le(2025)]{collins2025edge}
E.~Collins-Woodfin and H.~G. Le.
\newblock An edge {CLT} for the log determinant of {L}aguerre beta ensembles.
\newblock In \emph{Annales de l'Institut Henri Poincare (B) Probabilites et
  statistiques}, volume~61, pages 83--128. Institut Henri Poincar{\'e}, 2025.

\bibitem[Deift and Gioev(2009)]{deift2009random}
P.~Deift and D.~Gioev.
\newblock \emph{Random matrix theory: invariant ensembles and universality}.
\newblock American Mathematical Soc., 2009.

\bibitem[Ding and Yang(2018)]{ding2018necessary}
X.~Ding and F.~Yang.
\newblock A necessary and sufficient condition for edge universality at the
  largest singular values of covariance matrices.
\newblock \emph{The Annals of Applied Probability}, 28\penalty0 (3):\penalty0
  1679--1738, 2018.

\bibitem[Dobriban(2017)]{dobriban2017sharp}
E.~Dobriban.
\newblock Sharp detection in {PCA} under correlations: All eigenvalues matter.
\newblock \emph{The Annals of Statistics}, 45\penalty0 (4):\penalty0
  1810--1833, 2017.

\bibitem[Dumitriu and Edelman(2002)]{dumitriu_edelman}
I.~Dumitriu and A.~Edelman.
\newblock Matrix models for beta ensembles.
\newblock \emph{Journal of Mathematical Physics}, 43\penalty0 (11):\penalty0
  5830--5847, 2002.

\bibitem[Dumitriu and Paquette(2012)]{dumitriu2012global}
I.~Dumitriu and E.~Paquette.
\newblock Global fluctuations for linear statistics of $\beta$-jacobi
  ensembles.
\newblock \emph{Random Matrices: Theory and Applications}, 1\penalty0
  (04):\penalty0 1250013, 2012.

\bibitem[Durrett(2019)]{durrett2019probability}
R.~Durrett.
\newblock \emph{Probability: theory and examples, 5th edition}.
\newblock Cambridge university press, 2019.

\bibitem[Edelman and Persson(2005)]{edelman2005numerical}
A.~Edelman and P.-O. Persson.
\newblock Numerical methods for eigenvalue distributions of random matrices.
\newblock \emph{arXiv preprint math-ph/0501068}, 2005.

\bibitem[El~Alaoui et~al.(2020)El~Alaoui, Krzakala, and
  Jordan]{el2020fundamental}
A.~El~Alaoui, F.~Krzakala, and M.~Jordan.
\newblock Fundamental limits of detection in the spiked {W}igner model.
\newblock \emph{The Annals of Statistics}, 48\penalty0 (2):\penalty0 863--885,
  2020.

\bibitem[Erd{\H{o}}s and Yau(2017)]{erdos2017dynamical}
L.~Erd{\H{o}}s and H.-T. Yau.
\newblock \emph{A dynamical approach to random matrix theory}.
\newblock American Mathematical Soc., 2017.

\bibitem[Fan and Yao(2017)]{fan2017elements}
J.~Fan and Q.~Yao.
\newblock \emph{The elements of financial econometrics}.
\newblock Cambridge University Press, 2017.

\bibitem[Fan et~al.(2024{\natexlab{a}})Fan, Li, Xia, and Zheng]{fan2024tests}
J.~Fan, Y.~Li, N.~Xia, and X.~Zheng.
\newblock Tests for principal eigenvalues and eigenvectors.
\newblock \emph{arXiv preprint arXiv:2405.06939}, 2024{\natexlab{a}}.

\bibitem[Fan et~al.(2024{\natexlab{b}})Fan, Yan, and Zheng]{fan2024can}
J.~Fan, Y.~Yan, and Y.~Zheng.
\newblock When can weak latent factors be statistically inferred?
\newblock \emph{arXiv preprint arXiv:2407.03616}, 2024{\natexlab{b}}.

\bibitem[F{\'e}ral and P{\'e}ch{\'e}(2007)]{feral2007largest}
D.~F{\'e}ral and S.~P{\'e}ch{\'e}.
\newblock The largest eigenvalue of rank one deformation of large wigner
  matrices.
\newblock \emph{Communications in mathematical physics}, 272\penalty0
  (1):\penalty0 185--228, 2007.

\bibitem[Forrester(2010)]{forrest}
P.~Forrester.
\newblock \emph{Log-gases and random matrices}.
\newblock Princeton University Press, 2010.

\bibitem[Forrester(1993)]{forrester1993spectrum}
P.~J. Forrester.
\newblock The spectrum edge of random matrix ensembles.
\newblock \emph{Nuclear Physics B}, 402\penalty0 (3):\penalty0 709--728, 1993.

\bibitem[Forrester and Rains(2001)]{forrester2001interrelationships}
P.~J. Forrester and E.~M. Rains.
\newblock Interrelationships between orthogonal, unitary and symplectic matrix
  ensembles.
\newblock In \emph{Random Matrix Models and Their Applications}, pages
  171--207. Cambridge University Press, 2001.

\bibitem[Freyaldenhoven(2022)]{freyaldenhoven2022factor}
S.~Freyaldenhoven.
\newblock Factor models with local factors—determining the number of relevant
  factors.
\newblock \emph{Journal of Econometrics}, 229\penalty0 (1):\penalty0 80--102,
  2022.

\bibitem[F{\"u}redi and Koml{\'o}s(1981)]{furedi1981eigenvalues}
Z.~F{\"u}redi and J.~Koml{\'o}s.
\newblock The eigenvalues of random symmetric matrices.
\newblock \emph{Combinatorica}, 1:\penalty0 233--241, 1981.

\bibitem[Gabaix(2009)]{gabaix2009power}
X.~Gabaix.
\newblock Power laws in economics and finance.
\newblock \emph{Annu. Rev. Econ.}, 1\penalty0 (1):\penalty0 255--294, 2009.

\bibitem[Gavish and Donoho(2014)]{gavish2014optimal}
M.~Gavish and D.~L. Donoho.
\newblock The optimal hard threshold for singular values is $4/\sqrt{3}$.
\newblock \emph{IEEE Transactions on Information Theory}, 60\penalty0
  (8):\penalty0 5040--5053, 2014.

\bibitem[Giannone et~al.(2021)Giannone, Lenza, and
  Primiceri]{giannone2021economic}
D.~Giannone, M.~Lenza, and G.~E. Primiceri.
\newblock Economic predictions with big data: The illusion of sparsity.
\newblock \emph{Econometrica}, 89\penalty0 (5):\penalty0 2409--2437, 2021.

\bibitem[Giglio et~al.(2023)Giglio, Xiu, and Zhang]{giglio2023prediction}
S.~Giglio, D.~Xiu, and D.~Zhang.
\newblock Prediction when factors are weak.
\newblock \emph{Preprint}, 2023.

\bibitem[Gittins(1985)]{gittins1985canonical}
R.~Gittins.
\newblock \emph{Canonical Analysis: A Review with Applications in Ecology}.
\newblock Springer, 1985.

\bibitem[Gopikrishnan et~al.(1999)Gopikrishnan, Plerou, Amaral, Meyer, and
  Stanley]{gopikrishnan1999scaling}
P.~Gopikrishnan, V.~Plerou, L.~A.~N. Amaral, M.~Meyer, and H.~E. Stanley.
\newblock Scaling of the distribution of fluctuations of financial market
  indices.
\newblock \emph{Physical Review E}, 60\penalty0 (5):\penalty0 5305, 1999.

\bibitem[Han et~al.(2018)Han, Pan, and Yang]{HanPanYang}
C.~Han, G.~Pan, and Q.~Yang.
\newblock A unified matrix model including both {CCA} and {F} matrices in
  multivariate analysis: The largest eigenvalue and its applications.
\newblock \emph{Bernoulli}, 24\penalty0 (4B):\penalty0 3447--3468, 2018.

\bibitem[Holland et~al.(1983)Holland, Laskey, and
  Leinhardt]{Holland1983stochastic}
P.~W. Holland, K.~B. Laskey, and S.~Leinhardt.
\newblock Stochastic blockmodels: First steps.
\newblock \emph{Social networks}, 5\penalty0 (2):\penalty0 109--137, 1983.

\bibitem[Hou et~al.(2023)Hou, Zhang, Bai, and Hu]{hou2023spiked}
Z.~Hou, X.~Zhang, Z.~Bai, and J.~Hu.
\newblock Spiked eigenvalues of noncentral fisher matrix with applications.
\newblock \emph{Bernoulli}, 29\penalty0 (4):\penalty0 3171--3197, 2023.

\bibitem[Huang and Zhang(2024)]{huang2024convergence}
J.~Huang and L.~Zhang.
\newblock A convergence framework for {A}iry$_\beta$ line ensemble via pole
  evolution.
\newblock \emph{arXiv preprint arXiv:2411.10586}, 2024.

\bibitem[Jiang and Bai(2021)]{jiang2021generalized}
D.~Jiang and Z.~Bai.
\newblock Generalized four moment theorem and an application to clt for spiked
  eigenvalues of high-dimensional covariance matrices.
\newblock \emph{Bernoulli}, 27\penalty0 (1):\penalty0 274--294, 2021.

\bibitem[Johnstone(2001)]{johnstone2001distribution}
I.~M. Johnstone.
\newblock On the distribution of the largest eigenvalue in principal components
  analysis.
\newblock \emph{Annals of {S}tatistics}, 29\penalty0 (2):\penalty0 295--327,
  2001.

\bibitem[Johnstone(2008)]{Johnstone_Jacobi}
I.~M. Johnstone.
\newblock Multivariate analysis and {J}acobi ensembles: largest eigenvalue,
  {T}racy-{W}idom limits and rates of convergence.
\newblock \emph{Annals of {S}tatistics}, 36\penalty0 (6):\penalty0 2638--2716,
  2008.

\bibitem[Johnstone and Ma(2012)]{johnstone2012fast}
I.~M. Johnstone and Z.~Ma.
\newblock Fast approach to the {Tracy}-{Widom} law at the edge of {GOE} and
  {GUE}.
\newblock \emph{The {A}nnals of {A}pplied {P}robability}, 22\penalty0
  (5):\penalty0 1962, 2012.

\bibitem[Johnstone and Onatski(2020)]{johnstone2020testing}
I.~M. Johnstone and A.~Onatski.
\newblock Testing in high-dimensional spiked models.
\newblock \emph{The Annals of Statistics}, 48\penalty0 (3):\penalty0
  1231--1254, 2020.

\bibitem[Johnstone and Paul(2018)]{johnstone2018pca}
I.~M. Johnstone and D.~Paul.
\newblock {PCA} in high dimensions: An orientation.
\newblock \emph{Proceedings of the IEEE}, 106\penalty0 (8):\penalty0
  1277--1292, 2018.

\bibitem[Johnstone et~al.(2021)Johnstone, Klochkov, Onatski, and
  Pavlyshyn]{johnstone2021spin}
I.~M. Johnstone, Y.~Klochkov, A.~Onatski, and D.~Pavlyshyn.
\newblock Spin glass to paramagnetic transition in spherical
  {S}herrington-{K}irkpatrick model with ferromagnetic interaction.
\newblock \emph{arXiv preprint arXiv:2104.07629}, 2021.

\bibitem[Johnstone et~al.(2025)Johnstone, Klochkov, Onatski, and
  Pavlyshyn]{johnstone2025edge}
I.~M. Johnstone, Y.~Klochkov, A.~Onatski, and D.~Pavlyshyn.
\newblock An edge {CLT} for the log determinant of {W}igner ensembles.
\newblock \emph{Bernoulli}, 31\penalty0 (1):\penalty0 55--80, 2025.

\bibitem[Jones et~al.(1978)Jones, Kosterlitz, and
  Thouless]{jones1978eigenvalue}
R.~Jones, J.~Kosterlitz, and D.~Thouless.
\newblock The eigenvalue spectrum of a large symmetric random matrix with a
  finite mean.
\newblock \emph{Journal of Physics A: Mathematical and General}, 11\penalty0
  (3):\penalty0 L45, 1978.

\bibitem[Ke et~al.(2023)Ke, Ma, and Lin]{ke2023estimation}
Z.~T. Ke, Y.~Ma, and X.~Lin.
\newblock Estimation of the number of spiked eigenvalues in a covariance matrix
  by bulk eigenvalue matching analysis.
\newblock \emph{Journal of the American Statistical Association}, 118\penalty0
  (541):\penalty0 374--392, 2023.

\bibitem[Kim et~al.(2024)Kim, Raponi, and Zaffaroni]{kim2024testing}
S.~Kim, V.~Raponi, and P.~Zaffaroni.
\newblock Testing for {W}eak {F}actors in {A}sset {P}ricing.
\newblock \emph{Preprint}, 2024.

\bibitem[Knowles and Yin(2013)]{knowles2013isotropic}
A.~Knowles and J.~Yin.
\newblock The isotropic semicircle law and deformation of {W}igner matrices.
\newblock \emph{Communications on Pure and Applied Mathematics}, 66\penalty0
  (11):\penalty0 1663--1749, 2013.

\bibitem[Knowles and Yin(2014)]{knowles2014}
A.~Knowles and J.~Yin.
\newblock The outliers of a deformed {W}igner matrix.
\newblock \emph{The Annals of Probability}, 42\penalty0 (5):\penalty0
  1980--2031, 2014.

\bibitem[Kritchman and Nadler(2009)]{kritchman2009non}
S.~Kritchman and B.~Nadler.
\newblock Non-parametric detection of the number of signals: Hypothesis testing
  and random matrix theory.
\newblock \emph{IEEE Transactions on Signal Processing}, 57\penalty0
  (10):\penalty0 3930--3941, 2009.

\bibitem[Lamarre and Shkolnikov(2019)]{lamarre2019edge}
P.~Y.~G. Lamarre and M.~Shkolnikov.
\newblock Edge of spiked beta ensembles, stochastic {A}iry semigroups and
  reflected {B}rownian motions.
\newblock 55\penalty0 (3):\penalty0 1402--1438, 2019.

\bibitem[Lambert and Paquette(2020)]{lambert2020strong}
G.~Lambert and E.~Paquette.
\newblock Strong approximation of {G}aussian $\beta$--ensemble characteristic
  polynomials: the edge regime and the stochastic {A}iry function.
\newblock \emph{arXiv preprint arXiv:2009.05003}, 2020.

\bibitem[Landon and Sosoe(2022)]{landon2022fluctuations}
B.~Landon and P.~Sosoe.
\newblock Fluctuations of the overlap at low temperature in the 2-spin
  spherical {SK} model.
\newblock \emph{Annales de l'Institut Henri Poincare (B) Probabilites et
  statistiques}, 58\penalty0 (3):\penalty0 1426--1459, 2022.

\bibitem[Lee and Yin(2014)]{lee2014necessary}
J.~O. Lee and J.~Yin.
\newblock A necessary and sufficient condition for edge universality of
  {W}igner matrices.
\newblock \emph{Duke Mathematical Journal}, 163\penalty0 (1), 2014.

\bibitem[Lettau and Pelger(2020)]{lettau2020estimating}
M.~Lettau and M.~Pelger.
\newblock Estimating latent asset-pricing factors.
\newblock \emph{Journal of Econometrics}, 218\penalty0 (1):\penalty0 1--31,
  2020.

\bibitem[Ma(2012)]{ma2012accuracy}
Z.~Ma.
\newblock Accuracy of the {Tracy}-{Widom} limits for the extreme eigenvalues in
  white {Wishart} matrices.
\newblock \emph{Bernoulli}, pages 322--359, 2012.

\bibitem[Mikusheva(2007)]{mikusheva2007uniform}
A.~Mikusheva.
\newblock Uniform inference in autoregressive models.
\newblock \emph{Econometrica}, 75\penalty0 (5):\penalty0 1411--1452, 2007.

\bibitem[Mo(2012)]{mo2012rank}
M.~Mo.
\newblock Rank 1 real {W}ishart spiked model.
\newblock \emph{Communications on Pure and Applied Mathematics}, 65\penalty0
  (11):\penalty0 1528--1638, 2012.

\bibitem[Muirhead(2009)]{muirhead2009aspects}
R.~J. Muirhead.
\newblock \emph{Aspects of multivariate statistical theory}.
\newblock John Wiley \& Sons, 2009.

\bibitem[Onatski(2010)]{onatski2010determining}
A.~Onatski.
\newblock Determining the number of factors from empirical distribution of
  eigenvalues.
\newblock \emph{The Review of Economics and Statistics}, 92\penalty0
  (4):\penalty0 1004--1016, 2010.

\bibitem[Onatski(2012)]{onatski2012asymptotics}
A.~Onatski.
\newblock Asymptotics of the principal components estimator of large factor
  models with weakly influential factors.
\newblock \emph{Journal of Econometrics}, 168\penalty0 (2):\penalty0 244--258,
  2012.

\bibitem[Onatski(2018)]{onatski2018asymptotics}
A.~Onatski.
\newblock Asymptotics of the principal components estimator of large factor
  models with weak factors and iid {G}aussian noise.
\newblock \emph{Cambridge Working Papers in Economics, Cambridge-INET Working
  Paper Series}, 1808, 2018.

\bibitem[Onatski et~al.(2013)Onatski, Moreira, and
  Hallin]{onatski2013asymptotic}
A.~Onatski, M.~J. Moreira, and M.~Hallin.
\newblock Asymptotic power of sphericity tests for high-dimensional data.
\newblock \emph{Annals of statistics}, 41\penalty0 (3):\penalty0 1204--1231,
  2013.

\bibitem[Onatski et~al.(2014)Onatski, Moreira, and Hallin]{onatski2014signal}
A.~Onatski, M.~J. Moreira, and M.~Hallin.
\newblock Signal detection in high dimension: The multispiked case.
\newblock \emph{The Annals of Statistics}, 42\penalty0 (1):\penalty0 225--254,
  2014.

\bibitem[O'Rourke(2010)]{o2010gaussian}
S.~O'Rourke.
\newblock Gaussian fluctuations of eigenvalues in {W}igner random matrices.
\newblock \emph{Journal of Statistical Physics}, 138\penalty0 (6):\penalty0
  1045--1067, 2010.

\bibitem[Pastur and Shcherbina(2011)]{pastur2011eigenvalue}
L.~Pastur and M.~Shcherbina.
\newblock \emph{Eigenvalue distribution of large random matrices}.
\newblock Number 171. American Mathematical Soc., 2011.

\bibitem[Paul(2007)]{paul2007asymptotics}
D.~Paul.
\newblock Asymptotics of sample eigenstructure for a large dimensional spiked
  covariance model.
\newblock \emph{Statistica Sinica}, pages 1617--1642, 2007.

\bibitem[Pesaran and Smith(2025)]{pesaran2025identifying}
M.~H. Pesaran and R.~P. Smith.
\newblock Identifying and exploiting alpha in linear asset pricing models with
  strong, semi-strong, and latent factors.
\newblock \emph{Journal of Financial Econometrics}, 23\penalty0 (3):\penalty0
  nbae029, 2025.

\bibitem[Pizzo et~al.(2013)Pizzo, Renfrew, and Soshnikov]{pizzo2013finite}
A.~Pizzo, D.~Renfrew, and A.~Soshnikov.
\newblock On finite rank deformations of wigner matrices.
\newblock In \emph{Annales de l'IHP Probabilit{\'e}s et statistiques},
  volume~49, pages 64--94, 2013.

\bibitem[Ramirez et~al.(2011)Ramirez, Rider, and Vir{\'a}g]{ramirez2011beta}
J.~Ramirez, B.~Rider, and B.~Vir{\'a}g.
\newblock Beta ensembles, stochastic {Airy} spectrum, and a diffusion.
\newblock \emph{Journal of the American Mathematical Society}, 24\penalty0
  (4):\penalty0 919--944, 2011.

\bibitem[Renfrew and Soshnikov(2013)]{renfrew2013finite}
D.~Renfrew and A.~Soshnikov.
\newblock On finite rank deformations of wigner matrices ii: Delocalized
  perturbations.
\newblock \emph{Random Matrices: Theory and Applications}, 2\penalty0
  (01):\penalty0 1250015, 2013.

\bibitem[Shabalin and Nobel(2013)]{shabalin2013reconstruction}
A.~A. Shabalin and A.~B. Nobel.
\newblock Reconstruction of a low-rank matrix in the presence of gaussian
  noise.
\newblock \emph{Journal of Multivariate Analysis}, 118:\penalty0 67--76, 2013.

\bibitem[Sodin(2018)]{sodin2018critical}
S.~Sodin.
\newblock On the critical points of random matrix characteristic polynomials
  and of the riemann $\xi$-function.
\newblock \emph{The Quarterly Journal of Mathematics}, 69\penalty0
  (1):\penalty0 183--210, 2018.

\bibitem[Soshnikov(2002)]{soshnikov2002note}
A.~Soshnikov.
\newblock A note on universality of the distribution of the largest eigenvalues
  in certain sample covariance matrices.
\newblock \emph{Journal of Statistical Physics}, 108:\penalty0 1033--1056,
  2002.

\bibitem[Soshnikov(2000)]{soshnikov2000gaussian}
A.~B. Soshnikov.
\newblock Gaussian fluctuation for the number of particles in {A}iry, {B}essel,
  sine, and other determinantal random point fields.
\newblock \emph{Journal of Statistical Physics}, 100:\penalty0 491--522, 2000.

\bibitem[Stock(1991)]{stock1991confidence}
J.~H. Stock.
\newblock Confidence intervals for the largest autoregressive root in us
  macroeconomic time series.
\newblock \emph{Journal of monetary economics}, 28\penalty0 (3):\penalty0
  435--459, 1991.

\bibitem[Stock and Watson(2002)]{stock2002forecasting}
J.~H. Stock and M.~W. Watson.
\newblock Forecasting using principal components from a large number of
  predictors.
\newblock \emph{Journal of the American statistical association}, 97\penalty0
  (460):\penalty0 1167--1179, 2002.

\bibitem[Tao and Vu(2012)]{tao2012random}
T.~Tao and V.~Vu.
\newblock Random matrices: the universality phenomenon for {W}igner ensembles.
\newblock \emph{arXiv preprint arXiv:1202.0068}, 2012.

\bibitem[Thibeault et~al.(2024)Thibeault, Allard, and
  Desrosiers]{thibeault2024low}
V.~Thibeault, A.~Allard, and P.~Desrosiers.
\newblock The low-rank hypothesis of complex systems.
\newblock \emph{Nature Physics}, 20\penalty0 (2):\penalty0 294--302, 2024.

\bibitem[Thompson(1984)]{thompson1984canonical}
B.~Thompson.
\newblock \emph{Canonical correlation analysis: Uses and interpretation}.
\newblock Number~47. Sage, 1984.

\bibitem[Tracy and Widom(1996)]{tracy1996orthogonal}
C.~A. Tracy and H.~Widom.
\newblock On orthogonal and symplectic matrix ensembles.
\newblock \emph{Communications in Mathematical Physics}, 177:\penalty0
  727--754, 1996.

\bibitem[Udell and Townsend(2019)]{udell2019big}
M.~Udell and A.~Townsend.
\newblock Why are big data matrices approximately low rank?
\newblock \emph{SIAM Journal on Mathematics of Data Science}, 1\penalty0
  (1):\penalty0 144--160, 2019.

\bibitem[Uematsu and Yamagata(2022)]{uematsu2022estimation}
Y.~Uematsu and T.~Yamagata.
\newblock Estimation of sparsity-induced weak factor models.
\newblock \emph{Journal of Business \& Economic Statistics}, 41\penalty0
  (1):\penalty0 213--227, 2022.

\bibitem[Wang and Fan(2017)]{wang2017asymptotics}
W.~Wang and J.~Fan.
\newblock Asymptotics of empirical eigenstructure for high dimensional spiked
  covariance.
\newblock \emph{Annals of statistics}, 45\penalty0 (3):\penalty0 1342, 2017.

\bibitem[Yang(2022{\natexlab{a}})]{FanYang}
F.~Yang.
\newblock Sample canonical correlation coefficients of high-dimensional random
  vectors: Local law and {Tracy-Widom} limit.
\newblock \emph{Random Matrices: Theory and Applications}, 11\penalty0
  (1):\penalty0 2250007, 2022{\natexlab{a}}.

\bibitem[Yang(2022{\natexlab{b}})]{yang2022limiting}
F.~Yang.
\newblock Limiting distribution of the sample canonical correlation
  coefficients of high-dimensional random vectors.
\newblock \emph{Electronic Journal of Probability}, 27:\penalty0 1--71,
  2022{\natexlab{b}}.

\bibitem[Yin et~al.(1988)Yin, Bai, and Krishnaiah]{yin1988limit}
Y.-Q. Yin, Z.-D. Bai, and P.~R. Krishnaiah.
\newblock On the limit of the largest eigenvalue of the large dimensional
  sample covariance matrix.
\newblock \emph{Probability theory and related fields}, 78\penalty0
  (4):\penalty0 509--521, 1988.

\end{thebibliography}

\end{document}